\documentclass[oneside,12pt,reqno]{amsart}
\PassOptionsToPackage{natbib=true,maxbibnames=10}{biblatex}
\usepackage[utf8]{inputenc}
\usepackage{amstext}
\usepackage{amsthm}
\usepackage{esint}
\usepackage{hyperref}
\usepackage{comment}
\usepackage[T1]{fontenc}
\usepackage{tikz}
\usepackage{tikz-3dplot}
\usepackage{xcolor}
\usepackage{array}
\usepackage{tabularray}
\usepackage{enumitem}
\usepackage{amsmath}

\allowdisplaybreaks

\tikzset{
    mynode/.style={
        circle, draw, fill=black, inner sep=1.5pt
    }
}

\usepackage[
backend=biber,
style=authoryear,
uniquename = false]{biblatex}
\renewbibmacro*{in:}{}

\makeatletter
\theoremstyle{plain}
\newtheorem{thm}{\protect\theoremname}
\theoremstyle{definition}
\newtheorem{example}{\protect\examplename}
\theoremstyle{definition}
\newtheorem{defn}{\protect\definitionname}
\theoremstyle{remark}

\theoremstyle{plain}
\newtheorem{prop}{\protect\propositionname}
\theoremstyle{plain}
\newtheorem{cor}{\protect\corollaryname}
\theoremstyle{plain}
\newtheorem{lem}{\protect\lemmaname}

\providecommand{\U}[1]{\protect \rule{.1in}{.1in}}

\usepackage{bbold}

\newcommand{\E}{\mathbb{E}}
\newcommand{\R}{\mathbb{R}}

\newcommand{\vt}[1]{\mathsf{#1}}
\newcommand{\marcin}[1]{}
\newcommand{\colin}[1]{}
\newcommand{\marcintwo}[1]{}
\newcommand{\colintwo}[1]{}

\providecommand{\corollaryname}{Corollary}
\providecommand{\definitionname}{Definition}
\providecommand{\examplename}{Example}
\providecommand{\lemmaname}{Lemma}
\providecommand{\propositionname}{Proposition}
\providecommand{\remarkname}{Remark}

\makeatother

\providecommand{\corollaryname}{Corollary}
\providecommand{\definitionname}{Definition}
\providecommand{\examplename}{Example}
\providecommand{\lemmaname}{Lemma}
\providecommand{\propositionname}{Proposition}
\providecommand{\remarkname}{Remark}
\providecommand{\theoremname}{Theorem}

\usepackage{setspace}
\usepackage[margin=1.25in]{geometry}

\setlist{nosep}

\begin{document}

\title{Nondistortionary belief elicitation}
\author{Marcin Pęski \and Colin Stewart\\ University of Toronto}
\date{\today}

\thanks{This research was supported by the Social Sciences and Humanities Research Council of Canada.}

\begin{abstract}
    A researcher wants to ask a subject about a belief related to a choice the subject made; examples include eliciting confidence or cognitive uncertainty. When can the researcher provide incentives for the subject to report his belief truthfully without distorting his choice? We identify necessary and sufficient conditions for nondistortionary elicitation and fully characterize all incentivizable questions in three canonical classes of problems. For these problems, we show how to elicit beliefs using variants of the Becker-DeGroot-Marschak mechanism.\\

\end{abstract}

\maketitle


\section{Introduction}

\marcintwo{I think list is fine, but if Leeat does not like the list, we can always say "An example is when ... . For another example, ... }
Experimentalists frequently elicit subjects' beliefs about their choices. Examples include the following: (i) The subject chooses an action with a payoff that depends on an unknown state of the world. The researcher asks what probability the subject assigns to his action being correct, i.e, maximizing the ex post payoff. (See, e.g., \citet{coffman2014evidence}.) (ii) The subject provides a guess of some quantity and receives a reward according to how close his guess is to the true quantity. The researcher asks him the likelihood that his guess is within some fixed amount $x$ of the correct value. (See, e.g., \citet{enke2023cognitive}.) (iii) The subject takes a test consisting of a number of multiple choice questions, with a reward for each correct answer. The researcher asks him about the probability that he received a particular score. (See, e.g., \citet{gillen2019experimenting}.)



To ensure subjects' reported beliefs are reliable, researchers typically provide incentives that make truthful reporting uniquely optimal.\footnote{See \citet{healy2024belief} for a discussion of incentivized vs.\ unincentivized belief elicitation.} However, when the belief to be elicited is tied to an action choice, doing so could distort the incentives governing that choice. 
\marcintwo{Instead: "Consider the last example mentioned above ..."}
To take a simple example, suppose the subject must answer a multiple-choice question and then is asked the probability that he gave the correct answer. Suppose moreover that the subject is rewarded at the belief elicitation stage with a payment that is increasing in the probability the subject assigns to the true event (namely, whether his answer was correct or not). Then a subject who is not confident about the correct answer but is confident that one of the answers is incorrect may be able to increase his overall expected payment by choosing the obviously incorrect answer and then reporting a high probability that it is not correct, thereby obtaining a high expected payoff at the belief elicitation stage.\footnote{\marcintwo{This paragraph stands out. Maybe we remove it? Or put it into a footnote? The latter would make sense, as our paper is not about problems outside of experimental context.}
Similar problems arise outside of the experimental context. For example, when a broker (the ``subject'') chooses an investment portfolio on behalf a client (the ``researcher''), the client might ask about the probability that the returns will exceed a certain threshold. Shareholders might ask a CEO about the expected profit from a particular strategy she is choosing to pursue. A builder might ask a contractor about the likelihood of finishing a project by a given date.}

We ask for what questions truthful reports can be incentivized without distorting the incentives in the original decision problem. For those questions that can be incentivized in this way, we construct simple mechanisms to do so.

\marcintwo{Is the role of this paragraph explaining why nondistoritionary elicitation is important? Shouldnt we have a longer argument}
Designing payments for belief elicitation that do not distort the incentives in the original problem allows the researcher to honestly tell the subject that he will maximize his expected payment by choosing the action he believes is optimal in the decision problem and then reporting his belief truthfully.\marcintwo{Maybe the footnote should be a separate paragraph. Regardless, other alternative approaches, like making payoff small should be mentioned as well as the counterargument - the cost is that the researcher needs to be super clear about all the details}\footnote{One alternative approach would be not to inform subjects about the belief elicitation stage until after they have chosen an action. This approach is unlikely to be effective in experiments with repeated choices and more than one instance of belief elicitation. In any case, we would view this approach as misleading if the experimental design relies on the subject believing that he can maximize his earnings by treating the decision problem in isolation.} 
Such instructions should minimize any attempts by the subject to distort his behavior, even in cases where profitable distortions are not obvious.


\marcintwo{We need a paragraph or some sentences that state the research question. "We are interested in two questions: First, how to incentivize belief elicitation without distoriting }

\marcintwo{"In order to answer the above questions, ..."}
We introduce a model combining a general decision problem with a belief elicitation stage. The model allows us to consider a wide variety of belief elicitation questions that can depend on the action $a$ chosen in the decision problem. Formally, a question asks the subject for his subjective expectation of a function $X(a;\theta)$ (according to his belief about the unknown state $\theta$). For example, $X(a;\theta)$ could be an indicator function for the event that $a$ is an optimal action in state $\theta$, which corresponds to asking how likely $a$ is to be optimal ex post. Alternatively, $X(a;\theta)$ could be equal to the utility function $u(a;\theta)$ in the decision problem, which corresponds to asking what the subjective expected utility is in the decision problem.

We say that a question is \emph{incentivizable} if there exists a payment scheme at the belief elicitation stage for which (i) truthfully reporting the expectation of $X(a;\theta)$ is always the unique maximizer of the subject's expected payment, and (ii) the incentives in the decision problem are not distorted, meaning that, for any belief the subject may have about $\theta$, the set of optimal actions in the decision problem is the same as in the combined problem that includes the belief elicitation stage. 

We first identify questions that are incentivizable regardless of the decision problem; we refer to these questions as being aligned with the utility $u(a;\theta)$ in the decision problem. Alignment allows for all questions of the form $X(a;\theta)=u(a;\theta)+d(\theta)$, as well as all invertible affine transformations of such questions with parameters that may depend on the action $a$. 
Examples include asking the subject about the payoff he expects to receive in the decision problem, or asking his willingness to pay to have his action replaced with an ex post optimal one.\footnote{The latter question directly extends a simpler one used by \citet{hu2023confidence} to elicit whether subjects are uncertain about the optimality of their choices.} \colin{Marcin suggests moving the following two sentences to after the discussion of necessary conditions. But that discussion is focused on adjacency graphs, so this doesn't seem to fit naturally unless we do some rewriting of the following paragraphs. (I also think there is some pedagogical value in listing something that does not satisfy the definition.)} On the other hand, a question that asks the subject about the probability that his choice is ex post optimal does not generally take this form (and indeed is, in many problems, not incentivizable). \colin{The next sentence doesn't fit so well here but seems like an important point to make.} A researcher interested in eliciting a measure of cognitive uncertainty \citep{enke2023cognitive} may therefore do better to ask about the subject's expected payoff relative to the optimum rather than the probability that his action is optimal.

For questions that are aligned with the utility, we provide a simple construction of payments satisfying both of our incentivizability criteria. This construction is based on the classic Becker-DeGroot-Marschak method. One can first normalize each question $X(a;\cdot)$ to lie in $[0,1]$, and then elicit the value of $y\in [0,1]$ at which the subject is indifferent between winning a prize with probability $y$ and winning it with probability $X(a;\theta)$. The expected reward from this mechanism is monotone in the expectation of $X(a;\theta)$; alignment with $u(a;\theta)$ ensures that, if the questions are appropriately normalized, it is also monotone with respect to the expected utility in the decision problem. 
A similar construction applies to other questions that we show are incentivizable in some decision problems.

Which other questions are incentivizable, if any, depends on the structure of the decision problem. A particularly important role is played by what we call the ``adjacency graph.'' Two actions are adjacent if there is some belief at which they are both optimal and no other action is. We show that each adjacency places restrictions on how the questions following the adjacent actions are related to one another. Problems with more adjacencies therefore tend to involve stronger restrictions on which questions are incentivizable.

We fully characterize the set of incentivizable questions in three canonical classes of decision problems that differ in the structure of their adjacency graphs: adjacency trees, complete adjacency, and product adjacency. (For complete and product adjacency, we also make some mild richness assumptions regarding linear independence of payoffs.) 

Adjacency trees naturally arise in many problems with ordered states and actions that are monotone in beliefs. For example, if states and actions are discrete real numbers and the subject incurs a quadratic loss based on the distance between his action and the state, then the adjacency graph forms a line.\footnote{A typical example is when the decision problem is itself a belief elicitation problem---such as one of belief updating---with the states representing objective probabilities known to the researcher.} Because there are so few adjacencies, this case is, in a sense, the most permissive in terms of which questions are incentivizable. In particular, alignment with the utility on the full set of actions is no longer necessary: it suffices for the question to be ``piecewise aligned,'' meaning, in this case, that for each pair of adjacent actions it is aligned with the utility (but with the parameters governing the alignment possibly differing across pairs).

Complete adjacency graphs naturally arise in problems in which the subject chooses an action to match an unknown state and receives a payoff based on whether or not he succeeds, as in a multiple-choice problem where the state corresponds to which answer is correct and the subject receives a payment for a correct answer. Relative to adjacency trees, the presence of cycles in the adjacency graph imposes additional restrictions on which questions can be incentivized. In these problems, only questions that are aligned with the utility (in the sense described above) are incentivizable.

Product adjacencies arise when the decision problem comprises a number of separate tasks with complete adjacency graphs and the subject's expected reward is a sum of rewards across these tasks. For example, the subject may complete a multiple-question test and receive a payment proportional to his score. More generally, the researcher may ask the subject about his choices in an experiment with a sequence of tasks, one of which is randomly chosen to be rewarded. The adjacency graph has a special structure: two actions are adjacent if only if they differ on a single task. In this case, a question is incentivizable if and only if it is aligned with some weighted sum of the utilities in the various tasks (but not necessarily aligned with the overall utility in the decision problem). Thus, for example, a question that asks the subject about the likelihood that his score is above some fixed cutoff is not incentivizable, while a question that asks about the expected improvement in his score across two parts of the test is.

\begin{figure}[t]
    \centering
    {\scriptsize \scalebox{1}{
    \renewcommand{\arraystretch}{2} 
    \resizebox{\textwidth}{!}{%
        \begin{tblr}{|Q[t,3cm]|Q[c,h,3.8cm]|Q[c,h,3.4cm]|Q[c,h,3.6cm]|} 
        \hline
        & Tree & Complete graph  & Product structure \\ \hline
 
        \textbf{Adjacency graph} 
        
        & \begin{tikzpicture}
            \draw[white] (0,0) -- (3,0);
            \foreach \i in {1,...,4} {
                \node[mynode] (\i) at (\i-1, 1.5) {};
            }
            \foreach \i/\j in {1/2, 2/3, 3/4} {
                \draw (\i) -- (\j);
            }
        \end{tikzpicture}   
        
        & \begin{tikzpicture}
            \foreach \i in {1,...,6} {
                \node[mynode] (\i) at ({60*\i}:1.5) {};
            }
            \foreach \i in {1,...,6} {
                \foreach \j in {1,...,6} {
                    \ifnum\i<\j
                        \draw (\i) -- (\j);
                    \fi
                }
            }
        \end{tikzpicture} 
        
        & \begin{tikzpicture}
            \foreach \i/\x/\y/\z in {1/0/0/0, 2/2/0/0, 3/2/2/0, 4/0/2/0,
                                     5/0/0/1, 6/2/0/1, 7/2/2/1, 8/0/2/1} {
                \node[mynode] (\i) at (\x+\z, \y+\z) {};
            }
            \foreach \i/\j in {1/2, 2/3, 3/4, 4/1,
                               5/6, 6/7, 7/8, 8/5,
                               1/5, 2/6, 3/7, 4/8} {
                \draw (\i) -- (\j);
            }
        \end{tikzpicture} \\ \hline

        \textbf{Examples} 
        & Cognitive uncertainty (\cite{enke2023cognitive}), monotone one-dimensional problem
        & Multiple choice question, match-the-state problem
        & Random problem selection, 
        multiple-question test \\ \hline

        \textbf{Condition} 
        & Piecewise aligned  
        & Aligned
        & Weighted aligned \\ \hline
        \end{tblr}
    }
    }}
    \caption{Summary of main results. Nodes of the adjacency graph are actions, with edges indicating adjacencies.}
    \label{fig:summary}
\end{figure}

Figure \ref{fig:summary} illustrates the graphs and summarizes the results for all three classes of decision problems.


Two related features distinguish our approach from previous work on belief elicitation. First, the researcher asks the subject only to report a single number.\footnote{We discuss in Section \ref{sec: multiple questions} how to extend our methods to multiple questions.} Second, the quantity of interest to the researcher---as described by the question $X(a;\theta)$---depends nontrivially on the subject's choice of action in the decision problem. In the absence of either of these features, any question $X(a;\theta)$ is incentivizable using standard methods. For example, if the researcher could ask the subject to report his entire belief, it would be enough to incentivize truthful reports and randomly reward the subject either for his choice in the decision problem or for his reported belief. From a practical perspective, however, this approach could be burdensome for subjects if there are more than a few states to report on; if the researcher is only interested in a one-dimensional statistic of the belief, asking about it directly could save time and reduce noise in the reports.



\subsection{Related literature: experiments}


Elicitation of subjects' beliefs about their own performance on a task is common in the literature on self-confidence, overconfidence, and motivated reasoning. In a typical experiment, subjects complete a multiple-choice test, such as an IQ test. Following the test, they are asked about their relative or absolute performance (or both). For instance, in \citet{Zimmerman20} and \citet{benoit2015does}, subjects are asked about the likelihood of being above the median; in \citet{gillen2019experimenting}, subjects are asked how many questions they answered correctly and where they lie in the distribution and, in a different task, to report confidence in their decision on a qualitative scale; in \citet{burks2013overconfidence}, subjects are asked what quintile their performance lies in.\footnote{Other examples include \citet{butler2007imprecision}, \citet{blavatskyy_betting_2009}, \citet{clark2009overconfidence}, \citet{EilRao11}, \citet{ERTAC2011532}, \citet{ortoleva2015overconfidence}, \citet{Serra-GarciaGneezy21}, \citet*{mobius2022managingself}, and \citet{abdellaoui2024unpacking}, among many others. For elicitation of confidence in the psychological literature, see, for example, \citet{liberman2004local}, \citet{moore2008trouble}, and \citet{hoffrage2016overconfidence}.} 
In most of these experiments, the belief elicitation is incentivized, either through a quadratic scoring rule \citep[e.g.,][]{Zimmerman20,clark2009overconfidence}, a BDM mechanism \citep[e.g.,][]{benoit2015does}, or a payment for the correct answer \citep[e.g.,][]{burks2013overconfidence,Serra-GarciaGneezy21}. These incentives are typically distortionary, and our results suggest that nondistortionary elicitation is, in most cases, not possible for the confidence questions that have been asked, but would be possible for alternative confidence questions based directly on expected payoffs.

Elicitation questions of this form are particularly common in the experimental literature on gender and discrimination \citep[e.g.,][]{hoff2006discrimination,coffman2014evidence,bordalo2019beliefs,ExleyNielsen24}. In the seminal paper of \citet{niederle2007women}, subjects complete an arithmetic task and guess their own rank in a group of four, receiving a bonus if their guess is correct. In the competition treatment, they are paid for their performance in the main task only if they have the highest score in their group. The belief elicitation bonus creates distortionary incentives: subjects who believe they are unlikely to be the top in the group may be able to increase their expected payment by performing badly in the main task and guessing that they have the lowest rank, thereby increasing their likelihood of getting the belief elicitation bonus.

The problem we study is motivated in part by the recent literature on cognitive uncertainty initiated by \citet{enke2023cognitive} \citep*[see also][]{amelio2022cognitive,arts2024measuring,xiang2021confidence,de2024caution}. 
In each of these papers, subjects' cognitive uncertainty is elicited using unincentivized questions. 
\citet{hu2023confidence} is the first paper we are aware of that provides strict nondistortionary incentives for subjects to reveal whether they are uncertain about their decision in a complex choice task. His mechanism is essentially a simplified version of the Becker-DeGroot-Marschak mechanism we employ; in his, subjects make a binary choice of whether to pay a cost to have some chance that their action can be replaced with the optimal one.

Several papers explicitly address the possibility of distortions due to belief elicitation incentives in particular experiments. \citet*{mobius2022managingself} describe how their belief elicitation mechanism is designed with the intention of preserving incentives in their main task. \citet{clark2009overconfidence} note that the quadratic scoring rule they employ for belief elicitation introduces distortionary incentives, which they try to mitigate by making the incentives in the belief elicitation problem relatively weak. \citet{coffman2014evidence} elicits, for each question in a multiple-choice test, subjects' beliefs in the likelihood that their answer is correct; the BDM mechanism she uses is nondistortionary for this task.

In game theory experiments, subjects are sometimes asked about their beliefs regarding other players' actions. In an extensive-form game, treating the other players' strategies as an unknown state, such a question may depend on the subject's action.  
\citet{NyarkoSchotter02} elicit beliefs about the action of the other player in simultaneous-move 2x2 games. They note the potential for a different kind of distortion from the one we study: subjects may want to coordinate on predictable action profiles to increase their payments from the belief elicitation mechanism. 
\citet{AoyagiFréchetteYuksel24} elicit beliefs about the stage-game actions of the other player in a repeated prisoners' dilemma; as in our model, this question depends on the subject's actions insofar as the other player's action is a function of the history of play. 
\citet{rutstrom_stated_2009} compare incentivized and unincentivized belief elicitation in a repeated game. They find that belief elicitation incentives affect the play of the game, which may be due more to cognitive factors than to distortion of incentives. Some experiments with sequential-move games elicit beliefs about responses to all possible actions, not only the action actually chosen by the subject, eliminating any possibility of distortions \citep[e.g.,][]{bellemare2008measuring,trautmann2015belief,FehrPowellWilkening21}. \citet{halevy2025magic} elicit subjects' certainty equivalents for games given the actions they chose.

Belief elicitation is also common in field experiments. In some cases, researchers elicit the subject's willingness-to-pay (WTP) for a product or intervention \citep[e.g.,][]{AshrafBerryShapiro10,Dupas14,DizonRoss19,BelzilMaurelSidibe21}. Since the WTP generally depends on the subject's behavior in the main task, this question is action-dependent as in our model. Although some of the observability assumptions of our model may be violated (including observability of actions, states, and payoffs), our results nonetheless imply that eliciting WTP is incentivizable with a BDM mechanism.

The binarized scoring rule introduced by \citet{smith1961consistency}---and popularized by \citet{hossain2013binarized}---has become a common choice for eliciting beliefs in experiments. \citet{danz2022belief} find that subjects report more accurate beliefs when they are told that reporting truthfully will maximize the payment they can expect to receive than when the payments in the binarized scoring rule are described explicitly. In keeping with this finding, we would expect to see less distortion in behavior in belief elicitation settings like ours if subjects are instructed that the incentives are designed to ensure they can maximize their expected payment by treating each task in isolation and choosing what they believe to be optimal; such instructions can be honestly provided only if the belief elicitation question is incentivizable.

\subsection{Related literature: theory}

Belief elicitation has been widely studied and used in both theory and experiments (see \citet{schlag2015penny}, \citet{charness2021experimental}, \citet{haaland2023designing}, and \citet{healy2024belief} for surveys). We are not the first to observe that incentivized belief elicitation can distort other decisions. \citet{chambers2018dynamic} and \citet{healy2024belief} discuss the possibility that a subject would purposefully fail a test to increase the payment from belief elicitation about the likelihood of passing. \citet*{blanco2010belief} (see also \citet{SchotterTrevino2014}) find evidence that, in some problems, subjects who are paid for both a choice in a game and a reported belief take advantage of hedging opportunities, distorting either choice (or both). We implicitly assume that subjects are randomly paid either for the main task or the reported belief, eliminating such risk-hedging opportunities (while retaining other interactions between the action choice and the belief report).

\citet{ChassangPadróIMiquelSnowberg12} develop a theory of elicitation of beliefs about returns to an investment that depends on unobservable effort. In their model, correct elicitation of beliefs requires that the effort choice is not distorted. They show that a willingness-to-pay method can be used to identify returns without distorting incentives in the effort choice. 

The closest prior theoretical work to this paper is that of \citet{lambert2008eliciting} and \citet{lambert2011elicitation}, which ask which properties of distributions can be elicited. Our model shares the feature that the belief elicitation question does not ask the subject to report his entire belief. We sidestep their question of elicitability by restricting attention to questions that always correspond to elicitable properties, and we add the condition that the elicitation must not distort the decisions in the main decision problem.\footnote{\citet{Frongillo15} (and references therein) analyze elicitation of vector-valued properties of beliefs. One could view the action-report pair in our model as a special case of a multi-dimensional property of the belief (and nondistortionary elicitation as elicitation of that property), with the caveat that, unlike \citet{Frongillo15}, we must allow for the property to take multiple values at some beliefs. We are grateful to a referee for suggesting this connection.} In subsequent work, \citet{chen2026ask} characterize a stronger form of nondistortionary elicitation when the researcher can ask multiple questions.

\citet{azrieli2018incentives} study incentives in a sequence of tasks and find that paying for a randomly selected problem is the only incentive-compatible mechanism when allowing for a general class of preferences. In their model, the sequence of tasks is exogenously given, whereas in ours the belief-elicitation stage is itself an endogenous task since it depends on the subject's choice (or choices) in the main task. Random selection for payment is therefore not sufficient to ensure incentive compatibility.

\section{Model}

A subject (he) chooses an action and then faces a belief elicitation problem posed by a researcher (she) that may depend on the action he chose.

A \emph{decision problem} consists of a tuple $\left(\Theta,A,u\right)$, where $\Theta$ is a finite set of states of the world, $A$ is a finite set of actions, and $u:A\longrightarrow \mathbb{R}^\Theta$ is a utility function specifying, for each action, the vector of payoffs across all states. We write $u(a;\theta)$ for the $\theta$-coordinate of the vector $\vt{u}(a)$. (Note that we use a distinct font for vectors to distinguish, for example, between the vector $\vt{u}(a)$ and the function $u$.) 
For each belief $p\in \Delta (\Theta)$, let $A(p)=\arg \max_{a\in A} \sum_{\theta} p(\theta) u(a;\theta) $ denote the set of optimal actions at $p$; we refer to actions in $A(p)$ as $u$-optimal. For simplicity, we assume that (i) there are no redundant actions, i.e., no $a$ and $a'$ such that $\vt{u}(a)=\vt{u}(a')$, and (ii) every action is strictly rationalizable, i.e., for each $a$, there exists some $p$ such that $A(p)=\{a\}$.

After choosing an action $a$, the subject faces a \emph{question} $\vt{X}(a)\in \mathbb{R}^\Theta$ about his belief chosen according to a \emph{question profile} $X:A\longrightarrow \mathbb{R}^\Theta$, with $\theta$-coordinate $X(a;\theta)$. We interpret the question $\vt{X}(a)$ as asking the subject to report his subjective expected value $\E_{p}X\left(a;\cdot\right)$ given the action $a$ that he chose in the first stage and his belief $p$. The dependence of $X$ on $a$ allows for the possibility that the researcher seeks information about the subject's belief that is related to the chosen action. 


The following examples illustrate how this formulation allows for considerable flexibility in what the subject is asked to report:
\begin{enumerate}
    \item The question ``how likely is it that the chosen action $a$ is optimal ex post?'' corresponds
to 
\begin{align*}
X\left(a;\theta\right) & =\mathbb{1}\{a\in\arg\max_{b\in A}u\left(b;\theta\right)\}.
\end{align*}
Then $\E_{p}X\left(a;\cdot\right)$ is the subjective probability of the chosen action being $u$-optimal ex post. 





\item The question ``what is the expected regret from the chosen action?'' 
corresponds to 
\begin{align}\label{ex: En Hua}
X\left(a;\theta\right) & =\max_{b\in A}u\left(b;\theta\right) - u\left(a;\theta\right).
\end{align}

\item The question ``how likely is it that your guess $a$ is within $x$ of the true value $\theta$?" corresponds to
\begin{equation}\label{ex: within x}
X(a;\theta) = \mathbb{1} \{ |a-\theta|\leq x \}.
\end{equation}
\end{enumerate}

The subject announces a report $r\in\mathbb{R}$ and is given a reward that may depend on his report, his action $a$, and the realized state $\theta$. His overall payoff---including the payoff from the decision problem in the first stage---is given by a bounded function $V:\mathbb{R}\times A\times\Theta\longrightarrow\left[0,1\right]$ (which we normalize to the unit interval for convenience), the first argument of which is the subject's report $r$. We refer to $V$ as an \emph{elicitation method}.

For simplicity, we do not include the payoff $u(a;\theta)$ explicitly in the elicitation method; adding it would not change anything as the payoff from the elicitation problem can be adjusted accordingly so as to give the same overall payoff. In the applications we have in mind, the researcher pays the subject for the decision problem with some fixed probability $\alpha\in (0,1)$ and for the belief elicitation problem with the remaining probability $1-\alpha$. The elicitation method $V$ therefore takes the form $\alpha u (a;\theta) + (1-\alpha) V_0 (r,a,\theta)$ for some $V_0$, where $V_0$ is the payoff from the belief elicitation mechanism. As is standard in the recent literature on belief elicitation, we implicitly view $V_0$ as the probability of winning a fixed prize to avoid any influence of risk preferences on the reported belief.



\begin{defn}
A question profile $X$ is \emph{incentivizable} if there exists an elicitation method $V$ such that, for every $p\in \Delta(\Theta)$,
\begin{equation}\label{eq:incentivizability}
\arg\max_{(r,a)}\E_{p}V\left(r,a,\cdot\right)=\left\{ \left(\E_{p}X\left(a;\cdot\right),a\right):a\in A(p) \right\} . 
\end{equation}
Any $V$ satisfying this condition \emph{incentivizes} $X$.
\end{defn}

\marcintwo{Somewhere, we need to explain that our basic assumptions for the model are that the experimenter observes actions, states and payoffs in the experiment. The probability of the state, as well as beliefs, are not affected by the choice of action. However, the expected value of the question, or more generally, its distribution is, typically, affected. We can either do it here or, perhaps more naturally, at the beginning of the Overview for practitioners section }

If, given $p$, $(r,a)$ maximizes $\E_{p}V\left(r,a,\cdot\right)$, we say that $a$ is $V$-optimal.\footnote{We refer to reports $r$ or pairs $(a,r)$ that maximize $\E_{p}V\left(r,a,\cdot\right)$ simply as ``optimal'' since there is no risk of confusion with $u$-optimality.}

Incentivizability combines two requirements of the elicitation method. First, the payoffs at the belief elicitation stage must not distort the subject's action choices in the decision problem in the sense that the set of actions $a$ he optimally chooses in the overall problem with payoffs $V$ is the same as in the original decision problem with payoffs $u$; in other words, $u$-optimality and $V$-optimality are equivalent.\footnote{A researcher might be content with the weaker condition that the left-hand side of \eqref{eq:incentivizability} is a subset of the right-hand side, ensuring that any $V$-optimal action is also $u$-optimal. Since each action is $u$-optimal on a closed set of beliefs and multiplicity occurs only on the boundaries of these sets, and similarly for $V$-optimality, this requirement is equivalent to incentivizability as we define it.} Second, the subject must have strict incentives to report his true subjective expectation of $\E_p X(a;\cdot)$ given his action choice $a$. In an experiment, for any question profile that cannot be incentivized, no matter how the researcher designs the incentives at the belief elicitation stage, she cannot honestly tell subjects that they will maximize their earnings by considering the decision problem in isolation and by reporting beliefs truthfully.

Our formulation imposes two substantive restrictions on the belief elicitation problem. First, the subject is asked to report only a one-dimensional statistic of the belief rather than, say, a full probability distribution. In practice, collecting more complicated information about beliefs quickly becomes impractical beyond a small number of states. If, however, the full probability distribution could be elicited, our problem would reduce to a standard belief elicitation problem since there would be no need to make the question or incentives dependent on the action chosen in the decision problem. Second, the elicited belief is based on the expectation of some question $X(a;\cdot)$. While this formulation captures many relevant cases, in principle, the researcher may want to elicit other properties of the distribution for which our approach may not apply. We discuss in Section \ref{sec: discussion} what changes if either of these restrictions is relaxed.

The model makes two important implicit assumptions. First, the action, state, and utility function are all observed by the researcher. This assumption typically holds in lab experiments. Second, the action does not directly affect the state or the subject's belief about the state. We discuss both of these assumptions in Section \ref{sec: discussion}.

\section{Overview for practitioners}

\marcintwo{Some comments: 1. Can we split this section into subsections: Single experiment and Parallel (or multiple) experiments? It's three pages and it is a bit long and it contains too much distinct material for a single complex.
}

\marcintwo{2. I would start the first subsection with a description of a Basic Setup: The experimenter observes actions, states of the world, and controls (designs, observes) payoffs. The decisions about actions and belief report are made at the same time. The states are not affected by actions.}

\marcintwo{3. The first section should contain an example of "expected payoff" question as an incentizable question, perhaps even before "expected regret". This would be a good point to mention that BDM incentivization of "expected payoffs" looks mathematically like an elicitation of the willingess-to-pay. For example, the BDM incentivization of the "expected payoffs" question is: tell me the "expected payoff $r$. Then, we will draw a number $z$. If $z>r$, we give you $z$. Otherwise, you receive the realized payoff from the experiment. The latter is equivalent to being allowed to participate in the experiment.}

\marcintwo{4. I will add at the end a final subsection on extensions/observability/limits.}

As an illustration of our results, consider an experiment in which the subject makes a choice under uncertainty. The experimenter then wishes to ask the subject about his confidence in his decision. To do so, she could ask the subject how likely he thinks it is that his decision is optimal ex post, or in other words, how likely it is that if he learns the true state of the world he will not want to change his action; call this the ``ex post optimality question.'' Alternatively, she could ask how much he would be willing to pay to have his action replaced with the ex post optimal one; call this the ``expected regret question.'' The difference between these questions is that, in the latter, likelihoods of states in which the action is suboptimal must be weighted according to the loss in payoff that occurs in those states. (Other questions could be used to measure confidence, but we will focus on these two for simplicity.)

Whichever confidence question the experimenter uses, to elicit reliable answers, she would like to provide incentives that make it uniquely optimal for the subject to report the answer truthfully. In addition, she does not want these incentives to contaminate the original decision problem. For example, it should not be that the subject can increase his total expected payment by making a choice that he knows is suboptimal ex post in the main problem but leads to a larger reward at the belief elicitation stage when he (correctly) reports low confidence in his choice. The question is incentivizable if there exists a payment scheme that satisfies both of these criteria.

It turns out that the ex post optimality question is only incentivizable if the original decision problem has a specific structure. For example, suppose the decision problem consists of a multiple-choice question, with the subject receiving a fixed reward for a correct answer. Formally, the set of actions $A$, which describe possible answers, is equal to the set of states $\Theta$, which describe the correct answer; the subject's answer is correct if $a=\theta$. The payoff in the decision problem is
\begin{equation}\label{eq:state_matching_payoff}
u(a,\theta) =
\begin{cases}
R &\text{if $a=\theta$},\\
0 &\text{otherwise,}
\end{cases}
\end{equation}
where $R>0$. In this case, the ex post optimality question amounts to asking the subject how likely he thinks it is that the answer he gave is correct. Notice that the probability that the answer $a$ is correct is equal to the expectation of $u(a,\theta)/R$. Thus the ex post optimality question is described by $X(a;\theta)=u(a,\theta)/R$.

Proposition \ref{Observation 0} below implies that this question is incentivizable in this decision problem. Moreover, it can be incentivized using a standard BDM mechanism to elicit the probability equivalent of the event $\theta=a$ for the action $a$ chosen by the subject.\footnote{In this mechanism, the subject is asked to report the probability $r$ that he assigns to the state $\theta=a$. A random number $z\sim U[0,1]$ is drawn. If $r\geq z$, the subject wins a fixed reward $R'>0$ if indeed $\theta=a$; if $r<z$, the subject wins the reward $R'$ with probability $z$.}

As is well known, the BDM mechanism makes it optimal for the subject to report his true belief $q$. To check whether it distorts incentives in the original decision problem, consider the payoff the subject receives from the mechanism. Given that he reports $r=q$, his subjective expected payoff from the mechanism is
\begin{equation}\label{eq:BDM_value}
R'\left(\int_0^q q dz + \int_q^1 zdz\right) = \frac{R'}{2} \left( 1 + q^2\right).
\end{equation}
Notice that this expression is increasing in $q$. Thus the subject maximizes his expected payment in the mechanism by choosing the action $a$ that maximizes his subjective probability that $\theta = a$---which is exactly the action that maximizes his expected payoff in the original decision problem. Therefore, the optimal action in the problem that combines the original decision problem and the belief elicitation mechanism is the same as when facing the original problem in isolation, as needed for incentivizability.

This argument relies on two key features of the mechanism. First, the subject's expected payoff from the mechanism is increasing in $q$. More generally, this value should be monotone in the subjective expectation of $X$. The use of the BDM mechanism ensures that this property holds for each $a$.\footnote{This property does not necessarily hold for other mechanisms that incentivize truthful reports. For example, the quadratic binarized scoring rule is not monotone.} Moreover, the expected payoff from the mechanism depends only on $q$, not on the choice of action $a$. More generally, the use of BDM ensures that this property holds after some suitable renormalization of each $X(a;\cdot)$ by a non-zero affine transformation.

Second, the original decision problem is such that the optimal action is the one that maximizes $q$; this feature results from $X$ being proportional to $u$. More generally, following a suitable renormalization, the optimal action $a$ in the original decision problem should be the one that maximizes the expectation of $X(a;\cdot)$. This property holds under a condition that we call ``alignment'' that extends beyond proportionality of $X$ and $u$ to a larger class of transformations.

For the ex post optimality question, the latter condition fails in many decision problems. Examples \ref{ex: second order beliefs} and \ref{ex: complete adjacency thm} describe two decision problems in which the utility function is more complicated than the one above and is not aligned with $X$. In these examples, our necessary conditions imply that the ex post optimality question cannot be incentivized.

On the other hand, the expected regret question is \emph{always} aligned with the utility: $-X$ differs from $u$ by a constant that depends only on $\theta$. This question is therefore incentivizable in every decision problem, and can be incentivized using a BDM mechanism. In this case, the first step is to take an affine transformation of $X$ to turn it into a question $\hat X$ that takes values in $[0,1]$ and that is increasing as $u$ increases. A random number $z\sim U[0,1]$ is then compared to the subject's report $r$; if $r<z$, the subject wins a reward $R'>0$ with probability $z$, and if $r\geq z$, the subject wins the reward $R'$ with probability $\hat X (a,\theta)$.

Now suppose the decision problem is a test comprising $I$ multiple choice questions, with the subject receiving a reward for each question he answers correctly. Formally, actions take the form $a=(a_1,\dots,a_I)\in A^I$ and states the form $\theta=(\theta_1,\dots,\theta_I)\in \Theta^I$, with $A=\Theta$. The payoff in the decision problem is $U(a,\theta)=\frac{1}{I}\sum_{i=1}^I u(a_i,\theta_i)$, where $u$ is defined by \eqref{eq:state_matching_payoff}.

Following the test, the experimenter would like to learn some statistic of the subject's belief regarding his performance. For example, the experimenter may want to elicit the subject's subjective expected score on the test. This corresponds to the question $X(a;\theta)=U(a,\theta)/R$, which is incentivizable using a BDM mechanism in essentially the same way as in the case of the expected regret question above.

On the other hand, our results imply that the experimenter cannot incentivize the subject to truthfully report the probability he assigns to achieving a score above a fixed threshold $\zeta\in \{0,\dots,I-1\}$ without distorting incentives in the decision problem. For example, suppose the experimenter uses a BDM mechanism to elicit this belief. Because the value associated with this mechanism is increasing in the belief, the subject maximizes his payoff at the belief elicitation stage by maximizing the probability that his score exceeds $\zeta$. This objective can be in conflict with maximizing the probability of answering each question correctly if the subject does not believe $\theta_i$ is independent across $i$. For a particularly simple (if unrealistic) example, suppose the test asks the same question four times and there are four possible answers (that is, $I=|A|=4)$, and take $\zeta=0$. The subject can guarantee a score of $1$ by providing a different answer to each question; since $1>\zeta$, this strategy combined with a reported belief that his score exceeds $\zeta$ with probability $1$ gives the largest possible payoff at the belief elicitation stage. The strategy that maximizes the subject's expected payoff in the decision problem alone is to provide the same answer---the one that is most likely correct---to all four multiple choice questions. If the subject is not sufficiently confident about the correct answer, the former strategy yields a higher overall expected payoff than the latter does.

Some more complicated belief elicitation questions \emph{are} incentivizable. Suppose the experimenter wants to elicit the subject's belief about his expected improvement across two parts of the test. For example, take $I$ to be even and let 
\[
X(a;\theta) = \sum_{i=I/2+1}^I \mathbb{1}(a_i=\theta_i) - \sum_{i=1}^{I/2} \mathbb{1}(a_i=\theta_i).
\]
This question corresponds to asking the subject his expectation of the increase in his raw score from the first half of the test to the second. The previous logic does not apply directly here because increasing $u(a_i,\theta_i)$ increases $X$ for $i>I/2$ but decreases it for $i\leq I/2$.

And yet this question is incentivizable using a BDM mechanism, with the caveat that the incentives in the belief elicitation task cannot be too strong relative to those in the decision problem. To implement this mechanism, we first rescale $X$ according to $\hat X(a;\theta) = 1/2 + X(a;\theta)/I$; doing so ensures that it takes values in $[0,1]$. By \eqref{eq:BDM_value}, reporting the expectation of $X$ truthfully gives an expected payoff of $R'(1+\E_p [\hat X(a,\cdot)]^2)/2$.

Now suppose the subject is rewarded for the decision problem with probability $\alpha$ and for the belief elicitation problem with the remaining probability $1-\alpha$, where $\alpha\in (0,1)$. Notice that it is possible for the subject to increase his expected payment from the belief elicitation mechanism by choosing an action profile that is suboptimal in the decision problem, for example by intentionally reducing his expected score in the first half of the test to increase his expected improvement. However, if he chooses an action profile that lowers his expected payoff in the decision problem by some $\Delta>0$ relative to the optimal profile $a^*$, doing so will change the value of $\hat X$ by at most $\Delta$, and therefore increases the expected payoff at the belief elicitation stage by at most a fixed multiple of $\Delta$. If the incentives in the decision problem are strong enough relative to those at the belief elicitation stage, it is therefore optimal for the subject to choose $a^*$. In fact, one can show that as long as $\alpha R\geq (1-\alpha)R'$, there is no distortion in incentives in the decision problem.

The same logic applies to all questions $X$ that are proportional to linear combinations of the payoffs $u(a_i,\theta_i)$, and to much more general decision problems in which the payoff is the sum of payoffs across $i$, with the action sets, state spaces, and payoffs possibly varying with $i$.

\section{Sufficient conditions}\label{sec: sufficient conditions}

We begin by identifying simple sufficient conditions under which a question profile $X$ is incentivizable. In Section \ref{sec: complete}, we show that these conditions are also necessary in some natural applications.

\begin{defn}
Questions $\vt{X},\vt{Y}\in\mathbb{R}^\Theta$ are \emph{equivalent} if there exist $\gamma\in\R\setminus\{0\}$ and $\kappa\in\R$ such that
\[
\vt{X}=\gamma\vt{Y}+\kappa \vt{1},
\]
where $\vt{1}\in \R^\Theta$ is the vector of all ones. We say that question profiles $X$ and $Y$ are equivalent if $\vt{X}(a)$ and $\vt{Y}(a)$ are equivalent for all $a\in A$.
\end{defn}

Given two equivalent questions, for any given belief, knowing the expectation of one of them is sufficient to calculate the expectation of the other (indeed, one can show that this property defines equivalence). Switching from one question to the other therefore amounts to shifting and rescaling the subject's reports, and does not affect the information the researcher would obtain from truthful responses.

\begin{defn}\label{def:aligned}
A question profile $X$ is \emph{aligned with $u$ on $B\subseteq A$} if there exists $\vt{d}\in\R^\Theta$ such that either $\vt{X}(a)$ is equivalent to the question $\vt{u}\left(a\right)+\vt{d}$ for all $a\in B$, or $\vt{X}(a)$ is equivalent to the question $\vt{d}$ for all $a\in B$. We say that $X$ is \emph{nontrivially aligned with $u$ on $B$} in the former case, and \emph{trivially} so in the latter case. If $B=A$, we say simply that $X$ is \emph{aligned with $u$} (and similarly with the (non)-trivial qualifier).
\end{defn}


Relative to the question profile $X=u$ that asks the subject about the expected utility from his chosen action, questions aligned with $u$ allow for three changes. First, a vector $\vt{d}$ may be added to payoffs. Since $\vt{d}$ is independent of $a$, this change has no effect on the $u$-optimal action. Then, for each $a$, the question $\vt{X}(a)$ can be rescaled by a (nonzero) constant $\gamma(a)$ and translated by another constant $\kappa(a)$ uniformly across $\theta$. These changes make each question $\vt{X}(a)$ equivalent to the question $\vt{u}(a)+\vt{d}$. The case of trivial alignment can be viewed as a limit of these operations as $\vt{d}$ is scaled up and $\gamma(a)$ scaled down by the same constant, causing the $\vt{u}(a)$ term to vanish.

\marcin{Can we add that the researcher asking question $u+d$ can always compute any other question? So, in some sense, $d$ is the only relevant degree of freedom. 
Maybe we can have a notion of equivalent questions (like we have a notion of equivalent decision problems) and then the definition of aligned would be: $X$ is equivalent to an action-free question or to a question $u+d$ for some vector $d$.}
\colin{To do later: add definition of equivalent questions. Add statements after (in?) theorems of the form in Marcin's comment.}

Note that for $X$ aligned with $u$, the parameters $\gamma$, $\kappa$, and $\vt{d}$ are not uniquely determined in general.


\begin{prop}
\label{Observation 0}
If $X$ is aligned with $u$, then it is incentivizable.
\end{prop}

Proofs omitted from the main text may be found in the appendix.

The proposition indicates that alignment with $u$ is sufficient for incentivizability. The proof proceeds by construction using a standard Becker-DeGroot-Marschak (BDM) mechanism.\footnote{While all of our constructions employ variants of BDM, these are presumably not the only nondistortionary elicitation methods that incentivize truthful reporting of beliefs for incentivizable question profiles. Identifying the set of such elicitation methods is beyond the scope of this paper.} The idea is to first replace each question with an equivalent one to make it of the form $\vt{X}\left(a\right)=\vt{u}\left(a\right)+\vt{d}$ for all $a$, or $\vt{X}(a)=\vt{d}$ for all $a$. Let $[L,M]$ be an interval containing every value of $X(a;\theta)$. After learning the subject's report, $r$, of his expectation of $\vt{X}(a)$, the researcher draws a number $x$ uniformly from $[L,M]$. The subject receives a fixed prize with probability $(X(a;\theta)-L)/(M-L)$ if $r>x$, and with probability $(x-L)/(M-L)$ otherwise. By standard arguments, this mechanism provides strict incentives for the subject to truthfully report his expectation of $\vt{X}(a)$. Moreover, given that the subject reports truthfully, his expected value from the mechanism is an increasing function of his expectation of $\vt{X}(a)$. If $\vt{X}\left(a\right)=\vt{u}\left(a\right)+\vt{d}$, it follows that, to maximize the expected payoff from the mechanism, the subject should choose an action $a$ that maximizes the expectation of $\vt{u}(a)$, as needed for incentivizability. If $\vt{X}(a)=\vt{d}$, the expected payoff from the mechanism does not depend on the action choice; an overall payoff $V$ equal to $u$ plus the payoff obtained from the BDM mechanism ensures incentivizability.

For a researcher interested in eliciting a measure of the subject's confidence in his choice of action, an immediate implication of the proposition, captured in the following corollary, is that it is possible to elicit the subject's expected regret without distorting his decisions.

\begin{cor}\label{cor: En Hua}
For any decision problem, the question about regret in equation \eqref{ex: En Hua} 
is incentivizable, as is any question profile that does not depend on the chosen action. 
\end{cor}


\section{Necessary conditions}\label{sec: necessary conditions}

We now identify simple necessary conditions for question profiles to be incentivizable. Since such questions must not distort incentives for the action choice in the decision problem, it is natural to focus on beliefs where the subject is indifferent between two actions. Any distortion that arises at a point of indifference implies that there are also distortions in an open set around that point.

Say that two actions $a,b\in A$ are \emph{adjacent} if there is a belief $p\in\Delta (\Theta)$ such that $A(p)=\{a,b\}$, that is, at belief $p$, $a$ and $b$ are both $u$-optimal and there is no other $u$-optimal action. 

Our necessary conditions make use of an observation about the structure of the \emph{value of information}
\[V^*(p)=\max_{a,r}\E_pV(a,r,\cdot);\]
$V^*(p)$ is the payoff associated with any optimal pair $(a,r)$ at belief $p$. Note that $V^*(p)$ is the upper envelope of the affine functions $\E_pV(a,r,\cdot)$ across $(a,r)$, and is therefore convex. Moreover, given beliefs $p\neq p'$, if there exists some $(a,r)$ that is optimal at both $p$ and $p'$, then $V^*$ is affine along the line segment $\overline{pp'}$ connecting $p$ to $p'$. On the other hand, if for some $\alpha\in (0,1)$ there is a pair $(a,r)$ that is optimal at $q=\alpha p +(1-\alpha) p'$ but is not optimal at one of $p$ or $p'$, then the value of information is not affine along $\overline{pp'}$.\footnote{To see why, first suppose $(a,r)$ is optimal at $p$ and $p'$. Since $V^*$ is convex, for any convex combination $q=\alpha p+(1-\alpha)p'$, we have $\E_qV(a,r,\cdot)\leq V^*(q)\leq \alpha V^*(p)+(1-\alpha)V^*(p')=\alpha \E_pV(a,r,\cdot)+(1-\alpha)\E_{p'}V(a,r,\cdot)=\E_qV(a,r,\cdot)$, which implies that all inequalities hold with equality, as needed. For the converse, suppose $(a,r)$ is optimal at $q$ but not at $p$. Then $\alpha V^*(p)+(1-\alpha) V^*(p')>\alpha \E_p V(a,r,\cdot)+(1-\alpha) \E_{p'}V(a,r,\cdot)=\E_qV(a,r,\cdot)=V^*(q)$.}

Consider an elicitation method $V$ that incentivizes a question profile $X$. Let $p$ be a belief with $u$-optimal actions $A(p)=\{a,b\}$. Notice that the values associated with actions $a$ and $b$ must match at belief $p$: $\E_p V(a,r,\cdot)=\E_p V(b,s,\cdot)$ for the reports $r=\E_p X(a;\cdot)$ and $s=\E_p X(b;\cdot)$. If $(a,r)$ is also optimal at another belief $p'$ with $A(p')=\{a,b\}$, then by the above observation, the value of information is affine along the line segment $\overline{pp'}$, and $(a,r)$ is optimal everywhere along $\overline{pp'}$. Because the values must match everywhere along $\overline{pp'}$, $\E_q V(b,s(q),\cdot)$ is affine along $\overline{pp'}$ for the optimal report $s(q)$ at each belief $q$. By the above observation, it follows that $(b,s)$ is optimal everywhere on $\overline{pp'}$.

These observations impose strong conditions on how the questions $\vt{X}(a)$ and $\vt{X}(b)$ relate to one another for adjacent actions $a$ and $b$. Formalizing this idea requires two additional pieces of notation. First, for any vector $\vt{v}\in\R^{\Theta}$, let 
$\bar{\vt{v}}=\vt{v}-\frac{1}{|\Theta|}\sum_{\theta'\in\Theta} v(\theta')\vt{1}$
be the projection of $\vt{v}$ onto the hyperplane of vectors whose coordinates sum to $0$. 
Second, let $\Delta_a^b=\bar{\vt{u}}(b)-\bar{\vt{u}}(a)$ be the payoff difference vector. 

\begin{lem}[Adjacency Lemma]\label{lem: adjacency}
Suppose $X$ is incentivizable. Then it is aligned with $u$ on $\{a,b\}$ for every pair of adjacent actions $a$ and $b$. Equivalently, for each pair of adjacent actions $a$ and $b$, there exist $\rho\in\R$ and $\sigma\in\R\setminus\{0\}$ such that 
\[
\bar{\vt{X}}\left(b\right)=\rho\Delta_a^b+\sigma\bar{\vt{X}}\left(a\right); 
\]
if $\bar{\vt{X}}(a)$ or $\bar{\vt{X}}(b)$ is collinear with $\Delta_a^b$, then we can take $\rho\neq 0$.\footnote{Recall that two vectors $\vt{u},\vt{v}\in\R^\Theta$ are \textit{collinear} if there exists $\alpha\neq0$ such that $\vt{v}=\alpha \vt{u}$.}
\end{lem}

The Adjacency Lemma provides a key tool in testing whether a question profile is incentivizable: it identifies a restriction on the values of the question at each pair of adjacent actions $a$ and $b$. 

The proof of the lemma has two steps. The first step uses the above observations about the value of information to show that, for adjacent actions $a$ and $b$ and beliefs $p$ and $p'$ at which $a$ and $b$ are $u$-optimal, if the same report $r$ is optimal at beliefs $p$ and $p'$ following action $a$, then there must also be a report $s$ that is optimal at both $p$ and $p'$ following action $b$. 
For the second step, we employ a linear algebra argument to show that the vector $\vt{X}(b)$ must belong to the linear space spanned by the vectors $\vt{u}(b)-\vt{u}(a)$, $\vt{X}(a)$, and $\vt{1}$ (the latter because we apply the linear condition to the space of beliefs, which satisfy the condition $\sum p(\theta)=1$). The result then follows from straightforward algebra.

Adjacency of actions $a$ and $b$ implies that there is a $\left(\left|\Theta\right|-2\right)$-dimensional set of beliefs $p\in\Delta(\Theta)$ for which actions $a$ and $b$ are uniquely $u$-optimal. Since the Adjacency Lemma is based on arguments restricted to this space, it has no bite when $|\Theta|=2$---indeed, when $|\Theta|=2$, every question profile is aligned with $u$ on each pair of adjacent actions---and limited bite when $|\Theta|=3$.  We discuss this issue in Section \ref{sec: Three states}.


\marcin{Colin - can you check computations in this Example?}\colin{I double checked and I think it's correct.}
\begin{example}[Second-order beliefs]\label{ex: second order beliefs}
The following example is a slightly simplified version of a belief-updating experiment from \citet{enke2023cognitive}; the same conclusions apply to their original experiment.

The decision problem involves forecasting a binary event. The action set and state space $A=\Theta=\{0,1/n,\dots,1\}$ consist of (discretized) probabilities that the event occurs, where $n\geq 4$. One can think of $\theta$ as the ``true'' probability given the available information, which is known to the researcher but about which the subject may be uncertain (for example because he has doubts about how to update his beliefs in light of the information he observes). This uncertainty is captured by the belief $p\in \Delta(\Theta)$.

The subject is rewarded more for forecasts that are closer to the state according to the payoff function $u(a;\theta)=-(a-\theta)^2$. The $u$-optimal actions are those that are closest to $\E_p[\theta]$. The adjacency graph therefore forms a line: $a_i$ and $a_j$ are adjacent if and only if $|a_i-a_j|=1/n$. 


For some $x\in [0,1/2]$, the researcher wishes to elicit the subject's confidence in his choice by asking how likely he believes it is that his action is within $x$ of the true value of $\theta$, as in equation \eqref{ex: within x}. 

To check whether $X$ is incentivizable, we use Lemma \ref{lem: adjacency}. Note that the coordinates of the vector $\vt{\bar{X}}(a)$ take at most two values. Hence if $a$ and $b=a+1/n$ are two adjacent actions and $\sigma\neq 0$, then the coordinates of $\vt{\bar{X}}(a)-\sigma\vt{\bar{X}}(b)\neq \vt{0}$ take at most four distinct values. 
At the same time, for each $\theta$,
\[
\Delta_a^b(\theta)=\bar{u}(b;\theta)-\bar{u}(a;\theta)=\frac{2\theta-1}{n},
\]
and hence the coordinates of $\Delta_a^b$ take on $n+1\geq 5$ different values. By the Adjacency Lemma, since $\vt{\bar{X}}(a)-\sigma\vt{\bar{X}}(b)\neq \vt{0}$ implies $\rho\neq 0$, $X$ is not incentivizable.

This result is not specific to the quadratic payoffs in the decision problem: the same conclusion applies if the payoff is replaced with any strictly proper scoring rule (one for which reporting an action close to the expectation of $\theta$ is $u$-optimal). 

Faced with this result, what should the researcher do? One option is to use a different measure of decision confidence that \emph{is} incentivizable. For instance, according to Corollary \ref{cor: En Hua}, the expected regret question of equation $\eqref{ex: En Hua}$ is incentivizable in every decision problem.

\end{example}




There is a gap between the necessary condition for incentivizability from the Adjacency Lemma ---namely, alignment on all adjacent pairs---and the sufficient condition from Proposition \ref{Observation 0}---namely, alignment on the full set of actions. In the rest of the paper, we show how to close this gap in three canonical classes of decision problems. In one of those classes, alignment on adjacent pairs is sufficient (and necessary); in another, alignment on the full set of actions is necessary (and sufficient); and in the third, a weaker form of alignment on all actions is both necessary and sufficient.

Our approach relies on the following observation. If the decision problem has three actions $a$, $b$, and $c$ such that $\{a,b\}$ and $\{b,c\}$ are both adjacent pairs, the restrictions implied by the Adjacency Lemma for these two pairs may interact with each other, leading to additional information about which questions are incentivizable. We therefore look to analyze the restrictions across all adjacent pairs simultaneously.

The \emph{adjacency graph} is the undirected graph with vertices $A$ and edges consisting of the adjacent pairs $\{a,b\}$. Note that, since there are no redundant or dominated actions, the adjacency graph is connected for every decision problem. A basic intuition across the next three sections is that the more edges there are in the adjacency graph, the more powerful are the restrictions imposed by the Adjacency Lemma.

\section{Adjacency trees}\label{sec: trees}

We first consider the case in which the adjacency graph is a tree (i.e., has no cycles). Trees naturally arise in problems with ordered actions, as in Example \ref{ex: second order beliefs}. Example \ref{ex: star} below describes a simple decision problem in which the adjacency graph forms a star. In this case, as the following result shows, alignment with $u$ on the full set of actions is not necessary for incentivizability.

\begin{thm}\label{prop: tree}
Suppose the adjacency graph is a tree. Then $X$ is incentivizable if and only if it is aligned with $u$ on every pair of adjacent actions.
\end{thm}

That alignment with $u$ on adjacent pairs is necessary for incentivizability holds in general, as indicated by Lemma \ref{lem: adjacency}. 
That it is sufficient follows from Proposition \ref{prop: piecewise aligned sufficient} below, which is a generalization of Proposition \ref{Observation 0}. 


Say that action $a$ is \emph{splitting} if removing it from the adjacency graph makes the graph disconnected. 
A \emph{splitting collection} $\left\{ A_{0},\dots,A_{k}\right\}$ is a collection of subsets $A_{i}\subseteq A$ such that $\bigcup A_{i}=A$ and, for all distinct $i$ and $j$, either $A_i\cap A_j = \emptyset$ or $A_i\cap A_j =\{a\}$ for some splitting action $a$. If the adjacency graph is a tree, a splitting collection is formed by the set $\mathcal{A}=\{A_0,\dots,A_k\}$ consisting of all pairs of adjacent actions; that is, $\tilde A\in \mathcal{A}$ if and only if $\tilde A = \{a,b\}$ for some adjacent actions $a$ and $b$. 




\begin{example}\label{ex: star}
    Suppose the decision problem involves guessing the correct state, with the option of taking a safe action. The action set is $A=\Theta\cup \{a_s\}$, with payoffs $u(\theta,\theta)=1$, $u(a_s,\theta)=s$ for some $s\in(0,1)$, and $u(a,\theta)=0$ for $a\neq \theta,a_s$. 
    If $s\geq 1/2$, then the adjacency graph is a star with action $a_s$ in the centre. Letting $\{B_0,B_1\}$ be any partition of $\Theta$, $\left\{B_0\cup \{a_s\}, B_1\cup\{a_s\}\right\}$ forms a splitting collection. Another splitting collection is given by the set of pairs $\{\theta,a_s\}$ for all $\theta\in\Theta$.
\end{example}


We say that $X$ is \emph{piecewise aligned with $u$} if it is aligned with $u$ on each element of some splitting collection. Theorem \ref{prop: tree} follows from the Adjacency Lemma together with the following result.

\begin{prop}\label{prop: piecewise aligned sufficient}
If $X$ is piecewise aligned with $u$, then it is incentivizable.
\end{prop}

To understand the main idea of the proof of this result, consider a binary splitting collection with sets of actions $A_0$ and $A_1$ and splitting action $a_0$. If actions are restricted to either $A_i$, then by the piecewise alignment assumption, the BDM construction of Proposition \ref{Observation 0} can be used to incentivize $X$. This construction gives rise to a well-defined elicitation method on the full action set if the two methods agree on $A_0\cap A_1 = \{a_0\}$. In the proof, we construct a positive affine transformation of one of the elicitation methods that ensures agreement on $a_0$. 

We show that the new elicitation method incentivizes $X$ on the union of the two sets. Indeed, by construction, given any belief $p$ at which an action $a\in A_i$ is $u$-optimal, no other action in $A_i$ leads to a higher expected payoff from the elicitation method. It remains to show that, similarly, no action $a'\in A_j$ (for $j\neq i$) does better than $a$ at $p$. To do so, we prove that if $p_0$ and $p_1$ are beliefs such that, for each $i$, some action in $A_i$ is $u$-optimal at belief $p_i$, then there exists a convex combination of $p_0$ and $p_1$ at which $a_0$ is $u$-optimal. Taking one of these beliefs to be $p$ and the other to be a belief at which $a'$ is $u$-optimal, it follows that there is a belief at which $a_0$ is $u$-optimal but not $V$-optimal, contrary to the way the method was constructed.

\section{Adjacency cycles}\label{sec: complete}

When the adjacency graph forms a tree, the incentivizable questions are those that are piecewise aligned with $u$; the relationship between $\bar{\vt{X}}(a)$ and $\bar{\vt{X}}(b)$ described in the Adjacency Lemma is both necessary and sufficient for incentivizability. Cycles in the adjacency graph impose additional restrictions: not only must the relationship in the Adjacency Lemma hold for adjacent actions, it must also be consistent all the way around each cycle.

Formally, an \textit{adjacency cycle} is a tuple $C=\left(a_{0},\dots,a_{n}\right)$ such that $a_n=a_0$ and actions $a_{i}$ and $a_{i+1}$ are adjacent for each $i=0,\dots,n-1$. We say that $n$ is the \emph{length} of the cycle $C$. We abuse notation slightly by writing $a\in C$ to mean that $a=a_i$ for some $i$. 

The additional restrictions cycles impose are related to linear independence relationships among the corresponding utility vectors. To understand how, consider a cycle $(a_0,\dots,a_n)$ and recall that the Adjacency Lemma implies that, for each $i$, there exist $\rho_i$ and $\sigma_i$ such that
\[
\bar{\vt{X}}(a_i)= \rho_i \Delta^{a_i}_{a_{i+1}} + \sigma_i \bar{\vt{X}}(a_{i+1}).
\]
Iterating these equations around the cycle starting from $a_0$ leads to an expression for $\bar{\vt{X}}(a_0)$ as a linear combination of $\bar{\vt{X}}(a_0)$ and $\Delta_{a_1}^{a_0},\dots,\Delta_{a_n}^{a_{n-1}}$. Since $\sum_{i=0}^{n-1} \Delta^{a_i}_{a_{i+1}}=0$, we can replace the $\Delta_{a_n}^{a_{n-1}}$ with $-\sum_{i=0}^{n-2} \Delta^{a_i}_{a_{i+1}}=0$. If the remaining $\Delta_{a_{i+1}}^{a_i}$ terms are linearly independent, the resulting equation implies restrictions on the coefficients in the linear combination and on the vectors $\bar{\vt{X}}(a_0)$  and $\bar{\vt{u}}(a_i)$ that do not follow directly from the Adjacency Lemma.


A common setting in which adjacency cycles appear is when the adjacency graph is complete. This is the case, for instance, in Example \ref{ex: star} when $s<1/2$. 

The main result of this section shows that, under a mild richness condition on utilities, in decision problems with complete adjacency graphs, a question is incentivizable if and only if it is aligned with $u$.\footnote{Theorem \ref{prop:complete adjacency regular} is related to Proposition 3 of \citet{morris2004best}, which discusses conditions under which two games with identical best response correspondences must have payoffs that can be obtained from each other by positive affine transformations. The proofs of the two results share similar mathematical ideas.}

\begin{thm}\label{prop:complete adjacency regular}
Suppose the adjacency graph is complete and $\left|A\right|\geq4$. 
Suppose in addition that for any four distinct actions $a,b_{0},b_{1},b_{2}$, the set of vectors
$\{\Delta_{a}^{b_{i}}\}_{i=0,1,2}$ 
is linearly independent. Then $X$ is incentivizable if and only if it is aligned with $u$. 
\end{thm}

In addition to completeness of the adjacency graph, this theorem relies on two assumptions: the set of actions must be sufficiently large and the payoffs from the actions must be sufficiently independent. These assumptions make the Adjacency Lemma particularly powerful. For example, as discussed in Section \ref{sec: necessary conditions}, the Adjacency Lemma only has any bite if there are at least three states of the world, which must be the case if the two assumptions of the theorem hold. These assumptions also exclude decision problems with larger state spaces that can effectively be reduced to problems with three states, such as those in which the payoff from each action is $0$ in every state but the first three.



\begin{example}\label{ex: complete adjacency thm}
    Consider a decision problem with at least four states in which the subject is asked to guess the state and receives a reward for guessing correctly that may depend on which state is realized. Thus $A=\Theta$ and
    \[
    u(a;\theta)=r_\theta\cdot \mathbb{1}\{a=\theta\},
    \]
    where $r_\theta>0$ for all $\theta$. The adjacency graph for this problem is complete: given any two states $\theta$ and $\theta'$, the actions $\theta$ and $\theta'$ are the only $u$-optimal actions at the belief that assigns probability $r_{\theta'}/(r_\theta+r_{\theta'})$ to state $\theta$ and all of the remaining probability to state $\theta'$. Note that the linear independence condition of Theorem \ref{prop:complete adjacency regular} is satisfied.

    Suppose the researcher seeks to elicit the subject's belief about whether he correctly guessed the state, which is described by the question profile
    \[
    X(a;\theta)=\mathbb{1}\{a=\theta\}.
    \]
    Taking $\gamma(a)=1/r_a$ and $\kappa(a)=0$ for each $a$ and $\vt{d}=0$, we see that $X$ is aligned with $u$, and therefore can be incentivized using the BDM construction of Proposition \ref{Observation 0}.

    Now consider the same question profile $X$ in a different decision problem where the subject can also receive a smaller reward for a ``close'' guess. Let $\Theta=\{1,\dots,n\}$ and
    \[
    \tilde u(a;\theta)=r_\theta\mathbb{1}\{a=\theta\}+\frac{r_\theta}{2}\mathbb{1}\{|a-\theta|=1\}
    \]
    where $\min_\theta r_\theta>\max_\theta r_\theta/2>0$. The adjacency graph for this problem is again complete. In this case, however, $X$ is not aligned with $\tilde u$. Since the linear independence condition in Theorem \ref{prop:complete adjacency regular} holds, $X$ is not incentivizable.
\end{example}

\marcin{Above: explain how }

To explain the main ideas in the proof of Theorem \ref{prop:complete adjacency regular}, we first need some additional terminology. 
An adjacency cycle $C=\left(a_{0},\dots,a_{n}\right)$ is \emph{internally independent} if the sequence of vectors $\Delta_{a_0}^{a_1},\dots,\Delta_{a_0}^{a_{n-1}} $ is linearly independent. Let $V_C$ be the linear space spanned by this sequence of vectors.  The cycle $C$ is internally independent if and only if $\mathrm{dim}\, V_C=n-1$. Notice that the space $V_C$ and the linear independence of $\Delta_{a_0}^{a_1},\dots,\Delta_{a_0}^{a_{n-1}} $ depend only on $C$ itself, and not on which action in the cycle is labeled as $a_0$.



As described above, for each action $a_0$ in an adjacency cycle, applying the Adjacency Lemma iteratively along pairs $\{a_i,a_{i+1}\}$ of adjacent actions in the cycle gives rise to an expression for $\bar{\vt{X}}(a_0)$ as a linear combination of $\Delta_{a_{i+1}}^{a_i}$ terms for $i=0,\dots,n-2$. If this cycle is internally independent, then either $\bar{\vt{X}}(a_0)\in V_C\setminus\{\vt{0}\}$ or all of the coefficients in the linear combination must be equal to $0$. In the latter case, iteratively applying the Adjacency Lemma again gives a linear relationship between any $\bar{\vt{X}}(c)$ (for $c$ in the cycle) and $\bar{\vt{X}}(a_0)$ together with $\Delta_{a_{i+1}}^{a_i}$ terms. This relationship can be simplified using the fact that the coefficients in the previous linear combination are zero to show that $X$ is aligned with $u$ on $C$, leading to the following lemma.

\begin{lem}\label{lem:adjacency cycles}
Let $C$ be an
internally independent adjacency cycle. If $X$ is incentivizable, then either $X$ is aligned with $u$ on $C$, or 
$\bar{\vt{X}}(a)\in V_C\setminus\{\vt{0}\}$ for all $a\in C$.
\end{lem}

While Lemma \ref{lem:adjacency cycles} identifies conditions that an incentivizable question profile must satisfy on a given internally independent cycle, it says nothing about more complicated adjacency graphs. The following lemma shows how alignment with $u$ on multiple cycles or other subsets of actions can, under mild conditions, be combined to obtain alignment with $u$ on their union.

\begin{lem} \label{lem:combine sets}
Suppose $X$ is aligned with $u$ on sets of actions $B$ and $D$. If there exist actions $a_0,a_1\in B\cap D$ such that $\bar{\vt{X}}(a_0)$ and  $\bar{\vt{X}}(a_1)$ are not collinear, then $X$ is aligned with $u$ on $B\cup D$.
\end{lem}

Alignment with $u$ on $B$ gives rise, for each action in $B$, to a system of equations that the parameters relating $\bar{\vt{X}}$ and $\bar{\vt{u}}$ must satisfy; naturally, the same is true for $D$. In the proof of this lemma, we show that if $B$ and $D$ share two actions for which the values of the questions are linearly independent, then the corresponding systems of equations must share the same solution for $B$ as for $D$. 


The main idea behind Theorem \ref{prop:complete adjacency regular} is that, given any cycle $C$ through an action $a$, Lemma \ref{lem:adjacency cycles} implies that any incentivizable $X$ is either aligned with $u$ on $C$ or lies in $V_C\setminus\{\vt{0}\}$. The linear independence condition can be used to eliminate the latter possibility by considering multiple cycles through $a$; more precisely, there must be \emph{some} cycle $C$ on which $X$ is aligned with $u$, as the intersection of the sets $V_{C'}$ across cycles $C'$ containing $a$ is $\{\vt{0}\}$. Varying $a$ and applying Lemma \ref{lem:combine sets} leads to alignment with $u$ on the entire action set.

In Appendix \ref{app: complete}, we identify a more general class of problems in which alignment with $u$ is both necessary and sufficient for incentivizability. This more general result applies when there is a rich enough structure of internally independent cycles in the adjacency graph, which can be the case even when the graph is not complete.

\section{Product problems}

In many experiments, subjects perform a sequence of tasks. These tasks may be identical or they may differ. The subject's payoff is a sum or weighted average of the payoffs in the various tasks, as in the common experimental design in which one task is randomly selected for payment (see \citet{charness2016experimental} and \citet{azrieli2018incentives} and the references therein). In such cases, the researcher may be interested in eliciting beliefs related to the entire sequence of actions chosen by the subject. For example, a subject may solve a test with multiple problems---with his payoff being equal to his score---and the researcher may want to ask the subject what he believes about his overall performance on the test. Alternatively, to gauge the impact of learning across repetitions of the same task, the researcher may want to elicit subjects' beliefs about their change in performance between the beginning and the end of the experiment. 

\marcin{The reference to \citet*{azrieli2018incentives} is not perfect because the paper only provides some conditions under which RPS "random problem selection" (that's how they call it; what is our preferred terminology?) is incentive compatible. The issues are different from our paper. The best would be to find some overview that discusses and advocates using such a design.} \colin{Is the addition of the Charness et al.\ paper sufficient?}

To formalize this idea, we define a \emph{product problem} $\left(\Theta,A,u\right)$ to be a decision problem in which 
$\Theta=\times_{i}\Theta_{i}$, $A=\times_{i}A_{i}$, 
and $u\left(a;\theta\right)=\sum_{i}u_{i}\left(a_{i},\theta_{i}\right)$ for some sets $(\Theta_1,\dots,\Theta_I)$ and $(A_1,\dots,A_I)$, and some functions $u_i:A_i\times\Theta_i\longrightarrow\R$.\footnote{\label{ft:lin independence}Theorem \ref{prop:product problems} below generalizes to decision problems with payoffs of the form $u(a;\theta)=\sum_iu_i(a_i,\theta)$, i.e., where the payoff within each task can depend on the entire state of the world $\theta$, under the assumption that the collection of subspaces $E_i=\mathrm{span}\{\Delta_{a_i}^{b_i}:a_i,b_i\in A_i\}$ is linearly independent.} 
As noted above, the additive separability of $u$ captures commonly used incentives in which one choice is randomly selected for payment. We refer to each $(\Theta_i, A_i, u_i)$ as a \emph{task}. We write $a_{-i}$ for a profile $(a_1,\dots,a_{i-1},a_{i+1},\dots,a_I)$ and $a_ia_{-i}$ for the profile whose $i$th coordinate is $a_i$ and remaining coordinates are given by $a_{-i}$.


\begin{example}\label{ex: Multiple choice test}
The decision problem is a test consisting of $I\geq 3$ multiple choice questions. The state space and action space are given by $\Theta=A=\Omega^{I}$, where $\Omega$ is a finite set containing at least two elements describing the possible answers to any given question. Coordinate $i$ corresponds to the $i$th question: $\theta_i$ is the correct answer to question $i$ and $a_i$ is the subject's answer to question $i$. The payoff in the decision problem is the score on the test:
$u(a;\theta) = \sum_{i=1}^I \mathbb{1}\{\theta_i=a_i\}.$
The subject has a belief $p\in\Delta(\Theta)$. The $u$-optimal choice of action in each task $i$ is the most likely state according to the marginal distribution of $p$ over $\Theta_i$.
\end{example}

Note that we make no assumptions about the correlation among states across tasks: the subject can hold any belief about the joint distribution of $(\theta_1,\dots,\theta_I)$. In particular, the subject need not view the states as independent, nor must there be a fixed state across tasks.

Product problems share a structure that distinguishes them from the problems we have analyzed so far. Two actions $a,b\in A$ are adjacent only if they differ in exactly one task, that is, if there is some $i$ such that $a_i\neq b_i$ and $a_{-i}= b_{-i}$.  Conversely, if $a_i$ and $b_i$ are adjacent in task $i$, then the product actions $a_ia_{-i}$ and $b_ia_{-i}$ are adjacent for all $a_{-i}$. The adjacency graph of the product problem is generally neither complete (even if the adjacency graph in each task is complete) nor a tree. 

\begin{figure}[t]
    \centering

\begin{tikzpicture}[scale=0.75, transform shape]
    \coordinate (A) at (0,0,3);\coordinate (AB) at (0,0,0);
    \coordinate (B) at (3,0,3);\coordinate (BB) at (3,0,0);
    \coordinate (C) at (3,3,3);\coordinate (CB) at (3,3,0);
    \coordinate (D) at (0,3,3);\coordinate (DB) at (0,3,0);
    
    \draw[] (C) -- (D) -- (DB) -- (CB) -- cycle;
    \draw[] (B) -- (A) -- (D) -- (C) -- cycle;
    \draw [ultra thick, blue] (C) -- (D);
    \draw [ultra thick, blue] (C) -- (CB);
    \draw [ultra thick, blue] (C) -- (B);
    \draw[] (CB) -- (BB);
    \draw[] (B) -- (BB);
    \draw[dashed] (DB) -- (AB);
    \draw [dashed, ultra thick,  orange] (C) -- (D);
    \draw[dashed] (A) -- (AB);
    \draw[dashed, orange] (BB) -- (AB);
    \draw[orange] (CB) -- (DB);
    \draw[orange] (A) -- (B);

    \node at (A) [circle,fill=black,inner sep=1pt,label=below left:{$110$}] {};
    \node at (B) [circle,fill=black,inner sep=1pt,label=below right:{$010$}] {};
    \node at (C) [circle,fill=black,inner sep=1pt,label=above left:{$000$}] {};
    \node at (D) [circle,fill=black,inner sep=1pt,label=above left:{$100$}] {};
    \node at (AB) [circle,fill=black,inner sep=1pt,label=below right:{$111$}] {};
    \node at (BB) [circle,fill=black,inner sep=1pt,label=above right:{$011$}] {};
    \node at (CB) [circle,fill=black,inner sep=1pt,label=above right:{$001$}] {};
    \node at (DB) [circle,fill=black,inner sep=1pt,label=above left:{$101$}] {};

\end{tikzpicture}
\caption{Adjacency graph for Example \ref {ex: Multiple choice test} with $\Omega=\{0,1\}$ and $I=3$.}\label{fig:MC adjacency}
\end{figure}
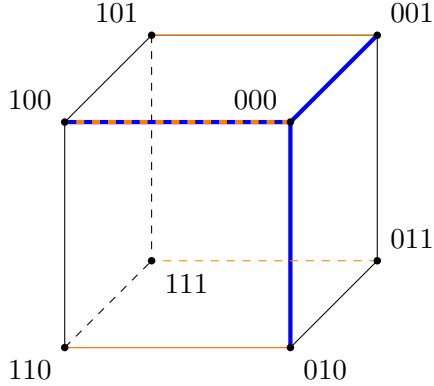

Figure \ref{fig:MC adjacency} depicts the adjacency graph in Example \ref{ex: Multiple choice test} when $\Omega=\{0,1\}$ and $I=3$. 
Note that while the collection of payoff difference vectors associated with edges exiting a single node---such as the blue edges $\Delta_{000}^{100}$, $\Delta_{000}^{010}$, and $\Delta_{000}^{001}$---are linearly independent, no cycle is internally independent because parallel edges correspond to identical payoff difference vectors. Thus, for instance, the orange edges correspond to $\Delta_{000}^{100}=\Delta_{010}^{110}=\Delta_{001}^{101}=\Delta_{011}^{111}$.


We say that question profile $X$ depends on task $i$ \textit{trivially} if, for each $a_{-i}$, the vectors $\bar{\vt{X}}(a_ia_{-i})$ are collinear for all $a_i$. Trivial dependence on task $i$ means that changing $a_i$ entails replacing $\vt{X}(a)$ with an equivalent question. The following result provides necessary and sufficient conditions for incentivizability in product problems.

\begin{thm}\label{prop:product problems}
Let $\left(\Theta_{i},A_{i},u_{i}\right)_{i=1}^I$ be a product problem with $I\geq 3$. Suppose that for each $i$, either $A_i$ contains only two actions, or the adjacency graph for problem $\left(\Theta_{i},A_{i},u_{i}\right)$ is complete and the vectors $\left\{ \Delta^{b_i}_{a_i},\Delta^{c_i}_{a_i}\right\}$ are linearly independent for all distinct $a_i,b_i,c_i\in A_i$. Suppose in addition that there are at least three tasks on which $X$ does not depend trivially. Then $X$ is incentivizable if and only if there exist $v(a),\kappa\left(a\right)\in\R$ with $v(a)\neq0$ for each $a\in A$, $\tau_i\in\R$ for each $i$, and $\vt{d}\in\R^{\Theta}$
such that 
\begin{align}\label{eq: weighted-aligned}
X\left(a;\theta\right)
=\kappa(a)
+v(a)\left(d(\theta)+\sum_{i}\tau_i u_i(a_i,\theta_i)\right)
\end{align}
for all $a$ and $\theta$. 
\end{thm}



Along the same lines as Theorem \ref{prop:complete adjacency regular}, we assume that each task has a complete adjacency graph (which holds vacuously when the action set is binary). However, relative to Theorem \ref{prop:complete adjacency regular}, the other requirements for each task are significantly weaker: there are no restrictions on the number of actions, and we require linear independence only of pairs of payoff difference vectors instead of triples. 

The characterization of incentivizable question profiles in \eqref{eq: weighted-aligned} is more permissive than that in Theorem \ref{prop:complete adjacency regular}.  If $\tau_i$ is constant across $i$, the expression in \eqref{eq: weighted-aligned} implies that $X$ is aligned with $u$. However, in contrast to the case of a single problem with complete adjacency, there are many questions not aligned with $u$ that are also incentivizable (namely, those for which $\tau_i$ varies across $i$). The additional freedom in the product problem results from a smaller number of cycles and a larger number of linear dependencies in the payoffs. 

We refer to any $X$ that satisfies the conditions of Theorem \ref{prop:product problems} as \emph{weighted aligned}. The following example illustrates the added flexibility of weighted alignment relative to alignment, as well as its restrictiveness.

\addtocounter{example}{-1}
\begin{example}[continued]\label{ex:MCT tw parts}
It is straightforward to verify that in the product problem of Example \ref{ex: Multiple choice test}, if $|\Omega|>2$, the adjacency graph for each problem $i$ is complete and all pairs $\left\{ \Delta^{b_i}_{a_i},\Delta^{c_i}_{a_i}\right\}$ are linearly independent. Therefore, Theorem \ref{prop:product problems} applies to any question profile that depends nontrivially on at least three tasks. 

Suppose the test consists of two parts: the first part comprises questions $1$ through $I_1\geq 2$, while the second part comprises questions $I_1+1$ through $I=I_1+I_2$. Consider the question that asks the subject about the expected improvement in his average score from the first part of the test to the second: 
\(
X(a;\theta)=\frac{1}{I_2}\sum_{i> I_1} \mathbb{1}(a_i=\theta_i)
-\frac{1}{I_1}\sum_{i\leq I_1} \mathbb{1}(a_i=\theta_i).
\)

This question profile is weighted aligned, and is therefore incentivizable by Theorem \ref{prop:product problems}. At first glance, this result may be surprising as $X$ seems to create opposing incentives in the two parts of the test. Nevertheless, $X$ can be incentivized using a simple modification of the BDM elicitation mechanism from Proposition \ref{Observation 0} that involves adding the payoff from the decision problem:
\begin{align}\label{ex: BDMproduct}
V(r,a,\theta)&=\int_0^r X(a;\theta)dx + \int_r^1 xdx -\frac{1}{2}+\sum_i u_i(a_i,\theta_i) \\
&= \sum_{i} \left[1-\mathbb{1}\{i\leq I_1\}\frac{r}{I_1}+\mathbb{1}\{i> I_1\}\frac{r}{I_2}\right]u_i(a_i,\theta_i)
-\frac{r^2}{2}.\notag
\end{align}
The two integral terms provide incentives for truthful reporting of $\E _p X(a;\cdot)$, as in the standard BDM mechanism. For the optimal action $a$, we can think of simultaneously choosing $r$ and $a$ to maximize $V$. For each $i$, because the corresponding coefficient on $u_i$ in the square brackets is positive for every $r\in [0,1]$, the optimal action $a_i$ is the one that maximizes the expectation of $u_i(a_i,\theta_i)$. Therefore, the overall payoff under $V$ provides the correct incentives for the action $a$.

\enlargethispage{\baselineskip}

For an example of a question that is not weighted aligned, and hence not incentivizable, consider elicitation of the probability the subject assigns to receiving a score of at least $z$, where $z\in \{1,\dots,I\}$.\footnote{This question is similar to that of \citet{mobius2022managingself}, who elicit the subject's belief that his score on an incentivized IQ test is above the median among the participants.} This question corresponds to
\[
Y(a;\theta) = 
\begin{cases}
1 & \text{if } \sum_{i=1}^I \mathbb{1}\{\theta_i=a_i\} \geq z\\
0 & \text{otherwise}.
\end{cases}
\]

We claim that $Y$ is not incentivizable. Since there is no task on which $Y$ depends trivially, Theorem \ref{prop:product problems} applies. Thus it suffices to show that $Y$ is not of the form specified in \eqref{eq: weighted-aligned}. 
To make this precise, suppose, in addition, that $z<I$; it is straightforward to adapt the argument for the case of $z=I$. 

The idea is that, for a given action, a change in the state that affects the utility but not the value of $Y$ implies restrictions on the parameters in \eqref{eq: weighted-aligned} that are violated for actions where the same change in the state changes the value of $Y$. On the contrary, suppose that $Y$ has representation \eqref{eq: weighted-aligned}. Given any answer profile $a$ and question $i$, let $\theta=a$ and let $\theta'$ be a state that differs from $\theta$ in the value of question $i$, i.e., such that $\theta_i'\neq \theta_i$ and $\theta_j'=\theta_j$ for all $j\neq i$. Thus $a$ gets a perfect score of $I$ if the state is $\theta$ and a score of $I-1\geq z$ if the state is $\theta'$. In particular, $Y(a;\theta)=Y(a;\theta')=1$. Since $v(a)\neq 0$, equating the right-hand sides of \eqref{eq: weighted-aligned} for $Y(a;\theta)$ and $Y(a;\theta')$ yields $d(\theta)+\sum_j\tau_j = d(\theta')+\sum_{j\neq i}\tau_j$, which further implies $d(\theta)+\tau_i=d(\theta')$. Repeating the argument for $a'=\theta'$ with the roles of $\theta$ and $\theta'$ reversed, we obtain $d(\theta')+\tau_i=d(\theta)$. Therefore, $\tau_i=0$ and $d(\theta)=d(\theta')$. Now consider an action $a''$ that agrees with $\theta$ on exactly $z$ coordinates, including coordinate $i$. Then $a''$ agrees with $\theta'$ on $z-1$ coordinates, and we have $Y(a'';\theta)=1\neq 0 = Y(a'';\theta')$. From \eqref{eq: weighted-aligned}, since $\tau_i=0$, it must be that $d(\theta)\neq d(\theta')$, a contradiction.

Intuitively, one may think that using an elicitation method (such as BDM) whose value is increasing in $\E_p[Y(a;\cdot)]$ would suffice for incentivizability, as then the subject can increase his payment at the belief elicitation stage by ``doing well'' on the test. The above argument shows that this intuition is incorrect. The problem lies in the definition of doing well: maximizing the expected score in the test is not generally the same thing as maximizing the likelihood of that score being above a fixed threshold.
\end{example}

\subsection{Proof idea}

The proof of Theorem \ref{prop:product problems} can be found in Appendix \ref{sec: Proof of Product Problems}. Here, we describe the key ideas. For simplicity, we focus on Example \ref{ex: Multiple choice test} with actions $A_i=\{0_i,1_i\}$.

The proof that equation \eqref{eq: weighted-aligned} is sufficient for incentivizability is relatively straightforward: the argument directly extends the construction from equation \eqref{ex: BDMproduct}.

For the remainder of this section, we assume that $X$ is incentivizable. By the Adjacency Lemma, for any pair of adjacent actions $a$ and $b$, there are coefficients $x(a,b)\neq 0$ and $y(a,b)$ such that $\bar{\vt{X}}(a)=x(a,b)\bar{\vt{X}}(b)+y(a,b)\Delta_a^b$. For the purpose of this discussion, we assume that the coefficients $x(a,b)$ and $y(a,b)$ are uniquely defined. 
Note also that, for any such pair, there is a single task $i$ on which the actions $a_i$ and $b_i$ differ, and $\Delta_{a}^{b} =\Delta_{a_i}^{b_i}$, where $\Delta_{a_i}^{b_i}$ refers to the payoff difference $\bar{\vt{u}}_{i}(b_i)-\bar{\vt{u}}_i(a_i)$. We say that $i$ is the \emph{relevant} task for the pair $(a,b)$.

The key step in the proof is to show that each cycle $(a^0,\dots,a^n=a^0)$ of adjacent actions is \emph{exact}, that is, that $x(a^0,a^1)x(a^1,a^2)\cdots x(a^{n-1},a^n)=1.$ We first explain the connection between exactness of all cycles and \eqref{eq: weighted-aligned}, and then explain why all cycles are exact.

The exactness of all cycles implies that we can define $v(a)=x(\vt{0},a^1)\cdots x(a^{n-1},a^n)$ for any path $a^0=\vt{0},\dots,a^n=a$ of adjacent actions, and the definition does not depend on the chosen path. In particular, for any two adjacent actions $a$ and $b$, $v(a)=x(b,a)v(b)$. Letting $\bar{\vt{X}}^*(a)=v(a)\bar{\vt{X}}(a)$, a repeated application of the Adjacency Lemma implies that
\begin{align}\label{eq:repr}
\bar{\vt{X}}^*(a)=\bar{\vt{X}}^*(\vt{0})+
\sum_{m<n} y^*(a^m,a^{m+1})\Delta_{a^m}^{a^{m+1}}   
\end{align}
for some (arbitrary) path $a^0=\vt{0},\dots,a^n=a$. 

Consider paths of length $n= |\{i:a_i=1_i\}|$ from $\vt{0}$ to $a$. 
Along such paths, each task in the product problem is relevant at most once. Let $i_m$ be the relevant task for pair $(a^m,a^{m+1})$. Then, $\Delta_{a^m}^{a^{m+1}} =\Delta_{0_{i_m}}^{1_{i_m}}$. Because the vectors $\Delta_{0_i}^{1_i}$ are linearly independent across tasks, by equating the right-hand side of \eqref{eq:repr} across all length-$n$ paths, we conclude that the coefficients $y^*$ depend only on the relevant task for the transition; thus we can define $y_{i_m}=y^*(a^m,a^{m+1})$ and obtain
\begin{align*}
\bar{\vt{X}}^*(a)
=\bar{\vt{X}}^*(\vt{0})+\sum_{i\text{ s.t.\ }a_i=1} y_i\Delta_{0_i}^{1_i}
=\bar{\vt{X}}^*(\vt{0})-\sum_{i} y_i\bar{\vt{u}}_i(0_i)+\sum_{i} y_i\bar{\vt{u}}_i(a_i).
\end{align*}
Taking $d(\theta)=\bar{\vt{X}}^*(\vt{0})-\sum_{i} y_i\bar{\vt{u}}_i(0_i)$ gives an expression of the form in \eqref{eq: weighted-aligned}.

We now return to the question of why all cycles are exact. As the first step, consider the four-action cycle corresponding to the front face in Figure \ref{fig:MC adjacency}. 
A repeated application of the Adjacency Lemma along the cycle leads to
\[
\bar{\vt{X}}(000)
=x(000,100)x(100,110)x(110,010)x(010,000)\bar{\vt{X}}(000)
+s_1\Delta_{0_1}^{1_1}+s_2\Delta_{0_2}^{1_2}
\]
for some $s_1,s_2\in\R$. If $\bar{\vt{X}}(000)$ is not in the subspace spanned by the vectors $\Delta_{0_1}^{1_1}$ and $\Delta_{0_2}^{1_2}$, then the product of the $x$ coefficients on the right-hand side must be equal to $1$, meaning that the cycle is exact. Otherwise, 
by analyzing a number of cases, one can show that the cycles corresponding to the other five faces in Figure \ref{fig:MC adjacency} are all exact. These five cycles can be combined in such a way that all edges not belonging to the front face ``cancel out,'' thereby implying exactness of the original cycle.
In general, we show that all similar 4-cycles are exact.

In the second step, we show how to decompose arbitrary cycles into 4-cycles in such a way that the exactness of the original cycle follows from the exactness of cycles in the decomposition.

These arguments are not specific to binary actions. Extending beyond binary actions requires more care and conditions on linear independence of payoffs within tasks. In the proof of exactness, aside from the 4-cycles described above, we also use 3-cycles of adjacent actions with the same relevant task. We leave the details to the formal proof in the appendix.

\section{Discussion}\label{sec: discussion}

\subsection{Two or three states}\label{sec: Three states}

As noted in Section \ref{sec: necessary conditions}, the Adjacency Lemma has limited bite if there are only two or three states due to the dimensionality of its restrictions. With two states, the set of beliefs at which the subject is indifferent between two actions is a single point. With three states, the indifference set intersected with the set of beliefs along which there is a constant answer to a particular question is generically a single point. In both cases, the thesis of the Adjacency Lemma is trivially satisfied. 

If there are only two states, the elicitation problem becomes trivial as the subject's belief about the state can be elicited independent of the action, which is sufficient for the researcher to determine the expectation of any question $\vt{X}(a)$. 

With three states, eliciting the entire belief may still be a practical option as it requires asking for only two probabilities. If the researcher wants to ask for only one number, although the Adjacency Lemma is vacuous, looking at adjacencies can nonetheless be useful for understanding incentivizability. Suppose the overall payoff is a weighted sum of the payoff from the decision problem and that from a scoring rule applied at the belief elicitation stage. Suppose moreover that the scoring rule depends only on the reported belief and the realized value of $X(a;\theta)$, and not on $\theta$ directly.\footnote{This is a natural restriction that we do not impose in our model as doing so would have no effect on our results. When there are three states, we do not know whether this restriction has any bite; we expect that the convexity of the value function would place restrictions on incentivizable questions even if the value can depend on $\theta$ directly.} At any belief at which the subject is indifferent between two actions, the value of truthful reporting at the belief elicitation stage must be equal following these two actions. Depending on the structure of the decision problem, following these constant values along cycles of adjacent actions can imply restrictions on $X$; see Figure \ref{fig: three states}.

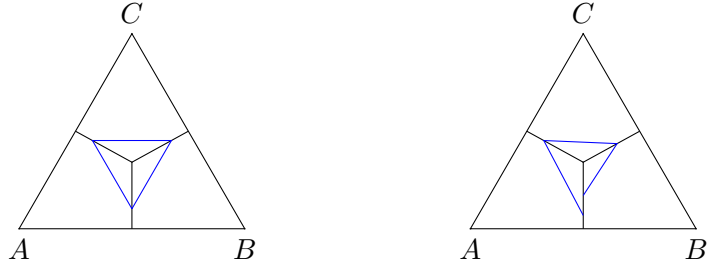
\begin{figure}[t]
\begin{center}
\begin{tikzpicture}[scale = 1.25]
\begin{scope}[xshift=4cm]
\node[mynode, scale=0.05] (center) at (1,0.59) {};
\node[mynode, label=below:$A$, scale=0.05] (A) at (0,0) {};
\node[mynode, label=below:$B$, scale=0.05] (B) at (2,0) {};
\node[mynode, label=above:$C$, scale=0.05] (C) at (1,1.73) {};
\node[] (AB) at ($(A)!0.5!(B)$) {};
\node[] (AC) at ($(A)!0.5!(C)$) {};
\node[] (BC) at ($(C)!0.5!(B)$) {};

\draw [->] (A) -- (B) -- (C) -- (A) -- cycle;
\draw [] (center) -- ($(A)!0.5!(B)$);
\draw [] (center) -- ($(A)!0.5!(C)$);
\draw [] (center) -- ($(C)!0.5!(B)$);

\draw [blue] ($(AB)!0.3!(center)$) -- ($(AC)!0.3!(center)$) -- ($(BC)!0.3!(center)$) -- ($(AB)!0.3!(center)$);

\end{scope}
\begin{scope}[xshift=8cm]
\node[mynode, scale=0.05] (center) at (1,0.59) {};
\node[mynode, label=below:$A$, scale=0.05] (A) at (0,0) {};
\node[mynode, label=below:$B$, scale=0.05] (B) at (2,0) {};
\node[mynode, label=above:$C$, scale=0.05] (C) at (1,1.73) {};
\node[] (AB) at ($(A)!0.5!(B)$) {};
\node[] (AC) at ($(A)!0.5!(C)$) {};
\node[] (BC) at ($(C)!0.5!(B)$) {};

\draw [->] (A) -- (B) -- (C) -- (A) -- cycle;
\draw [] (center) -- ($(A)!0.5!(B)$);
\draw [] (center) -- ($(A)!0.5!(C)$);
\draw [] (center) -- ($(C)!0.5!(B)$);

\draw [blue] ($(AB)!0.2!(center)$) -- ($(AC)!0.3!(center)$) -- ($(BC)!0.4!(center)$) -- ($(AB)!0.5!(center)$);
\end{scope}
\end{tikzpicture}    
\end{center}
\caption{Problems with three states. Each triangle depicts the simplex of beliefs. The black line segments illustrate the partition of the simplex according to which action is $u$-optimal. Blue line segments represent sets of beliefs on which $\vt{X}(a^*)$ is constant for the $u$-optimal action $a^*$. In the left triangle, because the blue lines form a triangle, adjacency considerations do not rule out incentivizability of $X$. In the right triangle, $X$ is not incentivizable if the payment for elicitation depends only on the value of $X$ and the reported belief.\label{fig: three states}}
\end{figure}

\subsection{Independent questions}


Our necessary conditions make use of independence assumptions on the payoffs in the decision problem. Similar results can be obtained if one replaces these assumptions with assumptions about independence of $\bar{\vt{X}}(a)$ across actions. For example, along the lines of Lemma \ref{lem:adjacency cycles}, if $X$ is incentivizable and the set of questions $\bar{\vt{X}}(a)$ is linearly independent for actions $a$ in some cycle $C$ of adjacencies, then one can show that $X$ must be aligned with $u$ on $C$. Lemma \ref{lem:combine sets} can then be used to obtain necessary conditions on the full set of actions.

\subsection{Non-affine questions}

We have restricted attention to eliciting beliefs about the expectation of some function $\vt{X}(a)$. \citet{lambert2011elicitation} studies elicitation of ``properties'' of beliefs, where a property corresponds to a partition of the simplex. He characterizes which properties are incentivizable, i.e., for which ones there exists a scoring rule incentivizing truthful reporting when the subject is asked only about the property associated with his belief. A question $\vt{X}(a)$ in our framework corresponds to a property that partitions the simplex into parallel hyperplanes (unless the question is trivial, in which case there is a single property for the entire simplex). This formulation captures many properties of interest and ensures incentivizability both in Lambert's sense and, if the question is independent of the action choice, in our sense. There are, however, properties---such as the median of some $X$---which may be of interest that are incentivizable in Lambert's context but to which our results do not apply. Nonetheless, we expect our general approach of focusing on adjacencies between actions to be useful for such non-affine questions.

\colin{Should we expand on this using the example of the median? With finitely many states, the median is a finite property. Looking along the boundary of adjacency between two actions $a$ and $b$, if the median for $\vt{X}(a)$ changes at a point where the median for $\vt{X}(b)$ does not, we should get a failure of incentivizability since the value must be affine for any fixed action and median.} \marcin{Let's add a result about generic non-incentivizability of the median.}




\subsection{Observability}

We assume the researcher knows the subject's utility function and observes actions and states. These assumptions are typically satisfied in lab experiments, but could be violated in field experiments. 
A lack of observability introduces additional issues that we view as orthogonal those we study. Our negative results describe the limits of nondistortionary elicitation; if the researcher does not observe all aspects of the experiment, elicitation becomes even more difficult. 

Interestingly, willingness-to-pay (or expected utility) is always incentivizable in our framework, and is also often the only question that can be elicited under incomplete observability \citep[see][]{ChassangPadróIMiquelSnowberg12}. In that sense, both issues may lead to similar prescriptions regarding what can be elicited.




\subsection{Endogenous beliefs}

Another important assumption of our framework is that the subject has a fixed belief that is not affected by his choice of action. Consider, for example, the moral hazard issue described in \citet{chambers2018dynamic} in which a student who is asked about the probability of passing a test can influence that probability through his behavior. If the state is defined as passing or failing, the student's belief is endogenous.

In many cases, a problem with endogenous beliefs can be converted into one with fixed beliefs by redefining variables. In the above example, the state could instead be the solution to the test. 
If the outcome of the test is a deterministic function of the state and the student's behavior, the likelihood of passing corresponds to a question profile $X$ in our model. More generally, one can redefine the state as a function from the set of actions into probability distributions over outcomes. If this reformulation satisfies the observability assumptions of our model, then our results apply.



\subsection{Multi-dimensional questions}\label{sec: multiple questions}
In our model, we assume that the researcher can only ask a single question. Here, we show how our methods can be extended to multiple questions. For simplicity, we focus on the case of two questions; the logic extends directly to more than two questions. 

We say that question profiles $X,Y:A\longrightarrow \mathbb{R}^\Theta$ are \emph{jointly incentivizable} if there exists an elicitation method  $V:\mathbb{R}^2\times A\times\Theta\longrightarrow\left[0,1\right]$  such that, for every $p\in \Delta(\Theta)$,
\[
\arg\max_{a,r,s}\E_{p}V\left(r,s,a,\theta\right)
=\left\{ \left(\E_{p}X\left(a;\theta\right),\E_{p}Y\left(a;\theta\right),a\right):a\in \arg\max_{b\in A}\E_{p}u\left(b;\cdot\right)\right\} .
\]

If $X$ and $Y$ are both aligned with $u$, it is straightforward to extend Proposition \ref{Observation 0} to show that they are jointly incentivizable.

For necessary conditions, our key result---Lemma \ref{lem: adjacency}---extends as follows.

\begin{lem}\label{lem: adjacency multiple}
    Suppose $X$ and $Y$ are jointly incentivizable. If actions $a$ and $b$ are adjacent, then there are $\rho_X,\rho_Y$ and $\sigma_x^y$ for $x,y=X,Y$, not all equal to $0$, such that 
\begin{align*}
\bar{\vt{X}}\left(b\right)&
=\rho_X\left(\bar{\vt{u}}(b)-\bar{\vt{u}}(a)\right)
+\sigma_X^X\bar{\vt{X}}\left(a\right)
+\sigma_X^Y\bar{\vt{Y}}\left(a\right)\\
\text{and} \qquad \bar{\vt{Y}}\left(b\right)&
=\rho_Y\left(\bar{\vt{u}}(b)-\bar{\vt{u}}(a)\right)
+\sigma_Y^X\bar{\vt{X}}\left(a\right)
+\sigma_Y^Y\bar{\vt{Y}}\left(a\right).
\end{align*}
If $\bar{\vt{X}}(a)$ or $\bar{\vt{X}}(b)$ is collinear with $\bar{\vt{u}}(b)-\bar{\vt{u}}(a)$, then we can take $\rho\neq0$.
\end{lem}

The proof, which we omit, follows the same reasoning as that of Lemma \ref{lem: adjacency}. We leave to future research the details of how to use this lemma to identify precise conditions for joint incentivizability.\footnote{\citet{chen2026ask} study a closely related problem in which the researcher can ask the subject multiple elicitation questions.}

\appendix

\section{Proofs for Section \ref{sec: sufficient conditions}}

\begin{lem}\label{lem: simple questions}
For any $\vt{d}\in\R^\Theta$, $\vt{X}\left(a\right)=\vt{u}\left(a\right)+\vt{d}$
and $\vt{X}(a)=\vt{d}$ are incentivizable.
\end{lem}

\begin{proof}
Let $L,M\in\R$ be such that, for all $a$ and $\theta$, $L<u(a;\theta)+d(\theta)<M$ and $L<d(\theta)<M$.  If $\vt{X}\left(a\right)=\vt{u}\left(a\right)+\vt{d}$, let 
\[
V\left(r,a,\theta\right)=\intop_{L}^{r}X\left(a;\theta\right)dx+\intop_{r}^{M}xdx-\frac{M^2}{2}=\left(u\left(a;\theta\right)+d\left(\theta\right)\right)(r-L)-\frac{1}{2}r^{2},
\]
where $d(\theta)$ is the $\theta$-coordinate of $\vt{d}$. For this $V$, simple calculations show that the optimal choice of $r$ given $a$ is $\E_{p}X\left(a;\theta\right)$. Since the optimal $r$ is greater than $L$, the $V$-optimal choice of action $a$ is the same as the $u$-optimal one.  A similar argument applies if $\vt{X}(a)=\vt{d}$, in which case we let $V(r,a,\theta)=u\left(a;\theta\right)+d\left(\theta\right)(r-L)-\frac{1}{2}r^{2}$.
\end{proof}
    
\begin{proof}[Proof of Proposition \ref{Observation 0}]
Suppose $V$ incentivizes a question profile $X$. Let $\vt{Y}\left(a\right)=\gamma (a)\vt{X}\left(a\right)+\kappa(a) \vt{1}$ for some $\gamma,\kappa:A\longrightarrow\R$. Letting $W\left(r,a,\theta\right)=V\left(\frac{1}{\gamma (a)}\left(r-\kappa (a)\right),a,\theta\right)$, it is straightforward to verify that $W$ incentivizes $\vt{Y}$. The result now follows from Lemma \ref{lem: simple questions}.
\end{proof}

\section{Proofs for Section \ref{sec: necessary conditions}}

\begin{proof}[Proof of Lemma \ref{lem: adjacency}]
We first prove the latter formulation of the result (i.e., that $\bar{\vt{X}}\left(b\right)=\rho\Delta_a^b+\sigma\bar{\vt{X}}\left(a\right)$), and then prove the equivalence between the two formulations.


Let $V$ be an elicitation method that incentivizes $X$. Consider two beliefs $p_{0}$ and $p_{1}$ such that (i) actions $a$ and $b$ are both $u$-optimal, 
i.e., $a,b\in\arg\max_{a'} \E_{p_{k}} [u(a';\cdot)]$ for $k=0,1$, and (ii) the question $\vt{X}(a)$ attains the same value at $p_0$ and $p_1$, i.e., $\E_{p_{0}}[X(a;\cdot)]=\E_{p_{1}}[X(a;\cdot)]$. 
Letting $r=\E_{p_{0}}[X(a;\cdot)]$, it follows that $\left(a,r\right)$ is optimal given $V$ at all $p_{\alpha}=\alpha p_{1}+\left(1-\alpha\right)p_{0}$ for $\alpha\in [0,1]$. 
The value of information $V^*(p_\alpha)=\E_{p_{\alpha}} [V\left(r,a,\cdot\right)]$ is therefore affine in $\alpha$. 
Since $b$ is also a $u$-optimal action at each $p_\alpha$, the optimal expected payoff $\max_{s}\E_{p_{\alpha}}[V\left(s,b,\cdot\right)]$ must be affine in $\alpha$ as well. Therefore, as explained in Section \ref{sec: necessary conditions}, there exists some $r'$ such that, for each
$\alpha$, $r'\in\arg\max_{s}\E_{p_{\alpha}} [V\left(s,b,\cdot\right)$]. In particular, $\E_{p_{k}}[X(b;\cdot)]=r'$ for each $k=0,1$.


The adjacency of actions $a$ and $b$ implies that we can find an interior $p_0$ for which the set of $u$-optimal actions is $\{a,b\}$. Consider the linear subspace
\begin{align*}
  D&=\{\vt{d}\in\R^\Theta:\vt{d}\cdot \vt{1}=\vt{d}\cdot (\vt{u}(a)-\vt{u}(b))=\vt{d}\cdot\vt{X}(a)=0\}\\&=(\mathrm{span}\{\vt{1}, \vt{u}(a)-\vt{u}(b),\vt{X}(a)\})^\perp.  
\end{align*}
For each $\vt{d}\in D$, there is a sufficiently small $\varepsilon>0$ for which $p_1=p_0+\varepsilon\vt{d}$ (i) is a well-defined belief, (ii) $a$ and $b$ are both $u$-optimal at $p_1$, and (iii) $\vt{X}(a)$ attains the same value at $p_0$ as at $p_1$; each of these properties follows from one of the orthogonality conditions defining $D$. In particular, $p_0$ and $p_1$ satisfy properties (i) and (ii) from the previous paragraph. As a result, the above argument applies and shows that $\varepsilon\vt{d}=p_1-p_0\perp \vt{X}(b)$. Therefore, $\vt{X}(b)\perp D$.

By a standard linear algebra argument, it follows that \[\vt{X}(b)\in D^\perp=\mathrm{span}\left(\mathbf{1},\vt{u}(a)-\vt{u}(b),\vt{X}(a)\right).\]
Noting that, for any vector $\vt{v}$, $\bar{\vt{v}}$ differs from $\vt{v}$ by a scalar multiple of $\vt{1}$, and that $\bar{\vt{v}}\perp \vt{1}$, we obtain
$\bar{\vt{X}}(b)=\rho\left(\bar{\vt{u}}(a)-\bar{\vt{u}}(b)\right)+\sigma\bar{\vt{X}}(a)$
for some $\rho,\sigma\in\R$.

If $\bar{\vt{X}}(b)$ is not collinear with $\bar{\vt{u}}(b)-\bar{\vt{u}}(a)$, then we must have $\sigma\neq 0$, as needed. Otherwise, switching the roles of $a$ and $b$ in the preceding argument yields that $\bar{\vt{X}}(a)$ is also collinear with
$\bar{\vt{u}}(b)-\bar{\vt{u}}(a)$, and therefore one can take $\rho$ and $\sigma$ to be nonzero.

We now show that the above conclusion is equivalent to alignment on pairs of adjacent actions. Note first that $X$ is nontrivially aligned with $u$ on a set of actions $B$ if and only if $\bar{\vt{X}}(a)\equiv_B \gamma(a) \bar{\vt{u}}(a) + \vt{d}$ for some $\gamma(a)\in\R$ and $\vt{d}\in\R^\Theta$, and is trivially aligned with $u$ on $B$ if and only if $\bar{\vt{X}}(a)\equiv_B \gamma(a)\vt{d}$ for some $\gamma(a)\in\R$ and $\vt{d}\in\R^\Theta$. In either case, eliminating $\vt{d}$ leads to an equation of the form $\bar{\vt{X}}\left(b\right)=\rho\Delta_a^b+\sigma\bar{\vt{X}}\left(a\right)$ with $\sigma\neq 0$.

For the converse, if $\rho=0$, then $X$ is trivially aligned with $u$ on $\{a,b\}$. If $\rho\neq0$, then the alignment is nontrivial, with
$
\vt{d}=\frac{1}{\rho}\bar{\vt{X}}(b)-\bar{\vt{u}}(b)
=\frac{\sigma}{\rho}\bar{\vt{X}}(a)-\bar{\vt{u}}(a),
$
$\gamma(a)=\rho/\sigma$, and $\gamma(b)=\rho$.
\end{proof}
 
\section{Proofs for Section \ref{sec: trees}}


\begin{lem}
\label{lem:Splitting action} Let $a$ be a splitting action and $B_0$ and $B_1$ be sets not containing $a$ such that $\left\{B_0\cup\{a\},B_1\cup\{a\}\right\}$ is a splitting collection. 
Then, for any pair of beliefs $p_0$ and $p_1$ such that, for each $i$, some action $b_{i}\in B_{i}$
is $u$-optimal at $p_{i}$, there is a convex combination $p=\alpha p_{1}+\left(1-\alpha\right)p_{0}$
such that $a$ is $u$-optimal at $p$.
\end{lem}

\begin{proof}
Suppose not. Then all actions that are $u$-optimal at convex combinations of $p_{0}$ and $p_{1}$ must be either from $B_{0}$ or $B_{1}$. Hence, for some such convex combination $p'$, the set of $u$-optimal actions contains at least one element from each of $B_0$ and $B_1$. Then in any neighbourhood of $p'$, there must be a belief at which there is exactly one member of each $B_i$ that is $u$-optimal, 
contradicting the fact that no action in $B_0$ is adjacent to any action in $B_1$.
\end{proof}

\begin{proof}[Proof of Proposition \ref{prop: piecewise aligned sufficient}]
We first derive explicit formulas for an elicitation method that incentivizes a question $X$ aligned with $u$. 
When 
$X$ is nontrivially aligned with $u$, using the BDM construction from the proof of Proposition \ref{Observation 0} gives
\begin{align}
\label{eq:BDM expression}
V^{BDM}\left(r,a,\theta\right) & =\left(u\left(a;\theta\right)+d\left(\theta\right)\right)\left(\frac{1}{\gamma (a)}\left(r-\kappa(a)\right)-L\right)-\frac{1}{2}\left(\frac{1}{\gamma(a)}\left(r-\kappa(a)\right)\right)^{2}\\
 \nonumber & =\frac{1}{\gamma(a)^{2}}\left[\left(X\left(a;\theta\right)-\kappa(a)\right)\left(r-\kappa(a)-\gamma(a)L\right)-\frac{1}{2}\left(r-\kappa (a)\right)^{2}\right]\\
\nonumber  & =\frac{1}{\gamma(a)^{2}}\left(X\left(a;\theta\right)-\frac{1}{2}r-\frac{1}{2}\kappa(a)\right)\left(r-\kappa(a)\right) - \frac{1}{\gamma(a)} (X(a;\theta)-\kappa(a))L.
\end{align}
Similarly, if 
$X$ is trivially aligned with $u$, adding $u(a;\theta)$ to either of the last two lines gives $V^{BDM}$.

When 
$X$ is nontrivially aligned with $u$, the expected payoff of a subject who chooses action $a$ and then chooses $r$ optimally is equal to 
\begin{align}
\label{eq:expected BDM payoff}
\max_{r}\biggl[\E_{p}\left[u\left(a;\cdot\right)+d\left(\cdot\right)\right]&\left(\frac{1}{\gamma(a)}\left(r-\kappa(a)\right)-L\right)-\frac{1}{2}\left(\frac{1}{\gamma(a)}\left(r-\kappa(a)\right)\right)^{2}\biggr] \\
\nonumber =&\max_{x}\left[\E_{p}\left[u\left(a;\cdot\right)+d\left(\cdot\right)\right] (x-L)-\frac{1}{2}x^{2}\right]\\
\nonumber =&\frac{1}{2} \left(\E_{p}\left[u\left(a;\cdot\right)+d\left(\cdot\right)\right]\right)^2 - L\E_p \left[u\left(a;\cdot\right)+d\left(\cdot\right)\right].
\end{align}
Similarly, the expected payoff is equal to $\E_p[u(a;\cdot)] + \left(\E_p [d(\cdot)]\right)^2/2-L\E_p[d(\cdot)]$ if $X$ is trivially aligned with $u$.

To keep the notation simple, we present the argument only for the case in which the splitting collection contains two elements, $A_0$ and $A_1$, with splitting action $a_0$. Extending the argument to the general case is straightforward.

We construct an elicitation method on each $A_i$, then show that they agree on $a_0$ and therefore give rise to well-defined elicitation method on the full action set, $A$.

First suppose $X$ is nontrivially aligned with $u$ on each $A_i$. For each $i=0,1$, let $\left(\gamma_{i},\kappa_{i},\vt{d}_{i}\right)$ be such that $\vt{X}(a)\equiv_{A_i} \gamma_i(a) (\vt{u}(a)+\vt{d})+\kappa(a)\vt{1}$. 
Define the elicitation methods
\begin{multline}\label{eq:V-split}
V_{i}\left(r,a,\theta;w_{i},\omega_{i}\right) 
 =\left(\frac{\gamma_i(a_{0})}{\gamma_i(a)}\right)^{2}\left(X\left(a;\theta\right)-\frac{1}{2}r-\frac{1}{2}\kappa_i(a)\right)\left(r-\kappa_i (a)\right)\\
 +w_{i}\left(\theta\right)-\omega_{i} - \frac{\gamma_i(a_0)^2}{\gamma_i(a)} (X(a;\theta)-\kappa_i(a))L, 
\end{multline}
where $w_{i}\in\R^{\Theta}$ and $\omega_{i}\in\R$. Note that this expression differs from the expression for $V^{BDM}$ in \eqref{eq:BDM expression} only by multiplication by a positive constant and addition of a function that depends only on the state. As neither of these changes affects the optimal choices of $a$ and $r$, it follows from the argument in the proof of Proposition \ref{Observation 0} that $V_i$ incentivizes $X$ on $A_i$ (i.e., in the decision problem with actions restricted to $A_i$).

Notice that
\begin{align*}
 V_{1}\big(&r,a_{0},\theta;w_{1},\omega_{1}\big)-V_{0}\left(r,a_{0},\theta;w_{0},\omega_{0}\right)\\
&= \left(X\left(a_{0};\theta\right)-\frac{1}{2}r-\frac{1}{2}\kappa_1(a_{0})\right)\left(r-\kappa_1(a_{0})\right)-\left(X\left(a_{0};\theta\right)-\frac{1}{2}r-\frac{1}{2}\kappa_0(a_{0})\right)\left(r-\kappa_0 (a_{0})\right)\\
 &\hspace{16pt}+\left(w_{1}\left(\theta\right)-w_{0}\left(\theta\right)\right)-\left(\omega_{1}-\omega_{0}\right)-\gamma_1(a_0)(X(a_0;\theta)-\kappa_1(a_0))L + \gamma_0(a_0)(X(a_0;\theta)-\kappa_0(a_0))L\\
&=  X\left(a_{0};\theta\right)\left(\kappa_0(a_{0})-\kappa_1(a_{0})+L\gamma_0(a_0)-L\gamma_1(a_0)\right)+\frac{1}{2}\left(\left(\kappa_1(a_{0})\right)^{2}-\left(\kappa_0(a_{0})\right)^{2}\right)\\
&\hspace{80pt}+\left(w_{1}\left(\theta\right)-w_{0}\left(\theta\right)\right)-\left(\omega_{1}-\omega_{0}\right) + (\gamma_1(a_0)\kappa_1(a_0)-\gamma_0(a_0)\kappa_0(a_0))L.
\end{align*}

Given any $w_{0}$ and $\omega_{0}$, let
\[
w_1(\theta) = w_0(\theta) - X\left(a_{0};\theta\right)\left(\kappa_0(a_{0})-\kappa_1(a_{0})+L\gamma_0(a_0)-L\gamma_1(a_0)\right)
\]
for each $\theta$, and
\[
\omega_1 = \frac{1}{2}\left(\kappa_1(a_{0})^{2}-\kappa_0(a_{0})^{2}\right)+\omega_0 + (\gamma_1(a_0)\kappa_1(a_0)-\gamma_0(a_0)\kappa_0(a_0))L.
\]
Then $V_{1}\big(r,a_{0},\theta;w_{1},\omega_{1}\big)=V_{0}\left(r,a_{0},\theta;w_{0},\omega_{0}\right)$ for all $r$ and $\theta$.\footnote{In the general case with splitting collection $\{A_1,\dots,A_k\}$, one can recursively define each $w_{i+1}$ and $\omega_{i+1}$ given $w_i$ and $\omega_i$ in an analogous fashion.}


Let $V\left(r,a,\theta\right)=V_{i}\left(r,a,\theta;w_i,\omega_i\right)$ for each $a\in A_{i}$. Because $V_0$ and $V_1$ agree for action $a_{0}$, $V$ is well defined.

In case $X$ is trivially aligned with $u$ on some $A_i$, we add $\gamma_i(a_0)^2 u(a;\theta)$ to the expression for $V_i(r,a,\theta;w_i,\omega_i)$ in \eqref{eq:V-split} and adjust the definition of $w_1(\theta)$ accordingly to include any such additional $u(a_0;\theta)$ terms.

To verify that $V$ incentivizes $X$, it suffices to show that at any belief $p$ at which no action in $A_i$ is $u$-optimal, $\arg\max_a \max_r \E_p [V(r,a,\cdot)]\subseteq A_j$, where $j\neq i$. Without loss of generality, take $i=1$ and $j=0$, and let $p_0$ denote such a belief. 

Note first that, by \eqref{eq:expected BDM payoff}, the expected value from choosing an action $a\in A_i$ followed by an optimal choice of $r$ is equal to
\begin{multline}
\max_{r}\E_{p}[V_{i}(r,a,\cdot)]=\gamma_i(a_{0})^{2}\frac{1}{2}\left(\E_{p}[u\left(a;\cdot\right)]+\E_{p}[d\left(\cdot\right)]\right)^{2}\\
+\E_{p}[w_{i}(\cdot)]-\omega_{i}- L\gamma_i(a_0)^2\E_p \left[u\left(a;\cdot\right)+d\left(\cdot\right)\right].\label{eq:V_i exp payoff}
\end{multline}
It follows that the expected value from an action $b\in A_i$ is at least as large as that from an action $a\in A_i$ if and only if
\[
\frac{1}{2}\left( \E_p[u(b;\cdot)]+\E_p[d(\cdot)] \right)^2 - L \E_p[u(b;\cdot)] \geq \frac{1}{2}\left( \E_p[u(a;\cdot)]+\E_p[d(\cdot)] \right)^2 - L \E_p[u(a;\cdot)],
\]
which holds if and only if $\E_p[u(b;\cdot)]\geq \E_p[u(a;\cdot)]$ since the function $(x+y)^2/2 -Lx$ is increasing in $x$ for $x+y>L$ and $L$ satisfies $\E_p[u(a';\cdot)]+\E_p[d(\cdot)]>L$ for all actions $a'$.

\sloppy
Suppose for contradiction that there is some $b\in A_{1}\backslash\left\{ a_{0}\right\} $
such that $b\in\arg\max_{a}\max_{r}\E_{p_{0}}V(r,a,\cdot)$. It follows that 
\[
\max_{r}\E_{p_{0}}V_{1}(r,b,\cdot)=\max_{r}\E_{p_{0}}V(r,b,\cdot)\geq 
\max_{r}\E_{p_{0}}V(r,a_0,\cdot)=\max_{r}\E_{p_{0}}V_{1}(r,a_{0},\cdot).
\]
By the observation in the preceding paragraph,
$\E_{p_{0}}u\left(b;\cdot\right) \geq\E_{p_0}u\left(a_{0};\cdot\right).$

\fussy



Let $p_{1}$ be a belief at which $b$ is strictly $u$-optimal, i.e.,
$\left\{ b\right\} =\arg\max_{a}\E_{p_{1}}[u\left(a;\cdot\right)]$. By
Lemma \ref{lem:Splitting action}, action $a_{0}$ must be $u$-optimal at some convex combination $p=\alpha p_{1}+\left(1-\alpha\right)p_{0}$.
At the same time, the above inequalities imply that 
\begin{align*}
\E_{p}[u\left(a_{0};\cdot\right)] & =\alpha\E_{p_{1}}[u\left(a_{0};\cdot\right)]+\left(1-\alpha\right)\E_{p_{0}}[u\left(a_{0};\cdot\right)]\\
 & <\alpha\E_{p_{1}}[u\left(b;\cdot\right)]+\left(1-\alpha\right)\E_{p_{0}}[u\left(b;\cdot\right)]
 =\E_{p}[u\left(b;\cdot\right)],
\end{align*}
contradicting the $u$-optimality of $a_0$ at $p$. 
\end{proof}

\section{Proofs for Section \ref{sec: complete} and extension of Theorem \ref{prop:complete adjacency regular}}\label{app: complete}

\begin{proof}[Proof of Lemma \ref{lem:adjacency cycles}]
Without loss of generality, let $a_0\in C$ be such that either $\bar{\vt{X}}(a_0)= 0$ or $\bar{\vt{X}}(a_0)\notin V_C$. (If no such $a_0$ exists, the lemma holds trivially.)

By Lemma \ref{lem: adjacency}, for each $i=1,\dots,n$, we have 
$\bar{\vt{X}}\left(a_{i}\right)
=\rho_{i}\Delta_{a_{i-1}}^{a_{i}}+\sigma_{i}\bar{\vt{X}}\left(a_{i-1}\right)$
for some $\sigma_i\neq0$ and some $\rho_i$. Iterating these equations gives 
\[
\bar{\vt{X}}\left(a_{0}\right)=\bar{\vt{X}}\left(a_{n}\right) =\sum_{i=1}^{n}\Gamma_{i}\rho_{i}\Delta_{a_{i-1}}^{a_{i}}+\Gamma_{0}\bar{\vt{X}}\left(a_{0}\right),
\]
where $\Gamma_{i}=\sigma_{n}\cdots\sigma_{i+1}$ and $\Gamma_{n}=1$.
Because $\Delta_{a_{0}}^{a_{1}}+\cdots+\Delta_{a_{n-1}}^{a_{n}}=0$,
we get
\begin{align*}
\sum_{i=1}^{n-1}\left(\Gamma_{i}\rho_{i}-\rho_{n}\right)\Delta_{a_{i-1}}^{a_{i}}+\left(\Gamma_{0}-1\right)\bar{\vt{X}}\left(a_{0}\right)=0.
\end{align*}
Since either $\bar{\vt{X}}(a_0)= 0$ or $\bar{\vt{X}}(a_0)\notin V_C$, it follows that $\sum_{i=1}^{n-1}\left(\Gamma_{i}\rho_{i}-\rho_{n}\right)\Delta_{a_{i-1}}^{a_{i}}=0.$ Internal independence implies that $\Gamma_{i}\rho_{i}=\rho_{n}$ for each $i=1,\dots,n-1$. 

If $\rho_n=0$, then $\rho_i=0$ for each $i$ since $\Gamma_i\neq 0$. In this case, all $\bar{\vt{X}}(a_i)$ are collinear, which implies that $X$ is trivially aligned with $u$ on $C$. 

For the case of $\rho_n\neq 0$, first note that, by the same iteration as above, for each $k$,
\begin{align*}
\bar{\vt{X}}(a_k) &=\sum_{i=1}^{k-1}\rho_i \frac{\Gamma_i}{\Gamma_k}\Delta_{a_{i-1}}^{a_i} + \frac{\Gamma_0}{\Gamma_k}\bar{\vt{X}}(a_0)
= \frac{\rho_n}{\Gamma_k}\sum_{i=1}^{k-1}\Delta_{a_{i-1}}^{a_i} + \frac{\Gamma_0}{\Gamma_k}\bar{\vt{X}}(a_0)\\
&= \frac{\rho_n}{\Gamma_k}(\bar{\vt{u}}(a_k)-\bar{\vt{u}}(a_0)) + \frac{\Gamma_0}{\Gamma_k}\bar{\vt{X}}(a_0).
\end{align*}
Letting $\gamma(a_k)=\rho_k=\rho_{n}/\Gamma_k\neq0$ and $\vt{d}=-\bar{\vt{u}}\left(a_{0}\right)+\rho_{n}^{-1}\Gamma_{0}\bar{\vt{X}}\left(a_{0}\right)$, we have $\bar{\vt{X}}(a_k)=\gamma(a_k)\left(\bar{\vt{u}}(a_k) +\vt{d})\right)$ for all $k$, and thus $X$ is (nontrivially) aligned with $u$ on $C$.
%
\end{proof}

\begin{proof}[Proof of Lemma \ref{lem:combine sets}]
Alignment with $u$ on $B$ implies that there exist $\gamma_{0}^{B},\gamma_{1}^{B}\in\R$ and $\vt{d}^{B}\in\R^{\Theta}$
such that 
\[
\gamma_{0}^{B}\bar{\vt{X}}(a_{0})-\bar{\vt{u}}(a_{0})
=\gamma_{1}^{B}\bar{\vt{X}}(a_{1})-\bar{\vt{u}}(a_{1})=\vt{d}^{B}.
\]
Together with the analogous equation for $D$, we obtain
\[
\gamma_{0}^{B}\bar{\vt{X}}(a_{0})-\gamma_{1}^{B}\bar{\vt{X}}(a_{1}) =\bar{\vt{u}}(a_{0})-\bar{\vt{u}}(a_{1})
\text{ and }\gamma_{0}^{D}\bar{\vt{X}}(a_{0})-\gamma_{1}^{D}\bar{\vt{X}}(a_{1})  =\bar{\vt{u}}(a_{0})-\bar{\vt{u}}(a_{1}).
\]
Subtracting one equation from the other leads to
\[
\left(\gamma_{0}^{B}-\gamma_{0}^{D}\right)\bar{\vt{X}}(a_{0})-\left(\gamma_{1}^{B}-\gamma_{1}^{D}\right)\bar{\vt{X}}(a_{1})=0.
\]
Since $\bar{\vt{X}}(a_0)$ and  $\bar{\vt{X}}(a_1)$ are linearly independent, it must be that $\gamma_{k}^{B}=\gamma_{k}^{D}$
for each $k=0,1$. This in turn implies that $\vt{d}^{B}=\vt{d}^{D}=\vt{d}$.

All that remains is to show that for any other action $a\in B\cup D$, the corresponding parameters $\gamma^B_a$ and $\gamma^D_a$ are equal. Since $\bar{\vt{X}}(a)$ cannot be collinear with both $\bar{\vt{X}}(a_0)$ and $\bar{\vt{X}}(a_1)$, we can repeat the argument replacing one of $a_0$ or $a_1$ with $a$ to obtain $\gamma^B_a=\gamma^D_a$.
\end{proof}

\subsection{General necessary conditions}

A set of actions $B\subseteq A$ is \emph{cycle-rich} if it contains at least four elements and, for any proper subset $B^\prime\subset B$ with at least three elements, there exists $a\in B\setminus B^\prime$ such that 
\[
\bigcap \left\{V_C:C\text{ is internally independent},
a\in C,\text{ and }
|C\cap B^\prime|\geq2
\right\}=\{\vt{0}\}.
\]
The intersection above goes over all internally independent cycles that contain $a$ and at least two elements of $B^\prime$.

\begin{example}\label{ex: cycle rich}
    Consider a variant of Example \ref{ex: complete adjacency thm} in which there is an additional safe action. Thus $A=\Theta\cup\{a_s\}$ with
    $u(a;\theta)= r_\theta \cdot \mathbb{1}\{a=\theta\} + s\cdot \mathbb{1}\{a=a_s\}.$
    Suppose in addition that $\Theta=\{\theta_0,\theta_1,\theta_2,\theta_3\}$, with $r_\theta=1/2$ for $\theta=\theta_0,\theta_1$ and $r_\theta=1$ for $\theta=\theta_2,\theta_3$. Let $s=3/10$. Then the adjacency graph 
    is incomplete since actions $\theta_0$ and $\theta_1$ are not adjacent, but $A$ is cycle-rich.
\end{example}

\begin{thm}\label{thm: aligned representation necessary}
Suppose $B\subseteq A$ is cycle-rich. If $X$ is incentivizable, then it is aligned with $u$ on $B$.
\end{thm}

\begin{proof}
We begin with the following observation. 

\begin{lem} \label{lem:trivial or colinear Delta}
    If $X$ is nontrivially aligned with $u$ on $B$, then 
    for any $a,b\in B$, if $\bar{\vt{X}}(a)$ and $\bar{\vt{X}}(b)$ are collinear, they are also collinear with $\Delta_a^b$. 
\end{lem}

\begin{proof}
Let $a,b\in B$ be such that $\bar{\vt{X}}(a)$ and $\bar{\vt{X}}(b)$ are collinear and let $\gamma(\cdot)\neq 0$ and $\vt{d}$ be such that $\bar{\vt{X}}(a)=\gamma(a) \left( \bar{\vt{u}}(a) +\vt{d}\right)$ and $\bar{\vt{X}}(b)=\gamma(b) \left( \bar{\vt{u}}(b) +\vt{d}\right)$. By the collinearity assumption, there exists $\alpha\neq0$ such that $\bar{\vt{u}}(a)+\vt{d}=\alpha\left(\bar{\vt{u}}(b)+\vt{d}\right).$ Because $\bar{\vt{u}}(a)\neq\bar{\vt{u}}(b)$, it must be that $\alpha\neq1$. It follows that $\vt{d}=\frac{1}{1-\alpha}\left(\alpha\bar{\vt{u}}(b)-\bar{\vt{u}}(a)\right)$ and 
\[
\frac{1}{\gamma(a)}\bar{\vt{X}}(a)
=\bar{\vt{u}}(a)+\frac{1}{1-\alpha}\left(\alpha\bar{\vt{u}}(b)-\bar{\vt{u}}(a)\right)
=\frac{\alpha}{1-\alpha}\left(\bar{\vt{u}}(b)-\bar{\vt{u}}(a)\right),
\]
as needed.
\end{proof}


    Cycle-richness implies that there is some $a\in B$ and a collection of internally independent cycles $\tilde{C}$ with $a\in \tilde{C}\subseteq B$ for which the intersection of the spaces $V_{\tilde{C}}$ is $\{\vt{0}\}$. Thus either $\bar{\vt{X}}(a)=\vt{0}$ or $\bar{\vt{X}}(a)\notin V_C$ for some such cycle $C$. By Lemma \ref{lem:adjacency cycles}, $X$ is aligned with $u$ on $C$.

    Let $B^\prime\subseteq B$ be a subset of $B$ of maximal cardinality on which $X$ is aligned with $u$. By the above argument, $B^\prime$ has at least three elements. Suppose for contradiction that $B^\prime\neq B$. By the same argument as in the preceding paragraph, cycle-richness implies that there exists $a\in B\setminus B^\prime$ and a cycle $C$ containing $a$ such that $|C\cap B^\prime|\geq2$ and either $\bar{\vt{X}}(a)=\vt{0}$ or $\bar{\vt{X}}(a)\notin V_C$. By Lemma \ref{lem:adjacency cycles}, $X$ is aligned with $u$ on $C$. 
    
    If there exists a pair of distinct actions $a_0,a_1\in C\cap B^\prime$ such that $\bar{\vt{X}}(a_0)$ and $\bar{\vt{X}}(a_1)$ are not collinear, Lemma \ref{lem:combine sets} implies that $X$ is aligned with $u$ on $C\cup B^\prime$, contradicting the maximality of $B^\prime$. 
    
    From now on, suppose that $a_0,a_1\in C\cap B^\prime$ are distinct actions such that $\bar{\vt{X}}(a_0)$ and $\bar{\vt{X}}(a_1)$ are collinear.
    
    If $\bar{\vt{X}}(a_0)$ and $\bar{\vt{X}}(a_1)$ are not collinear with $\Delta_{a_0}^{a_1}$, then Lemma  \ref{lem:trivial or colinear Delta} implies that the alignment with $u$ on $C$ and that on $B^\prime$ must both be trivial, which further implies that, for all $b\in C\cup B^\prime$, $\bar{\vt{X}}(b)$ is collinear with $\bar{\vt{X}}(a_0)$. Thus, $X$ is trivially aligned with $u$ on $C\cup B^\prime$, contradicting the maximality of $B^\prime$.

    If $\bar{\vt{X}}(a_0)$ and $\bar{\vt{X}}(a_1)$ are collinear with $\Delta_{a_0}^{a_1}$, then, by Lemma \ref{lem: adjacency}, $\bar{\vt{X}}(a)\in \text{span}(\Delta_a^{a_0}, \Delta_{a_0}^{a_1})=V_C$. The choice of cycle $C$ implies that $\bar{\vt{X}}(a)=0$. Another application of Lemma \ref{lem: adjacency} shows that $\bar{\vt{X}}(a_0)$ is collinear with $\Delta_a^{a_0}$, which contradicts collinearity with $\Delta_{a_0}^{a_1}$ due to the independence assumptions.  
\end{proof}

\begin{proof}[Proof of Theorem \ref{prop:complete adjacency regular}]
    It suffices to show that the set of all actions is cycle-rich; the result then follows from Theorem \ref{thm: aligned representation necessary}.

    Take any proper subset $B\subset A$ with at least three actions $b_0,b_1,b_2\in B$ and let $a\in A\setminus B$. Consider cycles $C_i$ with vertices $B\setminus \{b_i\} \cup \{a\}$. Then, $V_{C_i}= \text{span} \{ \Delta_a^{b_j}:j\neq i \}$. The independence assumption implies that $\bigcap_i V_{C_i}=\{\vt{0}\}$.
\end{proof}


It is straightforward to extend Theorem \ref{thm: aligned representation necessary} to problems in which there is a splitting collection $\{A_0,\dots,A_k\}$ such that, for each $l$, either $A_l$ is cycle-rich or it contains exactly two elements. In that case, only questions piecewise aligned with $u$ are incentivizable.



\section{Proof of Theorem \ref{prop:product problems}}\label{sec: Proof of Product Problems}

This appendix is divided into the following subsections. Section \ref{sec:Decomposition} shows that each incentivizable question in a product problem can be decomposed into linearly independent vectors that correspond to different tasks. We use this decomposition together with Lemma \ref{lem: adjacency} to derive restrictions on questions for adjacent actions. Subsections  \ref{sec:Exact cycles} to \ref{sec:Exact cycles - Free transitions} show that all adjacency cycles are exact. Section \ref{sec:Exact cycles - Tools} develops useful tools, and Sections \ref{sec:Exact cycles - No free transitions} and \ref{sec:Exact cycles - Free transitions} deal with different classes of cycles. 
Section \ref{sec: Product proof conclusion} concludes the proof. 

\subsection{Decomposition}\label{sec:Decomposition}

Assume throughout that $X$ is incentivizable.

For each $i$ and each $t_{i}\in\Theta_{i}$, let
$\vt{e}_{i}\left(t_{i}\right)\in\R^{\Theta}$ be the vector such that, for
each $\theta\in\Theta$, $\vt{e}_{i}\left(\theta|t_{i}\right)=\mathbf{1}\left\{ \theta_{i}=t_{i}\right\} $.  
Let $E_i=\text{span}\left\{\vt{\Delta}^{a_{i}}_{b_{i}}:a_i,b_i\in A_i\right\}\subseteq \mathbb{R}^\Theta$ be the subspace of $\mathbb{R}^\Theta$ spanned bu the utility difference vectors (see also footnote \ref{ft:lin independence}). Because payoffs in task $i$ depend only on the state coordinate $i$, subspaces $E_i$ are linearly independent. Let $E_0$ be a complementary space to the sum of $E_1$ through $E_I$. 




For each $a$, $\bar{\vt{X}}(a)=\sum_{i=0}^I\vt{w}_i(a)$ admits a unique decomposition
to $\vt{w}_i(a)\in E_i$ for all $i$. Note that the vectors $\vt{w}_0(a),\vt{w}_1(a),\dots,\vt{w}_I(a)$ are linearly independent.

Take any distinct $a$ and $b$ such that $a_{-i}=b_{-i}$ for some $i$. Then $a$ and $b$ are adjacent, and, by Lemma \ref{lem: adjacency}, there are $x(a,b)\neq0$ and $y(a,b)$ such that $\bar{\vt{X}}(a)=x(a,b)\bar{\vt{X}}(b)+y(a,b)\Delta_{b}^{a}$. Using the above decomposition, we have
\begin{align*}
 0
=&\left[\vt{w}_{i}(a)-x(a,b)\vt{w}_{i}(b)-y(a,b)\Delta_{b_{i}}^{a_{i}}\right]+\sum_{j\neq i}\left[\vt{w}_{j}(a)-x(a,b)\vt{w}_{j}(b)\right]
+\left[\vt{w}_{0}(a)-x(a,b)\vt{w}_{0}(b)\right].
\end{align*}
The proof of Lemma \ref{lem: adjacency} shows that $x(a,b)$ is uniquely defined if and only if $\bar{\vt{X}}(a)\neq0$ and $\bar{\vt{X}}(a)$ is not collinear with $\Delta_{b_{i}}^{a_{i}}$. When this is not the case, we call $(a,b)$ a \emph{free transition}. The values of $x(a,b)$ for free transitions are carefully chosen below. Our choice always satisfies $x(a,b)x(b,a)=1$ (which always holds for non-free transitions). Let $x(a,a)=1$.

Because all vectors in square brackets lie in linearly independent spaces, they must all be equal to 0, and hence
\begin{align}
\label{eq:zero vectors1}
\vt{w}_{i}(a)-x(a,b)\vt{w}_{i}(b)&=y(a,b)\Delta_{b_{i}}^{a_{i}},\\
\label{eq:zero vectors2}
\vt{w}_{j}(a)-x(a,b)\vt{w}_{j}(b)&=0\text{ for each }j\neq i,\\
\label{eq:zero vectors3}
\text{and}\qquad \vt{w}_{0}(a)-x(a,b)\vt{w}_{0}(b)&=0.
\end{align}

\begin{lem}\label{lem:Product decomposition 1}
There exist a vector $\vt{w}^*_0$, vectors $\vt{w}^*_i(a_i)$ for each $i=1,\dots,I$ and $a_i$, and scalars $\gamma_i(a)\neq0$ for each $a$ and $i=0,\dots,I$ such that 
$\vt{w}_i(a)=\gamma_i(a)\vt{w}^*_i(a_i)$ for each $i$, and $\vt{w}_0(a)=\gamma_0(a)\vt{w}^*_0$.
\end{lem}

This result says that, for a fixed $a_i$, the vectors $\vt{w}_i(a_{-i}a_i)$ for $a_{-i}$ are either all collinear or all equal to $0$. 

\begin{proof}
For the first claim, fix an action $a^*$. For each $a_i$, let $\vt{w}_i^*(a_i)=\vt{w}_i(a^*_{-i}a_i)$. 
For each $a$, fix an arbitrary a path of adjacent actions $a^0=a^*_{-i}a_i,\dots,a^n=a$ such that for each $l<n$, $a^l_i=a_i$. Let $\gamma_i(a)=x(a^n,a^{n-1})\cdots x(a^1,a^0)$. A repeated application of (\ref{eq:zero vectors2}) shows that 
$\vt{w}_i(a)=\gamma_i(a)\vt{w}^*_i(a_i).$ Analogous argument shows the second claim.
\end{proof}

\begin{lem}\label{lem:Product decomposition 2}\label{lem:nontrivial problems}
For each $i$, one of the following is true:
\begin{enumerate}[label=(\roman*)]
\item $\vt{w}_i^*(a_i)=0$ for each $a_i$,
\item there exists $a_i^0\in A_i$ such that $\vt{w}_i^*(a^0_i)=0$ and $\vt{w}_i^*(a_i)\neq0$ for each $a_i\neq a^0_i$,
\item $\vt{w}_i^*(a_i)\neq0$ for each $a_i$.
\end{enumerate}
If question $X$ does not depend trivially on task $i$, then there is an action $a_i$ such that $\vt{w}_i^*(a_i)\neq 0$.
\end{lem}

\begin{proof}
Suppose that there are three different actions $a_i$, $b_i$, and $c_i$ such that $\vt{w}_i^*(a_i)=\vt{w}_i^*(b_i)=\vt{0}$ and $\vt{w}_i^*(c_i)\neq0$. Equation (\ref{eq:zero vectors1}) together with Lemma \ref{lem:Product decomposition 1} imply that, for any $c_{-i}$, $\vt{w}_i(c_ic_{-i})$ is simultaneously collinear with $\Delta_{c_i}^{a_i}$ and $\Delta_{c_i}^{b_i}$, contradicting the independence assumption. 

For the last claim, suppose that $\vt{w}_i^*(a_i)=0$ for each $a_i$. Equation (\ref{eq:zero vectors1}) together with Lemma \ref{lem:Product decomposition 1} imply that $y(a,b)=0$ for each $a$ and $b$ such that $a_{-i}=b_{-i}$. But then, for each $a_{-i}$, the vectors $\bar{\vt{X}}(a_{-i}a_i)$ and $\bar{\vt{X}}(a_{-i}b_i)=x(a_{-i}b_i,a_{-i}a_i)\bar{\vt{X}}(a_{-i}a_i)$ are collinear. Hence, $X$ depends on task $i$ trivially.
\end{proof}

Case (i) of Lemma \ref{lem:Product decomposition 2} is equivalent to $X$ depending on problem $i$ trivially. Because $X$ depends nontrivially on at least three tasks, there must be at least three tasks in cases (ii) or (iii).

\subsection{Exact cycles}\label{sec:Exact cycles}

\sloppy

Recall that an adjacency cycle $a^0,\dots,a^n=a^0$ is \textit{exact} if  $\prod_{l<n}x(a^l,a^{l+1})=1.$ The goal of this subsection as well as subsections \ref{sec:Exact cycles - Tools} to \ref{sec:Exact cycles - Free transitions}  is to prove the following result.

\fussy

\begin{lem}\label{lem:all cycles exact}
    The values $x(a,b)$ for free transitions $(a,b)$ can be chosen so that (i) $x(a,b)x(b,a)=1$ for all adjacent $a$ and $b$, and (ii) every adjacency cycle is exact.
\end{lem}

Lemma \ref{lem:Product decomposition 1} and a repeated application of equation (\ref{eq:zero vectors3}) shows that, if  $\vt{w}^*_0\neq0$, then, all adjacency cycles are exact. From now on, we assume that $\vt{w}^*_0=0$.

\subsection{Tools}\label{sec:Exact cycles - Tools}

The two results in this section develop tools that we use in the subsequent analysis. 

The first tool allows us to replace the problem of whether a cycle is exact with a problem about related cycles. For each path $c=(a^0,\dots,a^{n_c})$ and each transition $(a,b)$ between two adjacent actions, define 
\[
m_a^b(c)=\#\left\{l<n_c:a^l=a,a^{l+1}=b\right\}-\#\left\{l<n_c:a^l=b,a^{l+1}=a\right\}.
\]

\begin{lem}\label{lem:sum of exact cycles}
Suppose that, for some adjacency cycle $c$, there exists a collection $D$ of exact adjacency cycles such that for each adjacent pair $a,b$,
$m_a^b(c)=\sum_{d\in D}m_a^b(d)$. Then cycle $c$ is exact.
\end{lem}

\begin{proof}
Let $\prec$ be an arbitrary strict order on the set of actions $A$. Then we have
\begin{align*}
\prod_{l<n_c}x(a^{c,l},a^{c,l+1})
&=\prod_{a\prec b}\left(x(a,b)\right)^{m_a^b(c)}
=\prod_{a\prec b}\left(x(a,b)\right)^{\sum_{d\in D}m_a^b(d)}\\
&=\prod_{d\in D} \prod_{a\prec b}\left(x(a,b)\right)^{m_a^b(d)}
=\prod_{d\in D} \prod_{l<n_d}x(a^{d,l},a^{d,l+1})
= 1,
\end{align*}
where the last equality comes from the fact that cycles in $D$   are exact. 
\end{proof}

The second result shows that it is enough to consider particular kinds of ``small'' cycles.

\begin{lem}\label{lem:small cycles}
Every adjacency cycle is exact if (and only if) the following cycles are exact:
\begin{enumerate}
\item $\left(a,b,a\right)$ for all adjacent actions $a$ and $b$,
\item $\left(a,b,c,a\right)$ for all triples of actions $a$, $b$, and $c$ such that $a_{-i}=b_{-i}=c_{-i}$,
\item $\left(a,a_{-i}b_i, a_{-ij}b_ib_j,a_{-j}b_j,a\right)$ for all actions $a$, $i\neq j$, and $b_i$ and $b_j$.
\end{enumerate}
\end{lem}

\begin{proof}
Call cycles of the three forms described in the lemma \emph{small} and any other cycle \emph{large}. 

Take any large cycle of adjacent actions $a=a^0,\dots,a^n=a$ and let $i_l$ be such that, for each $l<n$, $a^l_{-i_l}=a^{l+1}_{-i_l}$. For future reference, notice that if the action in some task $i$ is ever changed, then it must be changed at least twice: if $i_l=i$ for some $l$, then there is some $l^\prime\neq l$ such that $i_{l^\prime}=i$. We use Lemma \ref{lem:sum of exact cycles} and small cycles to reorder and reduce the large cycle without changing the value of the product of the associated $x$ terms:
\begin{itemize}
\item if $i_l>i_{l+1}$ for $l<n-1$ we use the small cycle of type (3) to switch the order of the two tasks, i.e., replace the cycle fragment $\dots,a^l,a^{l+1},a^{l+2},\dots$, where $a^{l+1}=a^l_{-i_l}a^{l+2}_{i_l}$, with 
$\dots,a^l,a^l_{-i_{l+1}}a^{l+2}_{i_{l+1}},a^{l+2},\dots$. More precisely, due to Lemma \ref{lem:sum of exact cycles} and the fact that the small cycle is exact, the exactness of the original cycle is equivalent to the exactness of the modified cycle;
\item if $i_l=i_{l+1}$ (including possibly $i_{n-1}=i_0$), we use either type (1) or type (2) to reduce the large cycle, i.e, replace the cycle fragment $\dots,a^l,a^{l+1},a^{l+2},\dots$ with $\dots,a^l,a^{l+2},\dots$ in the case of a type (2) cycle or with $\dots,a^l,\dots$ in the case of a type (1) cycle.
\end{itemize}

Consider a process in which one of the above operations is applied until it cannot be applied anymore. Because either the operations reorder tasks in an increasing direction, or, if they don't change the order, they reduce the length of the cycle, the process never reverts and it will eventually stop. If the process stops at a single-element cycle $a$, then, because $x(a,a)=1$, the original cycle must be exact.

Otherwise, the process stops with a nontrivial cycle $a=a^0,\dots,a^m=a$ for some $2\leq m\leq n$. Then it must be that $i_l<i_{l+1}$ for each $l<m$. But this contradicts the earlier observation that, in an adjacency cycle, if $i$ appears at least once, it must appear at least twice. 
\end{proof}

\subsection{Cycles without free transitions}\label{sec:Exact cycles - No free transitions}

Here, we consider the exactness of small cycles without free transitions. We refer to $i$ from the definition of type (2) cycles as the \textit{relevant task} for this cycle; similarly, we refer to $i$ and $j$ as the relevant tasks for type (3) cycles. We say that a type (2) or type (3) cycle $\left(a^0,\dots,a^n=a^0\right)$  is \textit{grounded} if there exists some $k$ such that $k$ is not a relevant task and $\vt{w}^*_k(a_k^0)\neq0$. 

\begin{lem}\label{lem:grounded cycles}
Any grounded type (2) or type (3) cycle is exact.  
\end{lem}

\begin{proof}
Suppose that $\vt{w}^*_k(a_k^0)\neq0$ for some irrelevant task $k$. Then $\vt{w}_k(a^i)\neq0$ for each action $a^i$ in the cycle. A repeated application of (\ref{eq:zero vectors2}) shows that for each $l\leq n$,
\[
\vt{w}_k(a^0)=\vt{w}_k(a^l)x(a^0,a^1)\cdots x(a^{l-1},a^l).
\] 
The result follows from the fact that $a^n=a^0$. 
\end{proof}

\begin{lem}\label{lem:cycles with collinear Deltas}
Let $\left(a,a_{-i}b_i, a_{-ij}b_ib_j,a_{-j}b_j,a\right)$ be a type (3) small cycle such that either (i) $\vt{w}^*_i(a_i)\neq0$ and $\vt{w}^*_i(a_i)$ is not collinear with $\Delta_{a_i}^{b_i}$, or (ii) $\vt{w}^*_j(a_j)\neq0$ and $\vt{w}^*_j(a_j)$ is not collinear with $\Delta_{a_j}^{b_j}$.
Then the cycle is exact. 
\end{lem}

\begin{proof}
Using \ref{eq:zero vectors1} and \ref{eq:zero vectors2}, we get
\begin{align*}
x(a_{-ij}b_ib_j,a_{-i}b_i)& \left(x(a_{-i}b_i,a)\vt{w}_i(a)+y(a_{-i}b_i,a)\Delta_{a_i}^{b_i}\right)\\
&=x(a_{-ij}b_ib_j,a_{-i}b_i)\vt{w}_i(a_{-i}b_i)
=\vt{w}_i(a_{-ij}b_ib_j)\\
&=x(a_{-ij}b_ib_j, a_{-j}b_j)\vt{w}_i(a_{-j}b_j)+y(a_{-ij}b_ib_j, a_{-j}b_j)\Delta_{a_i}^{b_i}\\
&=x(a_{-ij}b_ib_j, a_{-j}b_j)x(a_{-j}b_j, a)\vt{w}_i(a)+y(a_{-ij}b_ib_j, a_{-j}b_j)\Delta_{a_i}^{b_i}.
\end{align*}
Suppose without loss of generality that $\vt{w}^*_i(a_i)\neq0$ and $\vt{w}^*_i(a_i)$ is not collinear with $\Delta_{a_i}^{b_i}$, which implies that $\vt{w}_i(a)\neq 0$ and $\vt{w}_i(a)$ is not collinear with $\Delta_{a_i}^{b_i}$. The first and the last line of the above sequence of equalities yield
\[
x(a_{-ij}b_ib_j,a_{-i}b_i) x(a_{-i}b_i,a) = x(a_{-ij}b_ib_j, a_{-j}b_j)x(a_{-j}b_j, a).
\]
Hence, the cycle is exact. 
\end{proof}

\begin{lem}\label{lem:nonzero relevant}
Suppose that a type (2) or type (3) small cycle is such that $\vt{w}^*_i(a_i)\neq 0$ for each action $a$ in the cycle and each relevant task $i$ and none of the transitions is free. Then it is exact.
\end{lem}

\begin{figure}[ht]
    \centering
{\scriptsize \scalebox{0.8}{
\begin{tikzpicture}[scale=0.8]
    \coordinate (A) at (0,0,3);\coordinate (AB) at (0,0,0);
    \coordinate (B) at (3,0,3);\coordinate (BB) at (3,0,0);
    \coordinate (C) at (3,3,3);\coordinate (CB) at (3,3,0);
    \coordinate (D) at (0,3,3);\coordinate (DB) at (0,3,0);
    
    \draw[thick] (C) -- (D) -- (DB) -- (CB) -- cycle;
    \draw[thick] (B) -- (A) -- (D) -- (C) -- cycle;
    \fill[orange!50!white, opacity=0.5] (B) -- (A) -- (D) -- (C) -- cycle;
    \draw[thick] (CB) -- (BB);
    \draw[thick] (B) -- (BB);
    \draw[dashed] (DB) -- (AB);
    \draw[dashed] (A) -- (AB);
    \draw[dashed] (BB) -- (AB);

     \node at (A) [circle,fill=black,inner sep=1pt,label=below left:{$a$}] {};
    \node at (B) [circle,fill=black,inner sep=1pt,label=below right:{$a_{-i}b_i$}] {};
    \node at (C) [circle,fill=black,inner sep=1pt,label=above left:{$a_{-ij}b_ib_j$}] {};
    \node at (D) [circle,fill=black,inner sep=1pt,label=above left:{$a_{-j}b_j$}] {};
    \node at (AB) [circle,fill=black,inner sep=1pt,label=below right:{$a_{-k}b_k$}] {};
    \node at (BB) [circle,fill=black,inner sep=1pt,label=above right:{$a_{-ik}b_ib_k$}] {};
    \node at (CB) [circle,fill=black,inner sep=1pt,label=above right:{$a_{-ijk}b_ib_jb_k$}] {};
    \node at (DB) [circle,fill=black,inner sep=1pt,label=above left:{$a_{-jk}b_jb_k$}] {};

    \coordinate (AF) at (6,0,3);\coordinate (AB) at (7,0,0);
    \coordinate (BF) at (10,0,3);\coordinate (BB) at (11,0,0);
    \coordinate (CF) at (8,3,3);\coordinate (CB) at (9,3,0);

    \draw[thick] (AF) -- (CF) -- (BF) -- cycle;
    \fill[orange!50!white, opacity=0.5] (BF) -- (AF) -- (CF) -- cycle;
    \draw[dashed] (AB) --  (BB);
    \draw[thick] (BB) -- (CB);
    \draw[dashed] (CB) --  (AB);
    \draw[dashed] (AB) -- (AF);
    \draw[thick] (BB) -- (BF);
    \draw[thick] (CB) -- (CF);
    
    \node at (AF) [circle,fill=black,inner sep=1pt,label=below left:{$a$}] {};
    \node at (BF) [circle,fill=black,inner sep=1pt,label=below right:{$a_{-i}b_{i}$}] {};
    \node at (CF) [circle,fill=black,inner sep=1pt,label=above left:{$a_{-i}c_{i}$}] {};
    \node at (AB) [circle,fill=black,inner sep=1pt,label=below right:{$a_{-k}b_k$}] {};
    \node at (BB) [circle,fill=black,inner sep=1pt,label=above right:{$a_{-ik}b_ib_k$}] {};
    \node at (CB) [circle,fill=black,inner sep=1pt,label=above right:{$a_{-ik}c_ib_k$}] {};
\end{tikzpicture}
}}
\caption{Small cycles of type (3) and type (2)}\label{fig:cycles}
\end{figure}
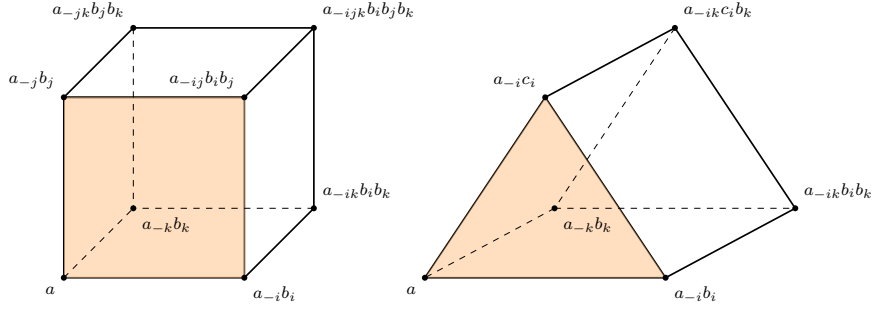

\begin{proof}
If the cycle is grounded, the result follows from Lemma \ref{lem:grounded cycles}. Accordingly, suppose the cycle is not grounded. Then, because of Lemma \ref{lem:nontrivial problems}, and due to the assumption that at least three different tasks are nontrivial, there exists a non-relevant $k$ and an action $b_k\neq a_k$ such that $\vt{w}^*_k(b_k)\neq 0$. 

Suppose that the original cycle $\left(a,a_{-i}b_i, a_{-ij}b_ib_j,a_{-j}b_j,a\right)$ is of type (3). Such a cycle corresponds to the orange face on the left-side of Figure \ref{fig:cycles}. Consider the type (3) cycles that are associated with all five remaining faces of the cube:
\begin{itemize}
\item the bottom, $(a,a_{-i}b_i,a_{-ik}b_ib_k,a_{-k}b_k,a)$;
\item the right, $(a_{-i}b_i,a_{-ij}b_ib_j,a_{-ijk}b_ib_jb_k,a_{-ik}b_ib_k,a_{-i}b_i)$;
\item the top, $(a_{-j}b_j,a_{-jk}b_jb_k,a_{-ijk}b_ib_jb_k,a_{-ij}b_ib_j, a_{-j}b_j)$;
\item the left, $(a,a_{-k}b_k,a_{-jk}b_jb_k,a_{-j}b_j,a)$;
\item and the back, $(a_{-k}b_k,a_{-ik}b_ib_k, a_{-ijk}b_ib_jb_k,a_{-jk}b_jb_k, a_{-k}b_k)$.
\end{itemize}
The first four of these cycles are grounded by the assumption that $\vt{w}^*_i(a_i)\neq 0$ for each action $a$ in the cycle and each relevant task $i$), and the back cycle is grounded by the choice of $k$ and $b_k$. Thus all five cycles are exact by Lemma \ref{lem:grounded cycles}. Moreover, the conditions of Lemma \ref{lem:sum of exact cycles} are satisfied. Therefore, the original cycle is exact as well.

If the cycle is of type (2), then because the cycle is not grounded, we have $\bar{\vt{X}}(a)=\gamma_i(a)\vt{w}_i^*(a_i)$, $\bar{\vt{X}}(a_{-i}b_i)=\gamma_i(a_{-i}b_i)\vt{w}_i^*(b_i)$, and $\bar{\vt{X}}(a_{-i}c_i)=\gamma_i(a_{-i}c_i)\vt{w}_i^*(c_i)$. Because the transition $(a,a_{-i}b_i)$ is not free, the vector $\vt{w}_i^*(a_i)$ cannot be collinear with $\Delta_{a_i}^{b_i}$ (otherwise, $\bar{\vt{X}}(a)$ would also be collinear with $\Delta_{a_i}^{b_i}$, in which case the transition is free). As a result, Lemma \ref{lem:cycles with collinear Deltas} applies to the 4-action cycle $(a,a_{-i}b_i,a_{-ik}b_ib_k,a_{-k}b_k,a)$, which means that this cycle is exact; see the right panel of Figure \ref{fig:cycles}. An analogous observation holds for the other two 4-action cycles depicted in Figure \ref{fig:cycles}. The 3-action cycle in the back is grounded, and hence exact by Lemma \ref{lem:grounded cycles}. Finally, an application of Lemma \ref{lem:sum of exact cycles} shows that the original cycle is exact as well. 
\end{proof}

\begin{lem}\label{lem:no free transitions}
If a type (2) or (3) cycle has no free transitions, then it is exact.
\end{lem}

\begin{proof}
By Lemmas \ref{lem:grounded cycles} and \ref{lem:nonzero relevant}, it suffices to assume that the cycle is not grounded and $\vt{w}^*_i(a_i)=0$ for some relevant $i$ and action $a$ in the cycle. 
In such a case, if the cycle were of type (2), all transitions to action $a$ would be free. 

Suppose the cycle is of type (3). Let $b\neq a$ be the action in the cycle such that $a_{-j}=b_{-j}$. It follows that $w^*_k(a_k)\neq0$ if and only if $k=j$. Then
\begin{align*}
\bar{\vt{X}}(a)=\gamma_j(a)\vt{w}^*_j(a_j)\text{ and }\bar{\vt{X}}(b)=\gamma_j(b)\vt{w}^*_j(b_j)=x(b,a)\gamma_j(a)\vt{w}^*_j(a_j)
+y(b,a)\Delta_{a_j}^{b_j}.
\end{align*}
Because the transition $(a,b)$ is not free, it must be that $\vt{w}^*_j(a_j)\neq0$ and $\vt{w}^*_j(a_j)$ is not collinear with $\Delta_{a_j}^{b_j}$. Lemma \ref{lem:cycles with collinear Deltas} therefore implies that the cycle is exact.
\end{proof}

\subsection{Cycles with free transitions}\label{sec:Exact cycles - Free transitions}

In this subsection, we show how to determine values of $x(a,b)$ for free transitions so that all cycles that include such transitions are exact. For simplicity, if $X$ depends trivially on task $i$, we say simply that $i$ is trivial (and similarly for nontrivial dependence).

Notice first that, if a transition between adjacent actions $a$ and $b$ such that $a_{-i}=b_{-i}$ is free, then it must be that $\vt{w}^*_j(a_j)=0$ for each $j\neq i$. Indeed, if
$\bar{\vt{X}}(a)$ is collinear with $\Delta_{a_i}^{b_i}$, because of the linear independence of $\Delta_{a_i}^{b_i}$ and $\vt{w}_j(a)$ for each $j\neq i$, it must be that $\vt{w}_j(a)=0$, and hence $\vt{w}^*_j(a_j)=0$.  We refer to this property as the \textit{test} for freeness of the transition (which provides necessary conditions). It follows that, if $j\neq i$ is nontrivial, then task $j$ is in case (ii) from Lemma \ref{lem:nontrivial problems} and $a_j=a_j^0$.

In what follows, we consider two cases.

\emph{Case I}: There exists a single task in case (iii) from Lemma \ref{lem:nontrivial problems}, i.e., one nontrivial $i$ such that $\vt{w}^*_i(a_i)\neq0$ for all $a_i$. Assume without loss of generality that $i=I$. In this case, all free transitions must be between adjacent actions $a$ and $b$ such that $a_{-I}=b_{-I}$. Moreover, it must be that $w^*_I(a_I)$ and $w^*_I(b_I)$ are collinear with $\Delta_{a_I}^{b_I}$.

Assume without loss of generality that task $1$ is nontrivial and fix action $b_1\neq a_1$ so that $\vt{w}^*(b_1)\neq 0$. Let 
\[
x(a,b)=x(a,a_{-1}b_1)x(a_{-1}b_1,a_{-1I}b_1b_I)x(a_{-1I}b_1b_I,b).
\]
The above definition implies that the cycle $(a,a_{-1}b_1,a_{-1I}b_1b_I,b,a)$ is exact. This cycle corresponds to the red face(s) in Figure \ref{fig:I Cycles with free transitions}. The orange edge corresponds to the free transition. Notice that each red face cycle contains only one free transition. This is because of the test: all other transitions of the cycle either keep fixed the action in task $I$ or the action $b_1$, and those actions are associated with nonzero $\vt{w}^*_\cdot(\cdot)$ vectors. 

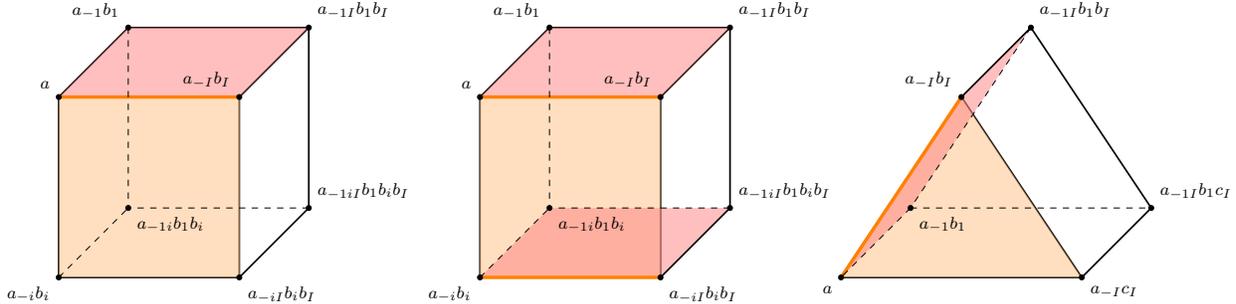
\begin{figure}[ht]
    \centering
    {\scriptsize \scalebox{0.8}{

\begin{tikzpicture}[scale=0.8]\label{fig:I Cycles with free transitions}
    \coordinate (A) at (0,0,3);\coordinate (AB) at (0,0,0);
    \coordinate (B) at (3,0,3);\coordinate (BB) at (3,0,0);
    \coordinate (C) at (3,3,3);\coordinate (CB) at (3,3,0);
    \coordinate (D) at (0,3,3);\coordinate (DB) at (0,3,0);
    
    \draw[thick] (C) -- (D) -- (DB) -- (CB) -- cycle;
    \fill[red!50!white, opacity=0.5] (C) -- (D) -- (DB) -- (CB) -- cycle;
    \draw[thick] (B) -- (A) -- (D) -- (C) -- cycle;
    \fill[orange!50!white, opacity=0.5] (B) -- (A) -- (D) -- (C) -- cycle;
    \draw [ultra thick, orange] (C) -- (D);
    \draw[thick] (CB) -- (BB);
    \draw[thick] (B) -- (BB);
    \draw[dashed] (DB) -- (AB);
    \draw[dashed] (A) -- (AB);
    \draw[dashed] (BB) -- (AB);

    \node at (A) [circle,fill=black,inner sep=1pt,label=below left:{$a_{-i}b_i$}] {};
    \node at (B) [circle,fill=black,inner sep=1pt,label=below right:{$a_{-iI}b_ib_I$}] {};
    \node at (C) [circle,fill=black,inner sep=1pt,label=above left:{$a_{-I}b_I$}] {};
    \node at (D) [circle,fill=black,inner sep=1pt,label=above left:{$a$}] {};
    \node at (AB) [circle,fill=black,inner sep=1pt,label=below right:{$a_{-1i}b_1b_i$}] {};
    \node at (BB) [circle,fill=black,inner sep=1pt,label=above right:{$a_{-1iI}b_1b_ib_I$}] {};
    \node at (CB) [circle,fill=black,inner sep=1pt,label=above right:{$a_{-1I}b_1b_I$}] {};
    \node at (DB) [circle,fill=black,inner sep=1pt,label=above left:{$a_{-1}b_1$}] {};

    \coordinate (A) at (7,0,3);\coordinate (AB) at (7,0,0);
    \coordinate (B) at (10,0,3);\coordinate (BB) at (10,0,0);
    \coordinate (C) at (10,3,3);\coordinate (CB) at (10,3,0);
    \coordinate (D) at (7,3,3);\coordinate (DB) at (7,3,0);
    
    \draw[thick] (C) -- (D) -- (DB) -- (CB) -- cycle;
    \fill[red!50!white, opacity=0.5] (C) -- (D) -- (DB) -- (CB) -- cycle;
    \draw[thick] (B) -- (A) -- (D) -- (C) -- cycle;
    \fill[orange!50!white, opacity=0.5] (B) -- (A) -- (D) -- (C) -- cycle;
    \fill[red!50!white, opacity=0.5] (A) -- (B) -- (BB) -- (AB) -- cycle;
    \draw [ultra thick, orange] (A) -- (B);
    \draw [ultra thick, orange] (C) -- (D);
    \draw[thick] (CB) -- (BB);
    \draw[thick] (B) -- (BB);
    \draw[dashed] (DB) -- (AB);
    \draw[dashed] (A) -- (AB);
    \draw[dashed] (BB) -- (AB);

    \node at (A) [circle,fill=black,inner sep=1pt,label=below left:{$a_{-i}b_i$}] {};
    \node at (B) [circle,fill=black,inner sep=1pt,label=below right:{$a_{-iI}b_ib_I$}] {};
    \node at (C) [circle,fill=black,inner sep=1pt,label=above left:{$a_{-I}b_I$}] {};
    \node at (D) [circle,fill=black,inner sep=1pt,label=above left:{$a$}] {};
    \node at (AB) [circle,fill=black,inner sep=1pt,label=below right:{$a_{-1i}b_1b_i$}] {};
    \node at (BB) [circle,fill=black,inner sep=1pt,label=above right:{$a_{-1iI}b_1b_ib_I$}] {};
    \node at (CB) [circle,fill=black,inner sep=1pt,label=above right:{$a_{-1I}b_1b_I$}] {};
    \node at (DB) [circle,fill=black,inner sep=1pt,label=above left:{$a_{-1}b_1$}] {};



    \coordinate (AF) at (13,0,3);\coordinate (AB) at (13,0,0);
    \coordinate (BF) at (17,0,3);\coordinate (BB) at (17,0,0);
    \coordinate (CF) at (15,3,3);\coordinate (CB) at (15,3,0);
    
    \draw[thick] (AF) -- (CF) -- (BF) -- (AF);
    \fill[orange!50!white, opacity=0.5] (BF) -- (AF) -- (CF) -- cycle;
    \fill[red!50!white, opacity=0.5] (CF) -- (AF) -- (AB) -- (CB) -- cycle;
    \draw [ultra thick, orange] (CF) -- (AF);
    \draw[dashed] (AB) --  (BB);
    \draw[thick] (BB) -- (CB);
    \draw[dashed] (CB) --  (AB);
    \draw[dashed] (AB) -- (AF);
    \draw[thick] (BB) -- (BF);
    \draw[thick] (CB) -- (CF);
    
    \node at (AF) [circle,fill=black,inner sep=1pt,label=below left:{$a$}] {};
    \node at (BF) [circle,fill=black,inner sep=1pt,label=below right:{$a_{-I}c_{I}$}] {};
    \node at (CF) [circle,fill=black,inner sep=1pt,label=above left:{$a_{-I}b_{I}$}] {};
    \node at (AB) [circle,fill=black,inner sep=1pt,label=below right:{$a_{-1}b_1$}] {};
    \node at (BB) [circle,fill=black,inner sep=1pt,label=above right:{$a_{-1I}b_1c_I$}] {};
    \node at (CB) [circle,fill=black,inner sep=1pt,label=above right:{$a_{-1I}b_1b_I$}] {};
\end{tikzpicture}
}}
\caption{Cycles with free transitions in case I.}\label{fig:free cycles I}
\end{figure}

There are three types of cycles that contain free transitions, depicted as orange faces in Figure \ref{fig:I Cycles with free transitions}. 

The left panel corresponds to the cycle $(a, a_{-i}b_i, a_{-iI}b_ib_I, a_{-I}b_I, a)$ when task $i$ is nontrivial. In this case, $b_i\neq a_i^0$, and, by Lemma \ref{lem:Product decomposition 2}, $\vt{w}_i^*(b_i)\neq 0$. The test implies that none of the other transitions in the orange cycle are free: either action $b_i$ or action $a_I$ is fixed. Analogously, an application of the test shows that none of the other cycles (the uncolored faces) is free: one of the actions $a_I$, $b_i$, or $b_1$ is fixed. Proceeding as in the proof of Lemma \ref{lem:nonzero relevant}, we see that this cycle is exact. 

The center panel corresponds to a situation when task $i$ is trivial. In this case, the transition $(a_{-i}b_i,a_{-iI}b_ib_I)$ is free and $x\left(a_{-i}b_i,a_{-iI}b_ib_I\right)$ can be chosen to make the cycle on the bottom face exact. None of the other transitions are free. Because the red faces and the uncolored cycles are exact, the above argument implies that the orange face cycle is exact as well.

The right panel corresponds to the orange cycle $(a, a_{-I}b_I, a_{-I}c_I,a)$. The other transitions of the orange cycle are not free (otherwise, $w^*(a_I)$ would be collinear with $\Delta_{a_I}^{c_I}$, which would violate linear independence of the latter vector with $\Delta_{a_I}^{b_I}$). All other transitions fix one of the actions: $a_I$, $b_I$, $c_I$, or $b_1$. Hence, due to the test, none of the remaining transitions are free. The claim follows from the same reasoning as in Lemma \ref{lem:nonzero relevant}. 

\emph{Case II}: For all nontrivial $i$, there exists a unique $a^0_i$ such that $\vt{w}^*_i(a^0_i)=0$. Let $a^0$ be the product problem action that consists of actions $a^0_i$. Assume without loss of generality that task $i=1$ is nontrivial. Fix action $a^*_1\neq a^0_1$.

In this case, a transition is free if and only if it takes the form $(a^0_{-i}a_i,a^0)$ for some $i$. Indeed, the above observation implies that if transition $(a,b)$ is free, then $a_{-i}=a^0_{-i}$. Furthermore, if neither $a=a^0$ nor $b=a^0$, then both $w^*_i(a_i)$ and $w^*_i(b_i)$ must be collinear with $\Delta_{a_i}^{b_i}$. But, together with the linear independence assumption, this implies that $w^*_i(a_i)$ is not collinear with $\Delta_{a_i}^{a^0_i}$, which contradicts the fact that 
\[x(a^0,a)\gamma_i(a)w^*_i(a_i)+y(a^0,a)\Delta_{a_i}^{a^0_i}
=\gamma_i(a^0)w^*_i(a_i^0)
=0.\]

For each $b_1\neq a^*_1$, each $i\neq 1$, and each $b_i$, let
\begin{align*}
x(a^0_{-1}a^*_1,a^0)&=1,\text{ }
x(a^0_{-1}b_1,a^0)=x(a^0_{-1}b_1,a^0_{-1}a^*_1),\text{ and }\\
\qquad x(a^0_{-i}b_i,a^0)&=x(a^0_{-i}b_i,a^0_{-1i}a^*_1b_i)x(a^0_{-1i}a^*_1b_i,a^0_{-1}a^*_1).
\end{align*} \sloppy
Due to the above definition, the cycles $(a^0,a^0_{-1}b_1,a^0_{-1}a^*_1,a^0)$ and $(a^0,a^0_{-i}b_i,a^0_{-i1}a^*_1b_i,a^0_{-1}a^*_1,a^0)$ are exact. These cycles correspond to the red faces in Figure \ref{fig:II Cycles without free transitions}.

\fussy

There are three types of cycles that contain free transitions in this case other than the cycles listed above. 

First, consider a cycle $(a^0,a^0_{-i}b_i,a^0_{-ij}b_ib_j,a^0_{-j}b_j,a^0)$ for $i\neq 1$. This cycle is depicted in orange on the left panel of Figure \ref{fig:II Cycles without free transitions}. The two cycles depicted in red are exact due to the choice of the $x$ coefficients. Finally, all of the other cycles are exact because they do not contain free transitions. 

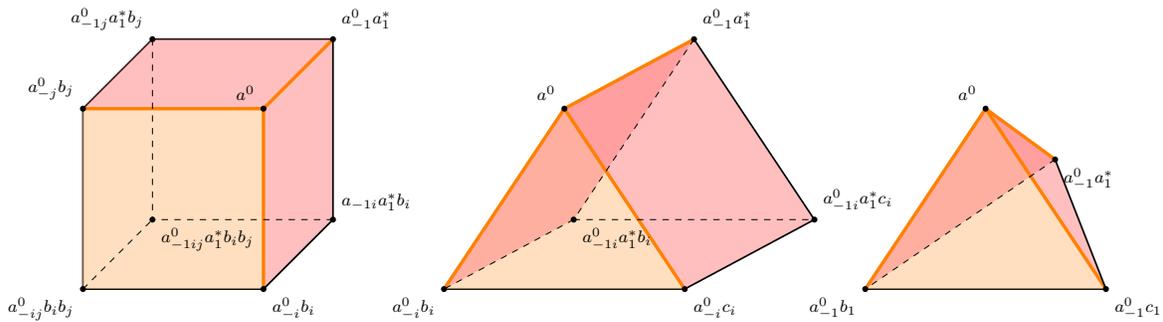
\begin{figure}[ht]
    \centering
{\scriptsize \scalebox{0.8}{
\begin{tikzpicture}[scale=0.8]\label{fig:II Cycles without free transitions}
    \coordinate (A) at (0,0,3);\coordinate (AB) at (0,0,0);
    \coordinate (B) at (3,0,3);\coordinate (BB) at (3,0,0);
    \coordinate (C) at (3,3,3);\coordinate (CB) at (3,3,0);
    \coordinate (D) at (0,3,3);\coordinate (DB) at (0,3,0);
    
    \draw[thick] (C) -- (D) -- (DB) -- (CB) -- cycle;
    \fill[red!50!white, opacity=0.5] (C) -- (D) -- (DB) -- (CB) -- cycle;
    \fill[red!50!white, opacity=0.5] (B) -- (C) -- (CB) -- (BB) -- cycle;
    \draw[thick] (B) -- (A) -- (D) -- (C) -- cycle;
    \fill[orange!50!white, opacity=0.5] (B) -- (A) -- (D) -- (C) -- cycle;
    \draw [ultra thick, orange] (C) -- (D);
    \draw [ultra thick, orange] (C) -- (CB);
    \draw [ultra thick, orange] (C) -- (B);
    \draw[thick] (CB) -- (BB);
    \draw[thick] (B) -- (BB);
    \draw[dashed] (DB) -- (AB);
    \draw[dashed] (A) -- (AB);
    \draw[dashed] (BB) -- (AB);

    \node at (A) [circle,fill=black,inner sep=1pt,label=below left:{$a^0_{-ij} b_ib_j$}] {};
    \node at (B) [circle,fill=black,inner sep=1pt,label=below right:{$a^0_{-i}b_i$}] {};
    \node at (C) [circle,fill=black,inner sep=1pt,label=above left:{$a^0$}] {};
    \node at (D) [circle,fill=black,inner sep=1pt,label=above left:{$a^0_{-j}b_j$}] {};
    \node at (AB) [circle,fill=black,inner sep=1pt,label=below right:{$a^0_{-1ij}a^*_1b_ib_j$}] {};
    \node at (BB) [circle,fill=black,inner sep=1pt,label=above right:{$a_{-1i}a^*_1b_i$}] {};
    \node at (CB) [circle,fill=black,inner sep=1pt,label=above right:{$a^0_{-1}a^*_1$}] {};
    \node at (DB) [circle,fill=black,inner sep=1pt,label=above left:{$a^0_{-1j}a^*_1b_j$}] {};

    \coordinate (AF) at (6,0,3);\coordinate (AB) at (7,0,0);
    \coordinate (BF) at (10,0,3);\coordinate (BB) at (11,0,0);
    \coordinate (CF) at (8,3,3);\coordinate (CB) at (9,3,0);

    \draw[thick] (AF) -- (CF) -- (BF) -- (AF);
    \fill[orange!50!white, opacity=0.5] (BF) -- (AF) -- (CF) -- cycle;
    \fill[red!50!white, opacity=0.5] (CF) -- (AF) -- (AB) -- (CB) -- cycle;
    \fill[red!50!white, opacity=0.5] (CF) -- (CB) -- (BB) -- (BF) -- cycle;
    \draw [ultra thick, orange] (CF) -- (AF);
    \draw [ultra thick, orange] (CF) -- (CB);
    \draw [ultra thick, orange] (CF) -- (BF);
    \draw[dashed] (AB) --  (BB);
    \draw[thick] (BB) -- (CB);
    \draw[dashed] (CB) --  (AB);
    \draw[dashed] (AB) -- (AF);
    \draw[thick] (BB) -- (BF);
    
    \node at (AF) [circle,fill=black,inner sep=1pt,label=below left:{$a^0_{-i}b_i$}] {};
    \node at (BF) [circle,fill=black,inner sep=1pt,label=below right:{$a^0_{-i}c_i$}] {};
    \node at (CF) [circle,fill=black,inner sep=1pt,label=above left:{$a^0$}] {};
    \node at (AB) [circle,fill=black,inner sep=1pt,label=below right:{$a^0_{-1i}a^*_1b_i$}] {};
    \node at (BB) [circle,fill=black,inner sep=1pt,label=above right:{$a^0_{-1i}a^*_1c_i$}] {};
    \node at (CB) [circle,fill=black,inner sep=1pt,label=above right:{$a^0_{-1}a^*_1$}] {};



    \coordinate (AF) at (13,0,3);
    \coordinate (BF) at (17,0,3);
    \coordinate (CF) at (15,3,3);
    \coordinate (B) at (15,1,0);

    \draw[thick] (AF) -- (CF) -- (BF) -- (AF);
    \fill[orange!50!white, opacity=0.5] (BF) -- (AF) -- (CF) -- cycle;
    \fill[red!50!white, opacity=0.5] (CF) -- (AF) -- (B) -- cycle;
    \fill[red!50!white, opacity=0.5] (CF) -- (B) -- (BF) -- cycle;
    \draw [ultra thick, orange] (CF) -- (AF);
    \draw [ultra thick, orange] (CF) -- (B);
    \draw [ultra thick, orange] (CF) -- (BF);
    \draw[dashed] (B) -- (AF);
    \draw[thick] (B) -- (BF);
    
    \node at (AF) [circle,fill=black,inner sep=1pt,label=below left:{$a^0_{-1}b_1$}] {};
    \node at (BF) [circle,fill=black,inner sep=1pt,label=below right:{$a^0_{-1}c_1$}] {};
    \node at (CF) [circle,fill=black,inner sep=1pt,label=above left:{$a^0$}] {};
    \node at (B) [circle,fill=black,inner sep=1pt,label=below right:{$a^0_{-1}a^*_1$}] {};
    
\end{tikzpicture}
}}
\caption{Cycles with free transitions in case II.}\label{fig:free cycles II}
\end{figure}

Second, consider a cycle $(a^0,a^0_{-i}b_i,a^0_{-i}c_i,a^0)$ for $i\neq 1$. This cycle corresponds to front wall depicted in orange on the center panel of Figure \ref{fig:II Cycles without free transitions}. The two cycles depicted in red are exact due to the choice of the $x$ coefficients. All of the other cycles (corresponding to the back and bottom walls) are exact because they do not contain free transitions.

Third, consider a cycle $(a^0,a^0_{-1}b_1,a^0_{-1}c_1,a^0)$. This cycle corresponds to the front wall depicted in orange on the right panel of Figure \ref{fig:II Cycles without free transitions}. The two cycles depicted in red (corresponding to the left and right walls) are exact due to the choice of the $x$ coefficients. The remaining cycle (corresponding to the bottom wall) is exact because it does not contain any free transitions.

Proceeding as in the proof of Lemma \ref{lem:nonzero relevant}, we see that each of the considered cycles is exact. 

Because there can be at most one task in case (iii) of Lemma \ref{lem:nontrivial problems}, there are no other cases to consider. Together with Lemma \ref{lem:no free transitions}, this concludes the proof of Lemma \ref{lem:all cycles exact}.

\subsection{Conclusion of the proof of Theorem \ref{prop:product problems}}\label{sec: Product proof conclusion}

Recall that $\vt{w}_i(a)=\gamma_i(a)\vt{w}_i^*(a_i)$. The next result delivers additional information about the function $\gamma_i(\cdot)$.

\begin{lem}\label{lem:potential}
There exist $\gamma:A\longrightarrow\R$, $\gamma_i^*:A_i\longrightarrow\R$, and $\gamma^*_0\in \R$ such that, for each $a$, $\gamma_i(a)=\gamma(a)\gamma^*_i(a_i)$ for each $i=1,\dots,I$, and $\gamma_0(a)=\gamma(a)\gamma^*_0$. 
In addition, $x(a,b)=\frac{\gamma(a)}{\gamma(b)}$ for any two adjacent actions $a$ and $b$.
\end{lem}

\begin{proof}
Fix an action $a^*$. For each $a$, find a path of adjacent actions $a^*=a^0,\dots,a^m=a$. Define $\gamma(a)=x(a^m,a^{m-1})\cdots x(a^1,a^0).$ Lemma \ref{lem:all cycles exact} implies that $\gamma(a)$ is well defined in that its definition does not depend on the choice of path from $a^*$ to $a$. Moreover, for any two adjacent actions $a$ and $b$, if $a^*=a^0,\dots,a^m=a$ is an adjacency path from $a^*$ to $a$, then $a^*=a^0,\dots,a^m,b$ is an adjacency path from $a^*$ to $b$, and
\[
\gamma(b)=x(b,a)x(a^m,a^{m-1})\cdots x(a^1,a^0)=x(b,a)\gamma(a).
\]
Let $\gamma^*_i(a_i)=\gamma_i(a^*_{-i}a_i)/\gamma(a^*_{-i}a_i)$. The claim for $i>0$ follows from the fact that, for each $a$, if $\vt{w}^*_i(a_i)\neq 0$, then
\[
\vt{w}_i(a)=\frac{\gamma_i(a)}{\gamma_i(a^*_{-i}a_i)}\vt{w}_i(a^*_{-i}a_i)
=\frac{\gamma(a)}{\gamma(a^*_{-i}a_i)}\vt{w}_i(a^*_{-i}a_i).
\]
A similar argument establishes the claim for $i=0$ (see also the proof of Lemma \ref{lem:Product decomposition 1}). 
\end{proof}

\begin{lem}
There exist $y^*_i\in \R$ and $\vt{d}_i\in E_i$ such that
$\gamma^*_i(a_i)\vt{w}^*_i(a_i) = y^*_i\bar{\vt{u}}_i(a_i) + \vt{d}_i$ for any $a_i$.
\end{lem}


\begin{proof}
By (\ref{eq:zero vectors1}), for any actions $a$ and $b$ such that $a_{-i}=b_{-i}$,
\begin{align}\label{eq:w representation}
&\gamma_i^*(a_i)\vt{w}_i^*(a_i)-\gamma_i^*(b_i)\vt{w}_i^*(b_i)
=\frac{1}{\gamma(a)}\vt{w}_i(a)-\frac{1}{\gamma(b)}\vt{w}_i(b)\\
\nonumber =&\frac{1}{\gamma(a)}\left(\vt{w}_i(a)-x(a,b)\vt{w}_i(b)\right)
=\frac{1}{\gamma(a)}y(a,b)\Delta^{a_i}_{b_i}.
\end{align}
Because the left-hand side and $\Delta^{a_i}_{b_i}$ do not depend on $a_{-i}$, neither does $y(a,b)/\gamma(a)$. Let $y_i^*(a_i,b_i)=y(a,b)/\gamma(a)$. 

If task $i$ has only two actions, it is easy to see that $y^*_i(a_i,b_i)=y^*_i(b_i,a_i)=:y^*_i$. The claim follows. 

If task $i$ has at least three actions, take $a$, $b$, and $c$ such that $a_{-i}=b_{-i}=c_{-i}$ and $a_i$, $b_i$, and $c_i$ are distinct. Applying the above equation to pairs $(a,b)$, $(b,c)$, and $(c,a)$ yields
\begin{align*}\label{eq:w representation Independence}
y^*_i(a_i,c_i)\Delta^{a_i}_{b_i}
+y^*_i(a_i,c_i)\Delta^{b_i}_{c_i}
&=y^*_i(a_i,c_i)\Delta^{a_i}_{c_i}
=\gamma_i^*(a_i)\vt{w}_i^*(a_i)-\gamma_i^*(c_i)\vt{w}_i^*(c_i)\\
&=\gamma_i^*(a_i)\vt{w}_i^*(a_i)-\gamma_i^*(b_i)\vt{w}_i^*(b_i)
+\gamma_i^*(b_i)\vt{w}_i^*(b_i)-\gamma_i^*(c_i)\vt{w}_i^*(c_i)\\
&=y^*_i(a_i,b_i)\Delta^{a_i}_{b_i}
+y^*_i(b_i,c_i)\Delta^{b_i}_{c_i}.
\end{align*}
The independence assumption implies that $y^*_i(a_i,b_i)=y^*_i(a_i,c_i)$. Because the claim holds for arbitrary and distinct actions, there must be $y^*_i$ such that for all $a_i$ and $b_i$, we have
$y^*_i\left(a_i,b_i\right)=y^*_i$.

Finally, fix $a_i^*$ and take $\vt{d}_i=\gamma_i^*(a_i^*)\vt{w}_i^*(a_i^*)-y_i^*\bar{\vt{u}}_i(a^*_i)$. The claim follows from equation (\ref{eq:w representation}).
\end{proof}

Substituting the observations from the two lemmas back into the decomposition of $\vt{\bar{X}}=\sum_{i}\vt{w}_i(a)$, we obtain
\begin{align*}
\bar{X}(a)
&=\gamma(a)\left(\sum_{i}\gamma^*_i(a_i)\vt{w}^*_i(a_i)+\gamma_0^*\vt{w}^*_0\right)=\gamma(a)\left(
\sum_{i}y^*_i\bar{\vt{u}}_i(a_i)
+\left[\sum_{i}\vt{d}_i+\gamma_0^*\vt{w}^*_0\right]\right).
\end{align*}
Let $\vt{d}$ be the vector in the square brackets. The result follows.


\subsection{Converse}

We have shown that \eqref{eq: weighted-aligned} is necessary for incentivizability. All that remains is to show that it is sufficient.

Notice that if $\tau_i>0$ for all $i$, the product problem is equivalent (in terms of $u$-optimal choices) to a problem with payoffs
$u\left(a;\theta\right)=\sum_{i}\tau_i u_{i}\left(a_{i},\theta_{i}\right).$
In this latter problem, any $X$ satisfying \eqref{eq: weighted-aligned} is aligned with $u$, and is therefore incentivizable. This in turn implies that $X$ is incentivizable in the original problem.

If $\tau_i\leq 0$ for some $i$, note that, by the proof of Proposition \ref{Observation 0}, to show that $X$ of the form described in \eqref{eq: weighted-aligned} is incentivizable, it suffices to show that
$X(a;\theta)=d(\theta)+\sum_i \tau_i u_i(a_i,\theta_i)$ is, where we may assume $|\tau_i|<1$ for all $i$ and $X(a;\theta)\in [0,1]$ for all $a$ and $\theta$. Letting
\begin{align*}
V(r,a,\theta)&=\int_0^r X(a;\theta)dx + \int_r^1 xdx -\frac{1}{2}+\sum_i u_i(a_i,\theta_i) \\
&= rd(\theta) +\sum_i (1+r\tau_i)u_i(a_i,\theta_i)-\frac{r^2}{2},
\end{align*}
the $V$-optimal choice of $a$ is the same as the $u$-optimal one since $1+r\tau_i>0$ for all $i$ and $r$, and the optimal choice of $r$ is $\mathbb{E}_p[X(a;\theta)]$, as needed.

{\small
\printbibliography

@unpublished{healy2024belief,
	title={Belief elicitation: A user's guide},
	author={Healy, Paul J. and Leo, Greg},
	year={2024},
	note={To appear in the Handbook of Experimental Methodology, ed.\ L. Yariv and E. Snowberg}
}

@article{enke2023cognitive,
  title={Cognitive uncertainty},
  author={Enke, Benjamin and Graeber, Thomas},
  journal={Quarterly Journal of Economics},
  volume={138},
  number={4},
  pages={2021--2067},
  year={2023},
  publisher={Oxford University Press}
}

@article{hossain2013binarized,
  title={The binarized scoring rule},
  author={Hossain, Tanjim and Okui, Ryo},
  journal={Review of Economic Studies},
  volume={80},
  number={3},
  pages={984--1001},
  year={2013},
  publisher={Oxford University Press}
}

@unpublished{amelio2022cognitive,
  title={Cognitive uncertainty and overconfidence},
  author={Amelio, Andrea},
  year={2022},
  note={ECONtribute Discussion Paper No. 173}
}

@article{danz2022belief,
  title={Belief elicitation and behavioral incentive compatibility},
  author={Danz, David and Vesterlund, Lise and Wilson, Alistair J},
  journal={American Economic Review},
  volume={112},
  number={9},
  pages={2851--2883},
  year={2022},
  publisher={American Economic Association 2014 Broadway, Suite 305, Nashville, TN 37203}
}

@inproceedings{lambert2008eliciting,
  title={Eliciting properties of probability distributions},
  author={Lambert, Nicolas and Pennock, David M and Shoham, Yoav},
  booktitle={Proceedings of the 9th ACM Conference on Electronic Commerce},
  pages={129--138},
  year={2008}
}

@unpublished{lambert2011elicitation,
  title={Elicitation and evaluation of statistical forecasts},
  author={Lambert, Nicolas},
  year={2019},
note={Working paper}
}

@article{arts2024measuring,
  title={Measuring decision confidence},
  author={Arts, Sara and Ong, Qiyan and Qiu, Jianying},
  journal={Experimental Economics},
  pages={1--22},
  year={2024},
  publisher={Springer}
}

@article{azrieli2018incentives,
  title={Incentives in experiments: A theoretical analysis},
  author={Azrieli, Yaron and Chambers, Christopher P and Healy, Paul J.},
  journal={Journal of Political Economy},
  volume={126},
  number={4},
  pages={1472--1503},
  year={2018},
  publisher={University of Chicago Press Chicago, IL}
}

@article{chambers2018dynamic,
  title={Dynamic belief elicitation},
  author={Chambers, Christopher P and Lambert, Nicolas},
  journal={Econometrica},
  volume={89},
  number={1},
  pages={375--414},
  year={2021},
  publisher={The Econometric Society}
}

@article{schlag2015penny,
  title={A penny for your thoughts: A survey of methods for eliciting beliefs},
  author={Schlag, Karl H and Tremewan, James and Van der Weele, Jo{\"e}l J},
  journal={Experimental Economics},
  volume={18},
  pages={457--490},
  year={2015},
  publisher={Springer}
}

@unpublished{hu2023confidence,
  title={Confidence in Inference},
  author={Hu, En Hua},
  year={2023},
  note={Working paper}
}

@article{haaland2023designing,
  title={Designing information provision experiments},
  author={Haaland, Ingar and Roth, Christopher and Wohlfart, Johannes},
  journal={Journal of Economic Literature},
  volume={61},
  number={1},
  pages={3--40},
  year={2023},
  publisher={American Economic Association 2014 Broadway, Suite 305, Nashville, TN 37203-2425}
}

@article{xiang2021confidence,
  title={Confidence and central tendency in perceptual judgment},
  author={Xiang, Yang and Graeber, Thomas and Enke, Benjamin and Gershman, Samuel J},
  journal={Attention, Perception, \& Psychophysics},
  volume={83},
  pages={3024--3034},
  year={2021},
  publisher={Springer}
}

@article{charness2021experimental,
  title={Experimental methods: Eliciting beliefs},
  author={Charness, Gary and Gneezy, Uri and Rasocha, Vlastimil},
  journal={Journal of Economic Behavior \& Organization},
  volume={189},
  pages={234--256},
  year={2021},
  publisher={Elsevier}
}

@article{butler2007imprecision,
  title={Imprecision as an account of the preference reversal phenomenon},
  author={Butler, David J and Loomes, Graham C},
  journal={American Economic Review},
  volume={97},
  number={1},
  pages={277--297},
  year={2007},
  publisher={American Economic Association}
}

@article{mobius2022managingself,
  title={Managing self-confidence: theory and experimental evidence},
  author={Möbius, Markus M. and Niederle, Muriel and Niehaus, Paul and Rosenblat, Tanya S.},
  journal={Management Science},
  volume={68},
  number={11},
  pages={7793–-7817},
  year={2022}
}

@article{blanco2010belief,
  title={Belief elicitation in experiments: Is there a hedging problem?},
  author={Blanco, Mariana and Engelmann, Dirk and Koch, Alexander K and Normann, Hans-Theo},
  journal={Experimental Economics},
  volume={13},
  pages={412--438},
  year={2010},
  publisher={Springer}
}

@article{charness2016experimental,
  title={Experimental methods: Pay one or pay all},
  author={Charness, Gary and Gneezy, Uri and Halladay, Brianna},
  journal={Journal of Economic Behavior \& Organization},
  volume={131},
  pages={141--150},
  year={2016},
  publisher={Elsevier}
}

@article{morris2004best,
  title={Best response equivalence},
  author={Morris, Stephen and Ui, Takashi},
  journal={Games and Economic Behavior},
  volume={49},
  number={2},
  pages={260--287},
  year={2004},
  publisher={Elsevier}
}

@article{SchotterTrevino2014,
   author = "Schotter, Andrew and Trevino, Isabel",
   title = "Belief Elicitation in the Laboratory", 
   journal= "Annual Review of Economics",
   year = "2014",
   volume = "6",
   pages = "103-128",
   doi = "https://doi.org/10.1146/annurev-economics-080213-040927",
   url = "https://www.annualreviews.org/content/journals/10.1146/annurev-economics-080213-040927",
   publisher = "Annual Reviews",
   issn = "1941-1391",
   type = "Journal Article",
  }

@article{DizonRoss19,
Author = {Dizon-Ross, Rebecca},
Title = {Parents' Beliefs about Their Children's Academic Ability: Implications for Educational Investments},
Journal = {American Economic Review},
Volume = {109},
Number = {8},
Year = {2019},
Month = {August},
Pages = {2728–65},
DOI = {10.1257/aer.20171172},
URL = {https://www.aeaweb.org/articles?id=10.1257/aer.20171172}}

@article{ERTAC2011532,
title = {Does self-relevance affect information processing? Experimental evidence on the response to performance and non-performance feedback},
journal = {Journal of Economic Behavior \& Organization},
volume = {80},
number = {3},
pages = {532-545},
year = {2011},
issn = {0167-2681},
doi = {https://doi.org/10.1016/j.jebo.2011.05.012},
url = {https://www.sciencedirect.com/science/article/pii/S0167268111001326},
author = {Seda Ertac}
}

@article{EilRao11,
 ISSN = {19457669, 19457685},
 URL = {http://www.jstor.org/stable/41237187},
 author = {David Eil and Justin M. Rao},
 journal = {American Economic Journal: Microeconomics},
 number = {2},
 pages = {114--138},
 publisher = {American Economic Association},
 title = {The Good News-Bad News Effect: Asymmetric Processing of Objective Information about Yourself},
 urldate = {2025-09-24},
 volume = {3},
 year = {2011}
}

@article{Zimmerman20,
 ISSN = {00028282, 19447981},
 URL = {https://www.jstor.org/stable/26875123},
 author = {Florian Zimmermann},
 journal = {American Economic Review},
 number = {2},
 pages = {pp. 337--363},
 publisher = {American Economic Association},
 title = {The Dynamics of Motivated Beliefs},
 urldate = {2025-09-24},
 volume = {110},
 year = {2020}
}

@article{AoyagiFréchetteYuksel24,
Author = {Aoyagi, Masaki and Fréchette, Guillaume R. and Yuksel, Sevgi},
Title = {Beliefs in Repeated Games: An Experiment},
Journal = {American Economic Review},
Volume = {114},
Number = {12},
Year = {2024},
Month = {December},
Pages = {3944–75},
DOI = {10.1257/aer.20220639},
URL = {https://www.aeaweb.org/articles?id=10.1257/aer.20220639}}

@article{ExleyNielsen24,
Author = {Exley, Christine L. and Nielsen, Kirby},
Title = {The Gender Gap in Confidence: Expected but Not Accounted For},
Journal = {American Economic Review},
Volume = {114},
Number = {3},
Year = {2024},
Month = {March},
Pages = {851–85},
DOI = {10.1257/aer.20221413},
URL = {https://www.aeaweb.org/articles?id=10.1257/aer.20221413}}

@article{FehrPowellWilkening21,
Author = {Fehr, Ernst and Powell, Michael and Wilkening, Tom},
Title = {Behavioral Constraints on the Design of Subgame-Perfect Implementation Mechanisms},
Journal = {American Economic Review},
Volume = {111},
Number = {4},
Year = {2021},
Month = {April},
Pages = {1055–91},
DOI = {10.1257/aer.20170297},
URL = {https://www.aeaweb.org/articles?id=10.1257/aer.20170297}}

@article{BDM,
author = {Becker, Gordon M. and Degroot, Morris H. and Marschak, Jacob},
title = {Measuring utility by a single-response sequential method},
journal = {Behavioral Science},
volume = {9},
number = {3},
pages = {226-232},
doi = {https://doi.org/10.1002/bs.3830090304},
url = {https://onlinelibrary.wiley.com/doi/abs/10.1002/bs.3830090304},
eprint = {https://onlinelibrary.wiley.com/doi/pdf/10.1002/bs.3830090304},
year = {1964}
}

@article{BelzilMaurelSidibe21,
author = {Belzil, Christian and Maurel, Arnaud and Sidib\'{e}, Modibo},
title = {Estimating the Value of Higher Education Financial Aid: Evidence from a Field Experiment},
journal = {Journal of Labor Economics},
volume = {39},
number = {2},
pages = {361-395},
year = {2021},
doi = {10.1086/710701},
URL = {https://doi.org/10.1086/710701},
eprint = {https://doi.org/10.1086/710701}
}

@article{ChassangPadróIMiquelSnowberg12,
Author = {Chassang, Sylvain and Padró I Miquel, Gerard and Snowberg, Erik},
Title = {Selective Trials: A Principal-Agent Approach to Randomized Controlled Experiments},
Journal = {American Economic Review},
Volume = {102},
Number = {4},
Year = {2012},
Month = {June},
Pages = {1279–1309},
DOI = {10.1257/aer.102.4.1279},
URL = {https://www.aeaweb.org/articles?id=10.1257/aer.102.4.1279}}

@article{Dupas14,
author = {Dupas, Pascaline},
title = {Short-Run Subsidies and Long-Run Adoption of New Health Products: Evidence From a Field Experiment},
journal = {Econometrica},
volume = {82},
number = {1},
pages = {197-228},
doi = {https://doi.org/10.3982/ECTA9508},
url = {https://onlinelibrary.wiley.com/doi/abs/10.3982/ECTA9508},
eprint = {https://onlinelibrary.wiley.com/doi/pdf/10.3982/ECTA9508},
year = {2014}
}

@article{Serra-GarciaGneezy21,
Author = {Serra-Garcia, Marta and Gneezy, Uri},
Title = {Mistakes, Overconfidence, and the Effect of Sharing on Detecting Lies},
Journal = {American Economic Review},
Volume = {111},
Number = {10},
Year = {2021},
Month = {October},
Pages = {3160–83},
DOI = {10.1257/aer.20191295},
URL = {https://www.aeaweb.org/articles?id=10.1257/aer.20191295}}

@article{blavatskyy_betting_2009,
	title = {Betting on own knowledge: Experimental test of overconfidence},
	volume = {38},
	issn = {1573-0476},
	url = {https://doi.org/10.1007/s11166-008-9048-7},
	doi = {10.1007/s11166-008-9048-7},
	pages = {39--49},
	number = {1},
	journaltitle = {Journal of Risk and Uncertainty},
	shortjournal = {Journal of Risk and Uncertainty},
	author = {Blavatskyy, Pavlo R.},
	date = {2009-02-01},
}

@article{abdellaoui2024unpacking,
  title={Unpacking overconfident behavior when betting on oneself},
  author={Abdellaoui, Mohammed and Bleichrodt, Han and Gutierrez, C{\'e}dric},
  journal={Management Science},
  volume={70},
  number={10},
  pages={7042--7061},
  year={2024},
  publisher={INFORMS}
}

@incollection{hoffrage2016overconfidence,
  title={Overconfidence},
  author={Hoffrage, Ulrich},
  booktitle={Cognitive Illusions},
  pages={291--314},
  year={2016},
  publisher={Psychology Press}
}

@article{liberman2004local,
  title={Local and global judgments of confidence.},
  author={Liberman, Varda},
  journal={Journal of Experimental Psychology: Learning, Memory, and Cognition},
  volume={30},
  number={3},
  pages={729},
  year={2004},
  publisher={American Psychological Association}
}

@article{hoff2006discrimination,
  title={Discrimination, social identity, and durable inequalities},
  author={Hoff, Karla and Pandey, Priyanka},
  journal={American Economic Review},
  volume={96},
  number={2},
  pages={206--211},
  year={2006},
  publisher={American Economic Association}
}

@article{AshrafBerryShapiro10,
Author = {Ashraf, Nava and Berry, James and Shapiro, Jesse M.},
Title = {Can Higher Prices Stimulate Product Use? Evidence from a Field Experiment in Zambia},
Journal = {American Economic Review},
Volume = {100},
Number = {5},
Year = {2010},
Month = {December},
Pages = {2383–2413},
DOI = {10.1257/aer.100.5.2383},
URL = {https://www.aeaweb.org/articles?id=10.1257/aer.100.5.2383}}

@article{rutstrom_stated_2009,
	title = {Stated beliefs versus inferred beliefs: A methodological inquiry and experimental test},
	volume = {67},
	issn = {0899-8256},
	url = {https://www.sciencedirect.com/science/article/pii/S0899825609000591},
	doi = {https://doi.org/10.1016/j.geb.2009.04.001},
	pages = {616--632},
	number = {2},
	journaltitle = {Games and Economic Behavior},
	author = {Rutström, E. Elisabet and Wilcox, Nathaniel T.},
	date = {2009},
}

@article{NyarkoSchotter02,
author = {Nyarko, Yaw and Schotter, Andrew},
title = {An Experimental Study of Belief Learning Using Elicited Beliefs},
journal = {Econometrica},
volume = {70},
number = {3},
pages = {971-1005},
doi = {https://doi.org/10.1111/1468-0262.00316},
url = {https://onlinelibrary.wiley.com/doi/abs/10.1111/1468-0262.00316},
eprint = {https://onlinelibrary.wiley.com/doi/pdf/10.1111/1468-0262.00316},
year = {2002}
}

@article{smith1961consistency,
  title={Consistency in statistical inference and decision},
  author={Smith, Cedric AB},
  journal={Journal of the Royal Statistical Society: Series B (Methodological)},
  volume={23},
  number={1},
  pages={1--25},
  year={1961},
  publisher={Wiley Online Library}
}

@article{gillen2019experimenting,
  title={Experimenting with measurement error: Techniques with applications to the Caltech cohort study},
  author={Gillen, Ben and Snowberg, Erik and Yariv, Leeat},
  journal={Journal of Political Economy},
  volume={127},
  number={4},
  pages={1826--1863},
  year={2019},
  publisher={The University of Chicago Press Chicago, IL}
}

@article{burks2013overconfidence,
  title={Overconfidence and social signalling},
  author={Burks, Stephen V and Carpenter, Jeffrey P and Goette, Lorenz and Rustichini, Aldo},
  journal={Review of Economic Studies},
  volume={80},
  number={3},
  pages={949--983},
  year={2013},
  publisher={Oxford University Press}
}

@article{benoit2015does,
  title={Does the better-than-average effect show that people are overconfident?: Two experiments},
  author={Beno\^{i}t, Jean-Pierre and Dubra, Juan and Moore, Don A},
  journal={Journal of the European Economic Association},
  volume={13},
  number={2},
  pages={293--329},
  year={2015},
  publisher={Oxford University Press}
}

@article{clark2009overconfidence,
  title={Overconfidence in forecasts of own performance: An experimental study},
  author={Clark, Jeremy and Friesen, Lana},
  journal={The Economic Journal},
  volume={119},
  number={534},
  pages={229--251},
  year={2009},
  publisher={Oxford University Press Oxford, UK}
}

@article{ortoleva2015overconfidence,
  title={Overconfidence in political behavior},
  author={Ortoleva, Pietro and Snowberg, Erik},
  journal={American Economic Review},
  volume={105},
  number={2},
  pages={504--535},
  year={2015},
  publisher={American Economic Association 2014 Broadway, Suite 305, Nashville, TN 37203}
}

@unpublished{de2024caution,
  title={Caution in the face of complexity},
  author={de Clippel, Geoffroy and Moscariello, Paola and Ortoleva, Pietro and Rozen, Kareen},
  year={2024},
  note={Working paper}
}

@article{niederle2007women,
  title={Do women shy away from competition? Do men compete too much?},
  author={Niederle, Muriel and Vesterlund, Lise},
  journal={Quarterly Journal of Economics},
  volume={122},
  number={3},
  pages={1067--1101},
  year={2007},
  publisher={MIT Press}
}

@article{coffman2014evidence,
  title={Evidence on self-stereotyping and the contribution of ideas},
  author={Coffman, Katherine Baldiga},
  journal={Quarterly Journal of Economics},
  volume={129},
  number={4},
  pages={1625--1660},
  year={2014},
  publisher={MIT Press}
}

@article{bordalo2019beliefs,
  title={Beliefs about gender},
  author={Bordalo, Pedro and Coffman, Katherine and Gennaioli, Nicola and Shleifer, Andrei},
  journal={American Economic Review},
  volume={109},
  number={3},
  pages={739--773},
  year={2019},
  publisher={American Economic Association 2014 Broadway, Suite 305, Nashville, TN 37203}
}

@article{moore2008trouble,
  title={The trouble with overconfidence.},
  author={Moore, Don A and Healy, Paul J},
  journal={Psychological Review},
  volume={115},
  number={2},
  pages={502},
  year={2008},
  publisher={American Psychological Association}
}

@article{trautmann2015belief,
  title={Belief elicitation: A horse race among truth serums},
  author={Trautmann, Stefan T and van de Kuilen, Gijs},
  journal={The Economic Journal},
  volume={125},
  number={589},
  pages={2116--2135},
  year={2015},
  publisher={Oxford University Press Oxford, UK}
}

@article{bellemare2008measuring,
  title={Measuring inequity aversion in a heterogeneous population using experimental decisions and subjective probabilities},
  author={Bellemare, Charles and Kr{\"o}ger, Sabine and Van Soest, Arthur},
  journal={Econometrica},
  volume={76},
  number={4},
  pages={815--839},
  year={2008},
  publisher={Wiley Online Library}
}

@article{chen2026ask,
  title={How to Ask for Belief Statistics without Distortion?},
  author={Chen, Yi-Chun and Wang, Ruoyu and Zhang, Xinhan},
  journal={arXiv preprint arXiv:2602.10474},
  year={2026}
}

@InProceedings{Frongillo15,
  title = 	 {Vector-Valued Property Elicitation},
  author = 	 {Frongillo, Rafael and Kash, Ian A.},
  booktitle = 	 {Proceedings of The 28th Conference on Learning Theory},
  pages = 	 {710--727},
  year = 	 {2015},
  editor = 	 {Grünwald, Peter and Hazan, Elad and Kale, Satyen},
  volume = 	 {40},
  series = 	 {Proceedings of Machine Learning Research},
  address = 	 {Paris, France},
  month = 	 {03--06 Jul},
  publisher =    {PMLR},
  url = 	 {https://proceedings.mlr.press/v40/Frongillo15.html}
}

@unpublished{halevy2025magic,
  title={Magic Mirror on the Wall, who is the Smartest One of All?},
  author={Halevy, Yoram and Hoelzemann, Johannes C and Kneeland, Terri},
  year={2025},
  note={Working paper}
}
}

\end{document}